\documentclass[final,3p,authoryear]{elsarticle} 

\usepackage{amsthm}
\usepackage{amssymb}
\usepackage{mathtools} 
\usepackage{rotating}
\usepackage[latin1]{inputenc}
\usepackage{graphicx}
\usepackage{url}
\usepackage{amsmath}
\usepackage{amsfonts}
\usepackage{url}
\usepackage{color}
\usepackage[table]{xcolor}
\usepackage{subfigure}
\usepackage{slashbox}
\usepackage{ctable}
\usepackage{floatrow}
\usepackage{tikz}
\usepackage{xcolor}
\usepackage{pifont}
\usepackage{todonotes} 
\usepackage{enumitem}
\newfloatcommand{capbtabbox}{table}[][\FBwidth]

\usepackage{blindtext}

\newcommand{\real}{I\hspace{-1.3mm}R}

\def\bx{\boldsymbol{x}}

\def\bv{\boldsymbol{v}}
\def\bz{\boldsymbol{z}}
\def\bX{\boldsymbol{X}}

\def\bmu{\boldsymbol{\mu}}

\def\bSigma{\boldsymbol{\Sigma}}
\def\bGamma{\boldsymbol{\Gamma}}
\def\bDelta{\boldsymbol{\Delta}}

\def\bpsi{\boldsymbol{\psi}}
\def\bvartheta{\boldsymbol{\vartheta}}

\newtheorem{pro}{Proposition} 
\renewenvironment{proof}{\textit{Proof. }}{\qed}

\journal{arXiv}

\begin{document}

\begin{frontmatter}

\title{Parsimonious mixtures of multivariate contaminated normal distributions}

\author[CT]{Antonio Punzo\corref{cor}} 
\ead{antonio.punzo@unict.it}
\author[CA]{Paul D.~McNicholas} 
\ead{mcnicholas@math.mcmaster.ca}
\cortext[cor]{Corresponding author: 
Email: \texttt{antonio.punzo@unict.it}, 
Phone: +39-095-7537640, 
Fax: +39-095-7537610}
\address[CT]{
Department of Economics and Business, University of Catania, Catania, Italy.
							}
\address[CA]{
Department of Mathematics \& Statistics, McMaster University, Hamilton, Canada.
}

\begin{abstract}
\small
A mixture of multivariate contaminated normal distributions is developed for model-based clustering. 
In addition to the parameters of the classical normal mixture, our contaminated mixture has, for each cluster, a parameter controlling the proportion of mild outliers and one specifying the degree of contamination. 
Crucially, these parameters do not have to be specified \textit{a~priori}, adding a flexibility to our approach.
Parsimony is introduced via eigen-decomposition of the component covariance matrices, and sufficient conditions for the identifiability of all the members of the resulting family are provided. 
An expectation-conditional maximization algorithm is outlined for parameter estimation and various implementation issues are discussed. 
Using a large scale simulation study, the behaviour of the proposed approach is investigated and comparison with well-established finite mixtures is provided.
The performance of this novel family of models is also illustrated on artificial and real data.
\end{abstract}

\begin{keyword}
Contaminated normal distribution \sep contamination \sep EM algorithm \sep mixture models \sep model-based clustering.
\end{keyword}
\end{frontmatter}



\section{Introduction}
\label{sec:Introduction}

Mixtures of multivariate normal distributions have been extensively considered as a powerful device for clustering by typically assuming, as we do, that each mixture component represents a cluster (or group or class; cf.\ \citealp{McLa:Basf:mixt:1988}, \citealp{Fral:Raft:Howm:1998}, \citealp{Bohn:Comp:2000}, and \citealp{mcnicholas16}).
Their popularity is largely attributable to computational and theoretical convenience, as well as the speed with which these mixtures can be implemented for many data sets.
The volume of published work on normal mixtures has increased significantly since the work of \citet{Banf:Raft:mode:1993} and \citet{Cele:Gova:Gaus:1995}; the latter work completes the former by introducing a family of fourteen mixtures of multivariate normal distributions obtained by imposing some constraints on eigen-decomposed component covariance matrices. 

Unfortunately, real data are often contaminated by outliers that affect the estimation of the component means and covariance matrices (see, e.g., \citealp{Barn:Lewi:Outl:1994}, \citealp{Beck:Gath:Them:1999}, \citealp{Bock:Clus:2002}, and \citealp{Gall:Ritt:Trim:2009}).
Thus, outlier detection and the development of robust methods of parameter estimation insensitive to their presence are important problems (see \citealp{Garc:Gord:Robu:1999} and \citealp{Henn:Brea:2004}).
Outliers are observations that deviate from the (posited) reference model (\citealp{Agga:Outl:2013} and \citealp{Hawk:Iden:2013}), which is here assumed to be a mixture of multivariate normal distributions; for a discussion about the concept of reference model, see \citet{Davi:Gath:Thei:1993} and \citet{Henn:Fixe:2002}.
Outliers can be roughly distinguished into two types \citep[cf.][pp.~79--80]{Ritt:Robu:2015}:
\begin{description}
	\item[Mild] outliers are sampled from some population different or even far from the assumed model. 
	Such outliers generally reflect the difficulty of the model specification problem.
	In their presence, the statistician is recommended to choose a model flexible enough to accommodate all data points, including the outliers. 
	\item[Gross] outliers cannot be modelled by a distribution. 
They are unpredictable and incalculable.
In the presence of gross outliers, the statistician is recommended to choose a method for suppressing them.
A classical choice is trimming, introduced to cluster analysis by \citet{Cues:Gord:Matr:Trim:1997} and followed, only to cite a few, by \citet{Gall:Ritt:Arob:2005,Gall:Ritt:Trim:2009}, \citet{Garc:Gord:Matr:Trim:2003,Garc:Gord:Matr:Mayo:Agen:2008,Garc:Gord:Matr:Mayo:Arew:2010}, and \citet{Ruwe:Garc:Gord:Mayo:Onth:2013}.
The underlying idea is to decompose the data set into high-density regions, with the remainder consisting mainly of isolated observations \citep{Hart:Stat:1985}.
Of course, for this method to be effective the estimator must recognize the outliers automatically.
\end{description}
Mild outliers --- also referred to as ``bad'' points herein, following \citet{Aitk:Wils:Mixt:1980} --- are the focus of this paper.
For normal mixtures, they are often dealt using two approaches \citep[cf.][]{Ruwe:Garc:Gord:Mayo:TheI:2012}.
In the ``additional component'' approach, protection against outliers is obtained by adding a further convenient component distribution to the mixture of normal distributions to capture outliers.
The first and most famous example in this direction is represented by the addition of a uniform component on the convex hull of the data, as suggested by \citet{Banf:Raft:mode:1993}; see also \citet{Henn:Brea:2004} and \citet{Core:Henn:2011} for the univariate case.
In the ``componentwise'' approach, the component multivariate normal distributions are separately protected against outliers by using either convenient robust estimates of the means and covariance matrices (see \citealp{Camp:Mixt:1984}, \citealp[][Section~2.8]{McLa:Basf:mixt:1988}, \citealp{deVa:Krie:1990}, and \citealp{Mark:Mixt:2000}) or, more often, by embedding them in more general heavy-tailed, usually elliptically symmetric, multivariate distributions.
The classical example is the mixture of multivariate $t$ distributions, which was first used for clustering by \citet{McLa:Peel:Robu:1998} and \citet{Peel:McLa:Robu:2000}. 
Note that the multivariate $t$ distribution can be written as a normal scale mixture, where the mixing weight is a gamma random variable; in fact, the multivariate normal distribution is a limiting case of the multivariate $t$ distribution and the $t$ distribution can be viewed as a generalization of the normal distribution \citep[cf.][]{Peel:McLa:Robu:2000}. 
A further example is given by \citet{Brow:McNi:Spar:Mode:2012}; they introduce a mixture model whereby each mixture component is itself a mixture of a normal and a uniform distribution.
To have an idea of the data configurations where the ``additional component'' approach outperforms the ``componentwise'' approach, and \textit{vice versa}, see the extensive simulation study reported by \citet{Core:Henn:Robu:2015}.
Roughly speaking, the additional component approach is not expected to work well when the mild outliers are either cluster-dependent \citep{Gero:Niko:Lika:Them:2009} or cannot be modeled adequately by the additional component \citep[cf.][p.~233]{McLa:Peel:fini:2000}.

By considering the ``componentwise'' approach, a mixture of multivariate contaminated normal distributions is proposed in Section~\ref{subsec:The general model}. 
A multivariate contaminated normal distribution, which dates back to the seminal work of \citet{Tuke:Asur:1960}, is a two-component normal mixture in which one of the components, with a large prior probability, represents the good observations (reference cluster distribution), and the other, with a small prior probability, the same mean, and an inflated covariance matrix, represents the bad observations \citep[see also][]{Aitk:Wils:Mixt:1980}. 
It represents a common and simple theoretical model for the occurrence of bad points although, by construction, it cannot accommodate asymmetric contamination and/or ``groups'' of concentrated outliers.
Furthermore, parsimonious variants of the proposed model are introduced, in the fashion of \citet{Banf:Raft:mode:1993} and \citet{Cele:Gova:Gaus:1995}, by imposing constraints on eigen-decomposed component covariance matrices (Section~\ref{subsec:Parsimonious variants of the general model}).
The model-based clustering framework is outlined (Section~\ref{subsec:modelling frame}), and sufficient conditions for identifiability of our models are given (Section~\ref{sec:Identifiability}). 
An expectation-conditional maximization (ECM) algorithm for parameter estimation is outlined in Section~\ref{sec:Maximum likelihood estimation}. 
Further computational and operational aspects are discussed in Section~\ref{sec:fas}. 
Advantageously, as it will be better explained in Section~\ref{subsec:Automatic detection of noise}, once a mixture of multivariate contaminated normal distributions is fitted to the observed data, by means of maximum \textit{a~posteriori} probabilities, each observation can be first assigned to one of the clusters and then classified as good or bad.
Moreover, as detailed in Section~\ref{subsec:Some notes on robustness}, bad points are automatically down-weighted in the estimation of the component means and covariance matrices.
Thus, we have a model for simultaneous robust clustering and detection of mild outliers.
Furthermore, the fact that all of the parameters can be estimated by maximum likelihood (see Section~\ref{sec:Maximum likelihood estimation}), and automatic criteria, such as the BIC (see Section~\ref{subsec:BIC}), can be adopted to select the number of clusters and the parsimonious covariance structure (see Section~\ref{subsec:Parsimonious variants of the general model}), implies that there is no need to preliminary visualize the data to try to understand what the outliers could be.
This is the reason why our approach could be extended to higher dimensions where the visualization of the data becomes cumbersome.
In Section~\ref{sec:Numerical studies}, the behavior of the proposed model, in comparison with some of the approaches discussed above, is investigated through a large-scale simulation study.
Applications on artificial and real data are presented in Section~\ref{sec:Data analysis}.
The paper concludes with some discussion in Section~\ref{sec:Discussion and future work}.

\section{Methodology}
\label{sec:Methodology}

\subsection{The general model}
\label{subsec:The general model}

The distribution of a random vector $\bX$, taking values on $\real^p$, according to a parametric finite mixture model, can be written as
\begin{equation}
p\left(\bx;\boldsymbol{\psi}\right)=\sum_{g=1}^G\pi_gf\left(\bx;\bvartheta_g\right),
\label{eq:classical mixture}
\end{equation}
where $\pi_g$ is the mixing proportion for the $g$th component, with $\pi_g>0$ and $\sum_{g=1}^G\pi_g=1$, $f\left(\bx;\bvartheta_g\right)$ is the density of the $g$th component with parameters $\bvartheta_g$, and $\boldsymbol{\psi}=\left\{\boldsymbol{\pi},\bvartheta\right\}$, with $\boldsymbol{\pi}=\left\{\pi_g\right\}_{g=1}^G$ and $\bvartheta=\left\{\bvartheta_g\right\}_{g=1}^G$, contains all of the parameters of the mixture. 

In this paper, for the $g$th mixture component, $g=1,\ldots,G$, we adopt the multivariate contaminated normal distribution 
\begin{equation}
f\left(\bx;\bvartheta_g\right)=\alpha_g\phi\left(\bx;\bmu_g,\bSigma_g\right)+\left(1-\alpha_g\right)\phi\left(\bx;\bmu_g,\eta_g\bSigma_g\right),
\label{eq:contaminated normal distribution}
\end{equation}
where $\alpha_g\in\left(0.5,1\right)$, $\eta_g>1$, $\bvartheta_g=\left\{\alpha_g,\bmu_g, \bSigma_g,\eta_g\right\}$, and
\begin{equation}
\phi\left(\bx;\bmu,\bSigma\right)=\left(2\pi\right)^{-\frac{p}{2}}\left|\bSigma\right|^{-\frac{1}{2}}\exp\left\{-\frac{1}{2}\delta\left(\bx,\bmu;\bSigma\right)\right\}
\label{eq:Normal distribution}
\end{equation}
is the distribution of a $p$-variate normal random vector with mean $\bmu$ and covariance matrix $\bSigma$.
In \eqref{eq:Normal distribution}, $\delta\left(\bx,\bmu;\bSigma\right)=\left(\bx-\bmu\right)'\bSigma^{-1}\left(\bx-\bmu\right)$ denotes the squared Mahalanobis distance while $\left|\cdot\right|$ is the determinant. 
Note that $\alpha_g$ is constrained to be greater than 0.5 because, in robust statistics, it is usually assumed that at least half of the points are good; however, $\alpha_g\in\left(0,1\right)$ is acceptable in general, as often happens in the literature.
In \eqref{eq:contaminated normal distribution}, $\eta_g$ denotes the degree of contamination, and because of the assumption $\eta_g>1$, it can be interpreted as the increase in variability due to the bad observations (i.e., it is an inflation parameter; see Figure~\ref{fig:contaminated}). 
Indeed, the covariance matrix in the $g$th component, $g=1,\ldots,G$, is given by 
\begin{equation}
\left[\alpha_g+\left(1-\alpha_g\right)\eta_g\right]\bSigma_g,
\label{eq:component covariance matrix}
\end{equation}
where the scale factor satisfies the constraint $\left[\alpha_g+\left(1-\alpha_g\right)\eta_g\right]>1$ because $\eta_g>1$.
\begin{figure}[!ht]
  \centering
\includegraphics[width=0.49\textwidth]{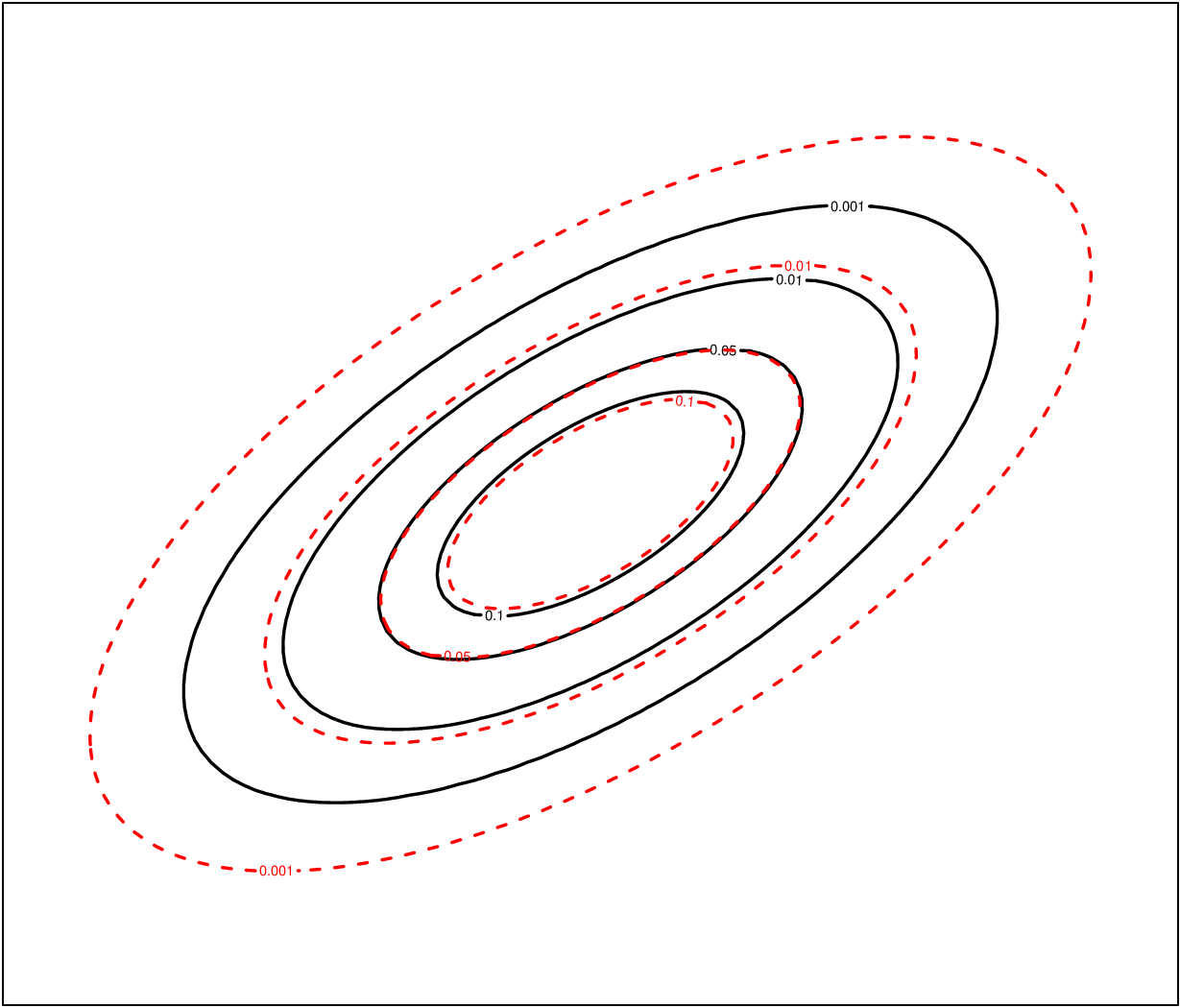}
\caption{Example of contours illustrating the inflation effect of the constraint that $\eta>1$.
A normal distribution (solid black contours) is compared to a contaminated normal distribution (dashed red contours) with $\alpha=0.8$ and $\eta=3$.
}
\label{fig:contaminated}       
\end{figure}
The density of our mixture of multivariate contaminated normal distributions is given by
\begin{equation}
p\left(\bx;\boldsymbol{\psi}\right)=\sum_{g=1}^G\pi_g\left[\alpha_g\phi\left(\bx;\bmu_g,\bSigma_g\right)+\left(1-\alpha_g\right)\phi\left(\bx;\bmu_g,\eta_g\bSigma_g\right)\right].
\label{eq:mixture of contaminated normal distributions}
\end{equation}
Because our contaminated approach contains $2G$ components, i.e., $G$ top-level components, each of which contains two second-level components, the model in \eqref{eq:classical mixture} shall be considered to contain $G$ clusters, rather than $G$ components, hereafter.
Based on this consideration, model~\eqref{eq:mixture of contaminated normal distributions} can be seen as a special case of the multi-layer mixture of normal distributions of \citet{Li:Clus:2005} if each of the $G$ clusters at the top level is itself a mixture of two components, with equal means and proportional covariance matrices at the secondary layer.

\subsection{Parsimonious variants of the general model}
\label{subsec:Parsimonious variants of the general model}

Because there are $p\left(p + 1\right)/2$ free parameters for each $\bSigma_g$, it is usually necessary to introduce parsimony into the model in \eqref{eq:mixture of contaminated normal distributions}. 
Following \citet{Banf:Raft:mode:1993} and \citet{Cele:Gova:Gaus:1995}, we consider the eigen-decomposition 
\begin{equation}
\bSigma_g=\lambda_g\bGamma_g\boldsymbol{\Delta}_g\bGamma_g',
\label{eq:eigenvalue decomposition}
\end{equation}
where $\lambda_g=\left|\boldsymbol{\Sigma}_g\right|^{1/p}$, $\boldsymbol{\Delta}_g$ is the scaled ($\left|\boldsymbol{\Delta}_g\right|=1$) diagonal matrix of the eigenvalues of $\bSigma_g$ sorted in decreasing order, and $\bGamma_g$ is a $p\times p$ orthogonal matrix whose columns are the normalized eigenvectors of $\bSigma_g$, ordered according to their eigenvalues.
Each element in the right-hand side of \eqref{eq:eigenvalue decomposition} has a different geometric interpretation: $\lambda_g$ determines the volume of the $g$th cluster of the good data only, $\bDelta_g$ determines the shape of the cluster, and $\bGamma_g$ determines the orientation of the cluster.
Based on \eqref{eq:component covariance matrix}, the volume of the cluster is given by $\lambda_g\left[\alpha_g+\left(1-\alpha_g\right)\eta_g\right]$. 

In the fashion of \citet{Cele:Gova:Gaus:1995}, we impose constraints on the three components of \eqref{eq:eigenvalue decomposition} resulting in a family of fourteen parsimonious mixtures of contaminated normal distributions models (\tablename~\ref{tab:models}).
The last column of \tablename~\ref{tab:models} specifies the scale invariant models of this family.
\renewcommand{\arraystretch}{0.4}  
\begin{table}[!ht]
\caption{
\label{tab:models}
Nomenclature, covariance structure, and number of free parameters in $\bSigma_1,\ldots,\bSigma_G$ for the models of our family.}
\centering
\resizebox{\textwidth}{!}{
\begin{tabular}{lllllrrc}
\toprule
Family & Model & Volume & Shape & Orientation & $\bSigma_g$ & \# of free parameters in $\bSigma_1,\ldots,\bSigma_G$ & Scale invariant
\\
\midrule
Spherical 
&EII & Equal    & Spherical & -            & $\lambda \boldsymbol{I}$  & 1 & No\\
&VII & Variable & Spherical & -            & $\lambda_g \boldsymbol{I}$ & $G$ & No\\[2mm]
Diagonal 
&EEI & Equal    & Equal     & Axis-Aligned & $\lambda \boldsymbol{\Delta}$ & $p$ & Yes\\
&VEI & Variable & Equal     & Axis-Aligned & $\lambda_g \boldsymbol{\Delta}$ & $G+p-1$ & Yes\\
&EVI & Equal    & Variable  & Axis-Aligned & $\lambda \boldsymbol{\Delta}_g$ & $1+G\left(p-1\right)$ & Yes\\
&VVI & Variable & Variable  & Axis-Aligned & $\lambda_g \boldsymbol{\Delta}_g$ & $Gp$ & Yes\\[2mm]
General 
&EEE & Equal    & Equal     & Equal        & $\lambda\bGamma\boldsymbol{\Delta}\bGamma'$ & $p\left(p+1\right)/2$ & Yes\\
&VEE & Variable & Equal     & Equal        & $\lambda_g\bGamma\boldsymbol{\Delta}\bGamma'$ & $G+p-1+p\left(p-1\right)/2$  & Yes \\
&EVE & Equal    & Variable  & Equal        & $\lambda\bGamma\boldsymbol{\Delta}_g\bGamma'$  & $1+G\left(p-1\right)+p\left(p-1\right)/2$  & No\\
&EEV & Equal    & Equal     & Variable     & $\lambda\bGamma_g\boldsymbol{\Delta}\bGamma_g'$ &  $p+Gp\left(p-1\right)/2$  & No\\
&VVE & Variable & Variable  & Equal        & $\lambda_g\bGamma\boldsymbol{\Delta}_g\bGamma'$ & $Gp+p\left(p-1\right)/2$  & No\\
&VEV & Variable & Equal     & Variable   & $\lambda_g\bGamma_g\boldsymbol{\Delta}\bGamma_g'$ & $G+p-1+Gp\left(p-1\right)/2$  & No\\
&EVV & Equal    & Variable  & Variable   & $\lambda\bGamma_g\boldsymbol{\Delta}_g\bGamma_g'$ & $1+G\left(p-1\right)+Gp\left(p-1\right)/2$  & Yes\\
&VVV & Variable & Variable  & Variable   & $\lambda_g\bGamma_g\boldsymbol{\Delta}_g\bGamma_g'$ & $Gp\left(p+1\right)/2$  & Yes\\
\bottomrule
\end{tabular}
}
\end{table}

\subsection{Model-based clustering}
\label{subsec:modelling frame}

The idea of defining clustering in terms of the components of a mixture model goes back at least 60 years \citep[cf.][Section~2.1]{mcnicholas16}, and model-based clustering has become increasingly popular since mixture models were first used for clustering \citep{wolfe65}. 
Consider $n$ independent $p$-dimensional unlabeled observations $\left\{\bx_i\right\}_{i=1}^n$ from model \eqref{eq:mixture of contaminated normal distributions}, and let $\left\{\boldsymbol{z}_i\right\}_{i=1}^n$ denote cluster memberships, where $\bz_i=\left(z_{i1},\ldots,z_{iG}\right)'$ and $z_{ig}=1$ if $\bx_i$ belongs to cluster $g$ and $z_{ig}=0$ otherwise. 
Using the same notation as before, the model-based clustering likelihood is given by 
\begin{displaymath}
L\left(\boldsymbol{\psi}\right)=\prod_{i=1}^n p\left(\bx_i;\boldsymbol{\psi}\right),
\end{displaymath}
and the predicted classifications are given by the maximum \textit{a~posteriori} probabilities (MAP).
Note that 
$$
\text{MAP}(\widehat{z}_{ig})=
\begin{cases}
1 & \text{if } \max_h\{\widehat{z}_{ih}\} \text{ occurs in cluster $g$,}\\
0 & \text{otherwise},
\end{cases}
$$
where
\begin{displaymath}
\widehat{z}_{ig}=\frac{\widehat{\pi}_gf(\bx_i;\widehat{\bvartheta}_g)}{p(\bx_i;\widehat{\boldsymbol{\psi}})}=\frac{\widehat{\pi}_gf(\bx_i;\widehat{\bvartheta}_g)}{\displaystyle\sum_{h=1}^G\widehat{\pi}_hf(\bx_i;\widehat{\bvartheta}_h)}
\end{displaymath}
is the \textit{a~posteriori} expected value of $Z_{ig}$ given $\bx_i$, i.e., the probability that $Z_{ig}=1$ given $\bx_i$ and based on the parameter estimates $\widehat{\bpsi}$.

\section{Identifiability}
\label{sec:Identifiability}

Before outlining parameter estimation for the models in our family, it is important to establish their identifiability. 
Identifiability is a necessary requirement, \textit{inter alia}, for the usual asymptotic theory to hold for maximum likelihood estimation of the model parameters (cf.\ Section~\ref{sec:Maximum likelihood estimation}). 
Before investigating the identifiability of our contaminated mixtures, it is convenient to rewrite the model density as
\begin{equation*}
p\left(\bx;\boldsymbol{\psi}\right)=\sum_{g=1}^G\sum_{h=1}^2\pi_g\alpha_{gh}\phi\left(\bx;\bmu_g,\eta_{gh}\bSigma_g\right),
\end{equation*}
where, with respect to equation \eqref{eq:mixture of contaminated normal distributions}, $\alpha_{g1}=\alpha_g$, $\alpha_{g2}=1-\alpha_{g1}$, $\eta_{g1}=1$, and $\eta_{g2}=\eta_g$. 

Identifiability of univariate and multivariate finite mixtures of normal distributions has been proved by \citet{Teic:Iden:1963} and \citet{Yako:Spra:Onth:1968}, respectively. 
As stated by \citet{DiZi:Guar:Rocc:Amix:2007}, in the absence of any constraint, a mixture of mixtures is not identifiable in general; this is essentially due to the possibility of interchanging component labels between the two levels of the model.
In our case, the contaminated normal distribution $f\left(\bx;\bvartheta_g\right)$ in cluster~$g$ is elliptical, and sufficient conditions for identifiability of finite mixtures of elliptical distributions are given in \citet{Holz:Munk:Gnei:Iden:2006}. However, these conditions will only apply here if we fix $\alpha_g$ and $\eta_g$ \textit{a~priori}. 
To avoid the requirement to fix parameters in advance, we need to take a different approach to prove identifiability.

In Proposition~\ref{pro:1}, it will be shown that the most general model in our family (i.e., VVV) is identifiable provided that, given two of the $G$ normal distributions $\phi\left(\bx;\bmu_g,\bSigma_g\right)$ representing the good observations, they have distinct means and/or non-proportional covariance matrices.
It is easy to show that the same sufficient condition also holds for the models EVI, VVI, EVE, EEV, VVE, VEV, and EVV.
Proposition~\ref{pro:2} shows that the VEE model is identifiable provided that, given two of the $G$ normal distributions $\phi\left(\bx;\bmu_g,\bSigma_g\right)$ representing the good observations, they have distinct means.
It is straightforward to show that the same sufficient condition also holds for the nested models: EII, VII, EEI, VEI, and EEE.
\begin{pro}
\label{pro:1}
Let
\begin{displaymath}
p\left(\bx;\boldsymbol{\psi}\right)=\sum_{g=1}^G\sum_{h=1}^2\pi_g\alpha_{gh}\phi\left(\bx;\bmu_g,\eta_{gh}\bSigma_g\right)
\end{displaymath}
and 
\begin{displaymath}
p\left(\bx;\widetilde{\boldsymbol{\psi}}\right)=\sum_{s=1}^{\widetilde{G}}\sum_{t=1}^2\widetilde{\pi}_s\widetilde{\alpha}_{st}\phi(\bx;\widetilde{\bmu}_s,\widetilde{\eta}_{st}\widetilde{\bSigma}_s)
\end{displaymath}
be two different parameterizations of the unconstrained model (i.e., VVV). 
If $g\neq g_1$ implies
\begin{equation}
\left\|\bmu_g-\bmu_{g_1}\right\|_2^2+\left\|\bSigma_g-a\bSigma_{g_1}\right\|_2^2\neq 0
\label{eq:sufficient condition}
\end{equation}
for all $a>0$, where $\left\| \cdot \right\|_2$ is the Froebenius norm, then the equality $p\left(\bx;\boldsymbol{\psi}\right)=p(\bx;\widetilde{\boldsymbol{\psi}})$ implies that $G=\widetilde{G}$ and also implies that there exists a relabelling such that
\begin{displaymath}
\pi_g=\widetilde{\pi}_g,\quad \alpha_{gh}=\widetilde{\alpha}_{gh},\quad \bmu_g=\widetilde{\bmu}_g,\quad \bSigma_g=\widetilde{\bSigma}_g,\quad \text{and}\quad \eta_{gh}=\widetilde{\eta}_{gh}.
\end{displaymath}
\end{pro}
\begin{proof}
The identifiability of finite mixtures of normal distributions guarantees that $2G=2\widetilde{G}$, i.e., $G=\widetilde{G}$, and, for each pair $\left(g,h\right)$, there exists a pair $\left(s,t\right)$ such that 
\begin{equation}
\pi_g\alpha_{gh}=\widetilde{\pi}_s\widetilde{\alpha}_{st},\quad \bmu_g=\widetilde{\bmu}_s,\quad \text{and}\quad \eta_{gh}\bSigma_g=\widetilde{\eta}_{st}\widetilde{\bSigma}_s;
\label{eq:problematic equality} 	
\end{equation}
cf. \citet{DiZi:Guar:Rocc:Amix:2007}.
Note that Condition~\eqref{eq:sufficient condition}, the fact that $\eta_{g2}>\eta_{g1}$ ($\eta_{s2}>\eta_{s1}$), and the positivity of all the weights $\pi_g$ and $\alpha_{gh}$ ($\pi_s$ and $\alpha_{st}$) avoids nonidentifiability due to potential overfitting \citep[a potential problem for identifiability first noted by][]{Craw:Anap:1994}.
In particular, the positivity constraint on the weights avoids nonidentifiability due to empty components while the remaining two constraints avoid nonidentifiability due to identical components. 
	
	Based on Condition~\eqref{eq:sufficient condition}, only two of the $2G$ normal distributions --- those with corresponding $g$ for the first parameterization ($s$ for the second) --- can have the same mean and proportional covariance matrices.
	Hence, for each pair $\left(g,s\right)$, with $g,s\in\left\{1,\ldots,G\right\}$, satisfying \eqref{eq:problematic equality}, the problem reduces to comparing the pair
	\begin{equation}	
	\left\{\left\{\pi_g\alpha_{g1},\eta_{g1}\bSigma_g\right\},\left\{\pi_g\alpha_{g2},\eta_{g2}\bSigma_g\right\}\right\}
	\label{eq:para1}
	\end{equation}
	with the pair
	\begin{equation}	\left\{\left\{\widetilde{\pi}_s\widetilde{\alpha}_{s1},\widetilde{\eta}_{s1}\widetilde{\bSigma}_s\right\},\left\{\widetilde{\pi}_s\widetilde{\alpha}_{s2},\widetilde{\eta}_{s2}\widetilde{\bSigma}_s\right\}\right\}.
	\label{eq:para2}
	\end{equation}
Thanks to the constraint that the inflation parameters $\eta_{g2}$ and $\eta_{s2}$ must be greater than one, it is easy to show that $\eta_{g2}=\widetilde{\eta}_{s2}$ and $\bSigma_g=\widetilde{\bSigma}_s$.
		In particular, if we compare the first covariance matrix in \eqref{eq:para1} with the first covariance matrix in \eqref{eq:para2}, and the second covariance matrix in \eqref{eq:para1} with the second covariance matrix in \eqref{eq:para2}, we obtain
		\begin{equation}
			\left\{
			\begin{array}{l}
			\eta_{g1}\bSigma_g=\widetilde{\eta}_{s1}\widetilde{\bSigma}_s\\
			\eta_{g2}\bSigma_g=\widetilde{\eta}_{s2}\widetilde{\bSigma}_s
			\end{array}\right.
			\ 
			\Rightarrow
			\ 
			\left\{
			\begin{array}{l}
			\eta_{g2}=\widetilde{\eta}_{s2}\\
			\bSigma_g=\widetilde{\bSigma}_s
			\end{array}\right. ,
			\label{eq:var1}
		\end{equation}
		which is exactly what we need for identifiability.
		In \eqref{eq:var1}, we have used the fact that, by definition, $\eta_{g1}=\widetilde{\eta}_{s1}=1$.
		On the contrary, if we consider the remaining possibility to compare the first covariance matrix in \eqref{eq:para1} with the second covariance matrix in \eqref{eq:para2}, and the second covariance matrix in \eqref{eq:para1} with the first covariance matrix in \eqref{eq:para2}, we obtain the impossible equation $\eta_{g2}\widetilde{\eta}_{s2}=1$; this equation is impossible because $\eta_{g2}$ and $\eta_{s2}$ are both greater than one.
		
		With regard to the mixture weights, we know from \eqref{eq:var1} that the first element of \eqref{eq:para1} is related to the first element of \eqref{eq:para2} and the second element of \eqref{eq:para1} is related to the second element of \eqref{eq:para2}; accordingly, we have only to compare the corresponding weights.
			In particular, we obtain
			\begin{equation}
			\left\{
			\begin{array}{l}
			\pi_g\alpha_{g1}=\widetilde{\pi}_s\widetilde{\alpha}_{s1}\\
			\pi_g\alpha_{g2}=\widetilde{\pi}_s\widetilde{\alpha}_{s2}
			\end{array}\right.
			\ 
			\Rightarrow
			\ 
			\left\{
			\begin{array}{l}
			\pi_g\alpha_{g1}=\widetilde{\pi}_s\widetilde{\alpha}_{s1}\\
			\pi_g\left(1-\alpha_{g1}\right)=\widetilde{\pi}_s\left(1-\widetilde{\alpha}_{s1}\right)
			\end{array}\right.
			\ 
			\Rightarrow
			\ 
			\left\{
			\begin{array}{l}
			\pi_g=\widetilde{\pi}_s\\
			\alpha_{g1}=\widetilde{\alpha}_{s1}
			\end{array}\right. .
			\label{eq:dimw}
		\end{equation}
Finally, based on \eqref{eq:problematic equality}, \eqref{eq:var1}, and \eqref{eq:dimw}, after a suitable relabelling, we obtain
\begin{displaymath}
\pi_g=\widetilde{\pi}_g,\quad \alpha_{gh}=\widetilde{\alpha}_{gh},\quad \bmu_g=\widetilde{\bmu}_g,\quad \lambda_g=\widetilde{\lambda}_g,\quad \boldsymbol{\Omega}=\widetilde{\boldsymbol{\Omega}},\quad \text{and}\quad \eta_{gh}=\widetilde{\eta}_{gh}, 
\end{displaymath}
with $g\in\left\{1,\ldots,G\right\}$ and $h\in\left\{1,2\right\}$, and this completes the proof.
\end{proof}
\begin{pro}
\label{pro:2}
Let
\begin{displaymath}
p\left(\bx;\boldsymbol{\psi}\right)=\sum_{g=1}^G\sum_{h=1}^2\pi_g\alpha_{gh}\phi\left(\bx;\bmu_g,\eta_{gh}\lambda_g\boldsymbol{\Omega}\right)
\end{displaymath}
and 
\begin{displaymath}
p\left(\bx;\widetilde{\boldsymbol{\psi}}\right)=\sum_{s=1}^{\widetilde{G}}\sum_{t=1}^2\widetilde{\pi}_s\widetilde{\alpha}_{st}\phi\left(\bx;\widetilde{\bmu}_s,\widetilde{\eta}_{st}\widetilde{\lambda}_s\widetilde{\boldsymbol{\Omega}}\right)
\end{displaymath}
be two different parameterizations of the VEE model, with $\boldsymbol{\Omega}=\bGamma\boldsymbol{\Delta}\bGamma'$ and  $\widetilde{\boldsymbol{\Omega}}=\widetilde{\bGamma}\widetilde{\boldsymbol{\Delta}}\widetilde{\bGamma}'$. 
If $g\neq g_1$ implies
\begin{equation}
\left\|\bmu_g-\bmu_{g_1}\right\|_2^2\neq 0,
\label{eq:sufficient condition 2}
\end{equation}
then the equality $p\left(\bx;\boldsymbol{\psi}\right)=p(\bx;\widetilde{\boldsymbol{\psi}})$ implies that $G=\widetilde{G}$ and that there exists a relabelling such that
\begin{displaymath}
\pi_g=\widetilde{\pi}_g,\quad \alpha_{gh}=\widetilde{\alpha}_{gh},\quad \bmu_g=\widetilde{\bmu}_g,\quad \lambda_g=\widetilde{\lambda}_g,\quad \boldsymbol{\Omega}=\widetilde{\boldsymbol{\Omega}},\quad \text{and}\quad \eta_{gh}=\widetilde{\eta}_{gh}. 
\end{displaymath}
\end{pro}
\begin{proof}
Noting that the assumption $\left|\boldsymbol{\Delta}\right|=1$ (and $|\widetilde{\boldsymbol{\Delta}}|=1$) ensures that $\boldsymbol{\Omega}=\widetilde{\boldsymbol{\Omega}}$, the proof is almost identical to the proof of Proposition~\ref{pro:1}.
\end{proof}

\section{Maximum likelihood estimation}
\label{sec:Maximum likelihood estimation}

\subsection{An ECM algorithm}

To fit the models of our family, we use the expectation-conditional maximization (ECM) algorithm \citep{Meng:Rubin:Maxi:1993}. 
The ECM algorithm is a variant of the classical expectation-maximization (EM) algorithm \citep{Demp:Lair:Rubi:Maxi:1977}, which is a natural approach for maximum likelihood estimation when data are incomplete. 
In our case, there are two sources of missing data: one arises from the fact that we do not know the cluster labels $\left\{\boldsymbol{z}_i\right\}_{i=1}^n$  and the other arises from the fact that we do not know whether an observation in group $g$ is good or bad. 
To denote this second source of missing data, we use $\left\{\boldsymbol{v}_i\right\}_{i=1}^n$, where $\boldsymbol{v}_i=\left(v_{i1},\ldots,v_{iG}\right)'$ so that $v_{ig}=1$ if observation~$i$ in group~$g$ is good and $v_{ig}=0$ if observation $i$ in group $g$ is bad. 
Therefore, the complete-data are given by $\mathcal{S}=\left\{\bx_i,\boldsymbol{z}_i,\boldsymbol{v}_i\right\}_{i=1}^n$, and the complete-data log-likelihood can be written
\begin{equation}
l_c\left(\boldsymbol{\psi}|\mathcal{S}\right)=l_{1c}\left(\boldsymbol{\pi}|\mathcal{S}\right)+l_{2c}\left(\boldsymbol{\alpha}|\mathcal{S}\right)+l_{3c}\left(\bvartheta|\mathcal{S}\right),
\label{eq:complete-data log-likelihood}
\end{equation}
where 
\begin{align*}
l_{1c}\left(\boldsymbol{\pi}|\mathcal{S}\right)    &= \sum_{i=1}^{n}\sum_{g=1}^{G}{z}_{ig}\ln \pi_g,\qquad\quad
l_{2c}\left(\boldsymbol{\alpha}|\mathcal{S}\right) = \sum_{i=1}^{n}\sum_{g=1}^{G}z_{ig}\left[v_{ig}\ln \alpha_g+\left(1-v_{ig}\right)\ln \left(1-\alpha_g\right)\right],\\
l_{3c}\left(\boldsymbol{\theta}|\mathcal{S}\right) &=-\frac{1}{2}\sum_{i=1}^n\sum_{g=1}^G\Biggl\{z_{ig}\ln\left|\bSigma_g\right| +pz_{ig}\left(1-v_{ig}\right)\ln\eta_g+z_{ig}\left(v_{ig}+\frac{1-v_{ig}}{\eta_g}\right)\delta\left(\bx_i,\bmu_g;\bSigma_g\right)\Biggr\},
\end{align*}
with $\boldsymbol{\alpha}=\left(\alpha_1,\ldots,\alpha_G\right)'$ and $\boldsymbol{\theta}=\left\{\bmu_g,\bSigma_g,\eta_g\right\}_{g=1}^G$.
The ECM algorithm iterates between three steps, an E-step and two CM-steps, until convergence. 
The only difference from the EM algorithm is that each M-step is replaced by two simpler CM-steps. 
They arise from the partition $\boldsymbol{\psi}=\left\{\boldsymbol{\psi}_1,\boldsymbol{\psi}_2\right\}$, where $\boldsymbol{\psi}_1=\left\{\pi_g,\alpha_g,\bmu_g,\bSigma_g\right\}_{g=1}^G$ and $\boldsymbol{\psi}_2=\left\{\eta_g\right\}_{g=1}^G$.

\subsection{Model VVV}
\label{subsec:Model VVV}

Here, we detail the ECM algorithm for the most general VVV model~\eqref{eq:mixture of contaminated normal distributions} under no constraint for $\alpha_g$, i.e., $\alpha_g\in\left(0,1\right)$, $g=1,\ldots,G$. 

\subsubsection{E-step.}
\label{subsubsec:E-step}

The E-step, on the $\left(r+1\right)$th iteration of the ECM algorithm, requires the calculation of $Q(\boldsymbol{\psi}|\boldsymbol{\psi}^{\left(r\right)})$, the current conditional expectation of $l_c\left(\boldsymbol{\psi}|\mathcal{S}\right)$.
To do this, we need to calculate $E_{\boldsymbol{\psi}^{\left(r\right)}}\left(Z_{ig}|\bx_i\right)$ and  $E_{\boldsymbol{\psi}^{\left(r\right)}}\left(V_{ig}|\bx_i,\boldsymbol{z}_i\right)$, for $i=1,\ldots,n$ and $g=1,\ldots,G$.
They are given by
\begin{displaymath}
E_{\boldsymbol{\psi}^{\left(r\right)}}\left(Z_{ig}|\bx_i\right)=\frac{\pi_g^{\left(r\right)}f(\bx_i;\bvartheta_g^{\left(r\right)})}{p(\bx_i;\boldsymbol{\psi}^{\left(r\right)})}\eqqcolon z_{ig}^{\left(r\right)}
\end{displaymath}
and
\begin{equation}
E_{\boldsymbol{\psi}^{\left(r\right)}}\left(V_{ig}|\bx_i,\boldsymbol{z}_i\right)=\frac{\alpha_g^{\left(r\right)}\phi(\bx_i;\bmu_g^{\left(r\right)},\bSigma_g^{\left(r\right)})}{f(\bx_i;\bvartheta_g^{\left(r\right)})}\eqqcolon v_{ig}^{\left(r\right)},	
\label{eq:v update}
\end{equation}
respectively. 
Then, by substituting $z_{ig}$ with $z_{ig}^{\left(r\right)}$ and $v_{ig}$ with $v_{ig}^{\left(r\right)}$ in \eqref{eq:complete-data log-likelihood}, we obtain $Q(\boldsymbol{\psi}|\boldsymbol{\psi}^{\left(r\right)})$. 

\subsubsection{CM-step 1.}
\label{subsubsec:CM-step 1}

The first CM-step on the $\left(r+1\right)$th iteration of the ECM algorithm requires the calculation of $\boldsymbol{\psi}_1^{\left(r+1\right)}$ as the value of $\boldsymbol{\psi}_1$ that maximizes $Q(\boldsymbol{\psi}|\boldsymbol{\psi}^{\left(r\right)})$ with $\boldsymbol{\psi}_2$ fixed at $\boldsymbol{\psi}_2^{\left(r\right)}$.
In particular, we obtain
\begin{align}
\pi_g^{\left(r+1\right)}&=\frac{n_g^{\left(r\right)}}{n},\qquad\quad
\alpha_g^{\left(r+1\right)}=\frac{1}{n_g^{\left(r\right)}}\sum_{i=1}^n{z}_{ig}^{\left(r\right)}v_{ig}^{\left(r\right)},\nonumber\\
\bmu_g^{\left(r+1\right)}&=\frac{1}{s_g^{\left(r\right)}}\sum_{i=1}^n{z}_{ig}^{\left(r\right)}\left(v_{ig}^{\left(r\right)}+\frac{1-v_{ig}^{\left(r\right)}}{\eta_g^{\left(r\right)}}\right)\bx_i,\label{eq:update mu}\\
\bSigma_g^{\left(r+1\right)}&=\frac{1}{n_g^{\left(r\right)}}\boldsymbol{W}_g^{\left(r\right)}, \label{eq:SigmaX update}
\end{align}
where
\begin{equation*}\begin{split}
&s_g^{\left(r\right)}=\sum_{i=1}^nz_{ig}^{\left(r\right)}\left(v_{ig}^{\left(r\right)}+\frac{1-v_{ig}^{\left(r\right)}}{\eta_g^{\left(r\right)}}\right),\quad
\boldsymbol{W}_g^{\left(r+1\right)}=\sum_{i=1}^nz_{ig}^{\left(r\right)}\left(v_{ig}^{\left(r\right)}+\frac{1-v_{ig}^{\left(r\right)}}{\eta_g^{\left(r\right)}}\right)\left(\bx_i-\displaystyle\bmu_g^{\left(r+1\right)}\right)\left(\bx_i-\displaystyle\bmu_g^{\left(r+1\right)}\right)',
\end{split}
\end{equation*}
and $n_g^{\left(r\right)}=\sum_{i=1}^nz_{ig}^{\left(r\right)}$.

\subsubsection{CM-step 2.}
\label{subsubsec:CM-step 2}

The second CM-step, on the $\left(r+1\right)$th iteration of the ECM algorithm, requires the calculation of $\boldsymbol{\psi}_2^{\left(r+1\right)}$ as the value of $\boldsymbol{\psi}_2$ that maximizes $Q(\boldsymbol{\psi}|\boldsymbol{\psi}^{\left(r\right)})$ with $\boldsymbol{\psi}_1$ fixed at $\boldsymbol{\psi}_1^{\left(r+1\right)}$.
In particular, we have to maximize
\begin{equation}
-\frac{p}{2}\sum_{i=1}^nz_{ig}^{\left(r\right)}\left(1-v_{ig}^{\left(r\right)}\right)\ln \eta_g
-\frac{1}{2}\sum_{i=1}^n{z}_{ig}^{\left(r\right)}\frac{1-v_{ig}^{\left(r\right)}}{\eta_g}\delta\left(\bx_i,\bmu_g^{\left(r+1\right)};\bSigma_g^{\left(r+1\right)}\right),
\label{eq:maximization function for etaj}
\end{equation}
with respect to $\eta_g$, under the constraint $\eta_g>1$, for $g=1,\ldots,G$.
Operationally, the \texttt{optimize()} function, in the \texttt{stats} package for \textsf{R}, is used to perform a numerical search of the maximum $\eta_g^{\left(r+1\right)}$ of \eqref{eq:maximization function for etaj} over the interval $\left(1,\eta^*\right)$, with $\eta^*>1$.
In the analyses in Section~\ref{sec:Data analysis}, we fix $\eta^*=1000$ to facilitate faster convergence.

\subsection{Parsimonious models}
\label{subsec:Parsimonious models}

The ECM algorithm for the other models of our family changes only with respect to the way the terms of the eigen-decomposition of $\bSigma_g$ are obtained in the first CM-step. 
In particular, these updates are analogous to those given by \citet{Cele:Gova:Gaus:1995}. 
The only difference is that, on the $\left(r+1\right)$th iteration of the algorithm, $\boldsymbol{W}_g^{\left(r+1\right)}$ is used instead of the classical scatter matrix 
$$
\sum_{i=1}^nz_{ig}^{\left(r\right)}\left(\bx_i-\displaystyle\bmu_g^{\left(r+1\right)}\right)\left(\bx_i-\displaystyle\bmu_g^{\left(r+1\right)}\right)'.	
$$

\section{Further aspects}
\label{sec:fas}

\subsection{Implementation}
\label{sec:Computational aspects}

\textsf{R} source code implementing the ECM algorithm for all of the models of our family, in the form of an \textsf{R} package, is available from CRAN at \url{https://cran.r-project.org/web/package=ContaminatedMixt} \citep{Punz:Mazz:McNi:Cont:2015}.
As a basis to implement our code, we used the \texttt{mixture} package \citep{Brow:McNi:mixt:2015} for {\sf R} \citep{R}, which gives a flexible implementation of the EM algorithm for the family of parsimonious mixtures of multivariate normal distributions introduced by \citet{Cele:Gova:Gaus:1995}, hereafter abbreviated as GPCM family.
The \texttt{mixture} package differs from the \texttt{Rmixmod} package (\citealp{Bier:Cele:Gova:Lang:Noul:Vern:MIXM:2008}; \citealp{Lebr:Iovl:Lang:Bier:Cele:Gova:Rmix:2012}) with respect to the algorithm used in the M-step to estimate parameters for the EVE and VVE models.
In particular, the \texttt{Rmixmod} package adopts the classical FG-algorithm of \citet{Flur:Gaut:anal:1986}, while the \texttt{mixture} package makes use of majorization-minimization (MM) algorithms \citep{Lang:Hunt:Yang:Opti:2000,Brow:McNi:Adva:2014}.

\subsection{Initialization}
\label{subsec:Initialization}

The choice of the starting values for EM-based algorithms constitutes an important issue (see, e.g., \citealp{Bier:Cele:Gova:Choo:2003}, \citealp{Karl:Xeka:Choo:2003}, and \citealp{Bagn:Punz:Fine:2013}).
For the ECM algorithm described before, two natural strategies are:  
\begin{enumerate}
	\item providing the initial quantities $\bz_{i}^{\left(0\right)}$, $\bv_{i}^{\left(0\right)}$, and $\eta_g^{\left(0\right)}$, $i=1,\ldots,n$ and $g=1,\ldots,G$, to the first CM-step of the first iteration; and
	\item selecting an initial value $\boldsymbol{\psi}^{\left(0\right)}$ for $\boldsymbol{\psi}$ in order to run the E-step of the first iteration.
\end{enumerate}
By considering the first strategy, we suggest the following technique.
Each ($G$-cluster) model of the GPCM family tends to the corresponding ($G$-cluster) model of our family when $\alpha_g\rightarrow 1^-$ and $\eta_g\rightarrow 1^+$, $g=1,\ldots,G$.
Under these conditions, $v_{ig}\rightarrow 1^-$, $i=1,\ldots,n$ and $g=1,\ldots,G$.
Then, the posterior probabilities from the EM algorithm for each model of the GPCM family --- obtained with the \texttt{gpcm()} function of the \texttt{mixture} package --- along with the constraints $v_{ig}^{\left(0\right)}=v^{\left(0\right)}$, with $v^{\left(0\right)}\rightarrow 1^-$, and $\eta_g=\eta^{\left(0\right)}$, with $\eta^{\left(0\right)}\rightarrow 1^+$, $i=1,\ldots,n$, and $g=1,\ldots,G$, can be used to run the first CM-step of the first iteration of our ECM algorithm.
From an operational point of view, thanks to the monotonicity property of the ECM algorithm \citep[see, e.g.,][p.~28]{McLa:Kris:TheE:2007}, this also guarantees that the observed-data log-likelihood of a model from our family will be always greater than or equal to the observed-data log-likelihood of the corresponding model of the GPCM family (nested models); this is a fundamental consideration for the use of likelihood-based model selection criteria for choosing between models of our family and of the GPCM family (cf.~\citealp{Bohn:Ruan:Anot:2002} and \citealp{Punz:Brow:McNi:Hypo:2016}).
In the analyses of Section~\ref{sec:Data analysis}, $v^{\left(0\right)}=0.999$ and $\eta^{\left(0\right)}=1.001$. 

\subsection{Convergence criterion}
\label{subsec:Convergence criterion}

The Aitken acceleration \citep{Aitk:OnBe:1926} is used to estimate the asymptotic maximum of the log-likelihood at each iteration of the ECM algorithm. 
Based on this estimate, we can decide whether or not the algorithm has reached convergence, i.e., whether or not the log-likelihood is sufficiently close to its estimated asymptotic value. 
The Aitken acceleration at iteration $r+1$ is given by
\begin{equation*}
a^{\left(r+1\right)}=\frac{l^{\left(r+2\right)}-l^{\left(r+1\right)}}{l^{\left(r+1\right)}-l^{\left(r\right)}},
\end{equation*}
where $l^{\left(r\right)}$ is the observed-data log-likelihood value from iteration $r$. 
Then, the asymptotic estimate of the log-likelihood at iteration $r + 2$ is given by
\begin{equation*}
l_{\infty}^{\left(r+2\right)}=l^{\left(r+1\right)}+\frac{1}{1-a^{\left(r+1\right)}}\left(l^{\left(r+2\right)}-l^{\left(r+1\right)}\right);
\end{equation*}
cf.\ \citet{Bohn:Diet:Scha:Schl:Lind:TheD:1994}.
The ECM algorithm can be considered to have converged when $l_{\infty}^{\left(r+2\right)}-l^{\left(r+1\right)}<\epsilon$, with $\epsilon>0$, provided that this difference is positive (\citealp{McNi:Murp:McDa:Fros:Seri:2010}). 
In our analyses, we use $\epsilon=0.0001$.

\subsection{Local maxima and degeneracy of the likelihood}
\label{subsec:Degeneracy of the likelihood}

In the case of normal mixtures, it is well-known that the likelihood function: (1) presents spurious local maxima and (2) is unbounded.
It tends to infinity when one of the cluster means coincides with a sample observation and the corresponding covariance matrix tends to be singular \citep[cf.][]{Bier:Anas:2004}. 
The behaviour of the EM algorithm near a degenerate solution has been studied by \citet{Bier:Chre:Stat:2003}, \citet{Ingr:Alik:2004}, and \citet{Ingr:Rocc:Cons:2007,Ingr:Rocc:Dege:2011}, who tackle the problem by constraining the value of the smallest eigenvalue of the cluster covariance matrices \citep[see also][for the univariate case]{Hath:Acons:1986}. 
Recently, \citet{Brow:Sube:McNi:Cons:2013} consider constraining the smallest eigenvalue, the largest eigenvalue, and both the smallest and largest eigenvalues for a subset of models of the GPCM family. 
However, all these approaches require an \textit{a~priori} choice of the constraints they are based on and no rule of thumb is given to assist this choice. 
Because further study of the best threshold values for these techniques is beyond the scope of this work, we avoid considering ``preventive'' approaches in the implementation of the ECM algorithm. 

\subsection{Some notes on robustness}
\label{subsec:Some notes on robustness}

Based on \eqref{eq:update mu}, $\bmu_g^{\left(r+1\right)}$ is a weighted mean of the $\bx_i$ values, with weights depending on 
\begin{equation}
v_{ig}^{\left(r\right)}+\frac{1-v_{ig}^{\left(r\right)}}{\eta_g^{\left(r\right)}}.
\label{eq:downweight for X}
\end{equation}
Consider the update for $v_{ig}^{\left(r\right)}$, given in \eqref{eq:v update}, as a function of the squared Mahalanobis distance $\delta$, i.e.,
\begin{equation}
h\left(\delta;\alpha_g,\eta_g\right)=\frac{\alpha_g\exp\left(-\frac{\delta}{2}\right)}{\alpha_g\exp\left(-\frac{\delta}{2}\right)+\frac{\left(1-\alpha_g\right)}{\sqrt{\eta_g}}\exp\left(-\frac{\delta}{2\eta_g}\right)}=\frac{1}{1+\frac{\left(1-\alpha_g\right)}{\alpha_g}\frac{1}{\sqrt{\eta_g}}\exp\left[\frac{\delta}{2}\left(1-\frac{1}{\eta_g}\right)\right]},
\label{eq:updating function for u and v}
\end{equation}
with $\delta\geq 0$.
Due to the constraint $\eta_g>1$, from the last expression of \eqref{eq:updating function for u and v} it is straightforward to realize that $h\left(\delta;\alpha_g,\eta_g\right)$ is a decreasing function of $\delta$. 
Based on \eqref{eq:updating function for u and v}, \eqref{eq:downweight for X} can be written
\begin{equation}
w\left(\delta;\alpha_g,\eta_g\right)=h\left(\delta;\alpha_g,\eta_g\right)+\frac{1-h\left(\delta;\alpha_g,\eta_g\right)}{\eta_g}=\frac{1}{\eta_g}\left[1+\left(\eta_g-1\right)h\left(\delta;\alpha_g,\eta_g\right)\right].
\label{eq:component of the down-weighting}
\end{equation}
From the last expression of \eqref{eq:component of the down-weighting}, $w\left(\delta;\alpha_g,\eta_g\right)$ is an increasing function of $h\left(\delta;\alpha_g,\eta_g\right)$; this also means that $w\left(\delta;\alpha_g,\eta_g\right)$ is a decreasing function of $\delta$.
Therefore, the weights in \eqref{eq:downweight for X} reduce the impact of bad points in the estimation of the means $\bmu_g$, thereby providing robust estimates of these means.
In addition, from \eqref{eq:SigmaX update}, the larger $\delta$ values also have smaller effect on $\bSigma_g$, $g=1,\ldots,G$, due to the weights in \eqref{eq:downweight for X}. 
For a discussion on down-weighting for the contaminated normal distribution, see also \citet{Litt:Robu:1988}. 

\subsection{Automatic detection of bad points}
\label{subsec:Automatic detection of noise}

For a model belonging to our family, the classification of an observation $\bx_i$ means: 
\begin{description}
	\item[Step 1.] determine its cluster of membership;
	\item[Step 2.] establish whether it is a good or a bad observation in that cluster.
\end{description}
Let $\widehat{\boldsymbol{z}}_i$ and $\widehat{\boldsymbol{v}}_i$ denote, respectively, the expected values of $\boldsymbol{z}_i$ and $\boldsymbol{v}_i$ arising from the ECM algorithm, i.e., $\widehat{z}_{ig}$ is the value of $z_{ig}^{\left(r\right)}$ at convergence and $\widehat{v}_{ig}$ is the value of $v_{ig}^{\left(r\right)}$ at convergence. 
To evaluate the cluster membership of $\bx_i$, we use the MAP classification, i.e., $\text{MAP}\left(\widehat{z}_{ig}\right)$.
We then consider $\widehat{v}_{ih}$, where $h$ is selected such that $\text{MAP}\left(\widehat{z}_{ih}\right)=1$, and $\bx_i$ is considered good if $\widehat{v}_{ih}>0.5$ and $\bx_i$ is considered bad otherwise. 
The resulting information can be used to eliminate the bad points, if such an outcome is desired \citep{Berk:Bent:Esti:1988}. 
The remaining data may then be treated as effectively being distributed according to a mixture of normal distributions, and the clustering results can be reported as usual.
Finally, note that a \textit{a~posteriori} procedure (i.e., a procedure taking place once the model is fitted) to detect bad points with the mixture of multivariate $t$ distributions is illustrated by \citet[][p.~232]{McLa:Peel:fini:2000}. 
Such a procedure relies on a $\chi^2$-approximation, with $p$ degrees of freedom, of the squared Mahalanobis distances $\delta\left(\bx_i,\bmu_g;\bSigma_g\right)$, $i=1,\ldots,n$, after the MAP classification of each observation $\bx_i$ to one of the $G$ groups, and it requires a subjective choice of a percentile of the $\chi^2$ distribution in order for the observation to be classified as good or bad; such a procedure will be applied, for comparison's sake, in the analyses of Sections~\ref{sec:Numerical studies} and \ref{sec:Data analysis} by choosing the 95th percentile. 
On the contrary, the approach proposed herein is natural, in that it simultaneously identifies bad points and down-weights their impact on estimation of the mean (as well as on estimation of the covariance matrix; cf. Section~\ref{subsec:Some notes on robustness}), makes no additional distributional assumptions, and is not based on subjective choices.

\subsection{Constraints for detection of bad points}
\label{subsec:Constraints for detection of outliers}

When our models are used for detection of bad points in each group, $\left(1-\alpha_g\right)$ represents the proportion of bad points and $\eta_g$ denotes the degree of contamination.
Then, for the former parameter, one could require that in the $g$th group, $g=1,\ldots,G$, the proportion of good data is at least equal to a pre-determined value $\alpha_g^*$. 
In this case, the \texttt{optimize()} function is also used for a numerical search of the maximum $\alpha_g^{\left(r+1\right)}$, over the interval $(\alpha_g^*,1)$, of the function
\begin{displaymath}
\sum_{i=1}^nz_{ig}^{\left(r\right)}\left[v_{ig}^{\left(r\right)}\ln \alpha_g+\left(1-v_{ig}^{\left(r\right)}\right)\ln \left(1-\alpha_g\right)\right].	
\end{displaymath}
In the analyses herein (Section~\ref{sec:Data analysis}), we use this approach to update $\alpha_g$ and, as emphasized in Section~\ref{subsec:The general model}, we take $\alpha_g^*=0.5$, for $g=1,\ldots,G$.
Note that it is also possible to fix $\alpha_g$ and/or $\eta_g$ \textit{a~priori}. 

\subsection{Model selection}
\label{subsec:BIC}

The models from our family, in addition to $\boldsymbol{\psi}$, are also characterized by the particular covariance structure and by the number of clusters $G$.
Thus far, these quantities have been treated as \textit{a~priori} fixed; nevertheless, for practical purposes, model selection is usually required. One way (the usual way) to perform model selection is via computation of a convenient (likelihood-based) model selection criterion across all fourteen models and over a reasonable range of values for $G$, and then choosing the model associated with the best value of the adopted criterion (for the alternative use of likelihood-ratio tests to select either the parsimonious model or the number of components for a normal mixture, see \citealp{Lo:Mend:Rubi:Test:2001}, \citealp{Lo:Like:2005,Lo:Alik:2008},  and \citealp{Punz:Brow:McNi:Hypo:2016}).
Based on the simulation study performed by \citet{Li:Clus:2005} for the multi-layer mixture of normal distributions, in the data analyses of Section~\ref{sec:Data analysis} we will adopt the Bayesian information criterion \citep{Schw:Esti:1978}, i.e.,
\begin{displaymath}
\text{BIC}=-2l(\widehat{\boldsymbol{\psi}})+m\ln n,
\end{displaymath}
where $m$ is the overall number of free parameters in the model. 
Note that, Bayes factors can be used to compare models that are not nested, and the BIC approximation thereto holds when models are not nested (cf.~\citealp{Raft:Baye:1995}). 


\section{Simulation study: Comparison between mixtures that handle mild outliers}
\label{sec:Numerical studies}

\subsection{Overview}
In this section, we investigate the behaviour of the proposed model (for simplicity, in its unconstrained version VVV) through a large-scale simulation study performed using \textsf{R} \citep{R}.
We further provide a comparison with the unconstrained variants of the mixture models handling mild outliers and discussed in Section~\ref{sec:Introduction}.
A general feedback on advantages and drawbacks of each model is also given.
We compare:
\begin{enumerate}
	\item mixture of normal distributions (abbreviated by NM = normal mixture).
	The \texttt{gpcm()} function of the \texttt{mixture} package \citep{Brow:McNi:mixt:2015} for \textsf{R} is used to fit the unconstrained normal mixture (corresponding to the VVV model based on the nomenclature of the \texttt{mixture} package).
	The \texttt{gpcm()} function implements the EM algorithm.
	\item mixture of $t$ distributions ($t$M = $t$ mixture; \citealp{Andr:McNi:Mode:2012}).
	The \texttt{teigen()} function of the \textsf{teigen} package \citep{Andr:McNi:teig:2015} for \textsf{R} is used to fit the unconstrained $t$ mixture (corresponding to the UUUU model with respect to the nomenclature of the \textsf{teigen} package).
	The \texttt{teigen()} function implements the ECM algorithm described, for example, in \citep{Andr:McNi:teig:2015}.
	Degrees of freedom are estimated and they are allowed to vary across groups.
	\item mixture of contaminated normal distributions (CNM = contaminated normal mixture).
	The \texttt{CNmixt()} function of the \texttt{ContaminatedMixt} package \citep{Punz:Mazz:McNi:Cont:2015} for \textsf{R} is used to fit the unconstrained contaminated normal mixture (corresponding to the VVV model with respect to the nomenclature of the \texttt{ContaminatedMixt} package).
	The \texttt{CNmixt()} function implements the ECM algorithm described in Section~\ref{sec:Maximum likelihood estimation}.
	\item mixture of mixtures of a normal and a uniform distribution (NUM = normal-uniform mixture; \citealp{Brow:McNi:Spar:Mode:2012}).
	A specific \textsf{R} code, implementing the generalized-EM (GEM) algorithm described in \citet{Brow:McNi:Spar:Mode:2012}, is used to fit the unconstrained normal-uniform mixture.
	No constraint is imposed on the component uniform distributions (corresponding to model IV with respect to the nomenclature of \citealp{Brow:McNi:Spar:Mode:2012}).
	\item mixture of normal distributions plus a uniform component (NCM = noise component mixture; \citealp{Banf:Raft:mode:1993}).
	The \texttt{Mclust()} function of the \texttt{mclust} package \citep{Fral:Raft:Murp:Scru:mclu:2012,Fral:Raft:Scru:Murp:Fop:mclu:2015} for \textsf{R} is used to fit the unconstrained noise component mixture (corresponding to the VVV model with respect to the nomenclature of the \texttt{mclust} package).
	The \texttt{Mclust()} function implements the EM algorithm.
\end{enumerate}
To generate the data, we consider the following five data generation processes with $p=2$ dimensions and $G=2$ clusters:
\begin{enumerate}[label=\itshape\alph*\upshape)]
	\item\label{item:m1} NM;
	\item\label{item:m2} $t$M with $\nu_1=4$ and $\nu_2=10$ degrees of freedom; 
	\item\label{item:m3} CNM with $\alpha_1=0.9$, $\alpha_2=0.8$, $\eta_1=20$, and $\eta_2=30$;
	\item\label{item:m4} NM with 1\% of points randomly substituted by high atypical points with coordinates $\left(0,x_{i2}^*\right)$, where $x_{i2}^*$ is generated from a uniform distribution over the interval $\left(10,15\right)$.
	\item\label{item:m5} NM with 5\% of points randomly substituted by noise points generated from a uniform distribution over the interval $\left(-10,10\right)$ on each dimension.
\end{enumerate}
All of these data generation processes share the following common parameters
\begin{equation*}
	\pi_1=0.3,
	\quad
	\bmu_2=
	\begin{pmatrix*}[c]
0	\\
3	 
\end{pmatrix*},
\quad 
	\bSigma_1=\begin{pmatrix*}[c]
1 & -0.5	\\
-0.5 & 1	 
\end{pmatrix*},
\quad \text{and} \quad 
	\bSigma_2=\begin{pmatrix*}[c]
1 & 0.5	\\
0.5 & 1	 
\end{pmatrix*}.
\label{eq:common generating parameters}
\end{equation*}
As concerns the mean in the first group, two alternatives are considered in order to reproduce two different degrees of overlap between clusters: $\bmu_1=\left(0,-3\right)'$ in the ``far'' case, and $\bmu_1=\left(0,-1\right)'$ in the ``close'' case.
The five scenarios above cover different situations which may arise dealing with real-world data: no bad points for scenario~\ref{item:m1}, heavy-tails cluster distributions for scenarios \ref{item:m2} and \ref{item:m3}, and two different types of bad points for scenarios \ref{item:m4} and \ref{item:m5}.
Under each scenario, we simulate 1,000 samples considering the number of analyzed units $n$ (100, 200, and 500), as well as the degree of overlap (``far'' and ``close''), as experimental factors.
This yields a total of 30,000 generated data sets.
On each generated data set, the five competing models are directly run with $G=2$.
As concerns the initialization strategy of the EM-based algorithms for the first four models (NM, $t$M, CNM, and NUM), the partition provided by the $k$-means method, as implemented by the \texttt{kmeans()} function, with default arguments, of the \texttt{stats} package for \textsf{R}, is considered.
As concerns the NCM, an initial guess of the noise observations must be supplied via the \texttt{noise} component of the \texttt{initialization} argument in \texttt{Mclust()}.
Nearest neighbor based clutter/noise detection proposed by \citet{Byer:Raft:Near:1998} is applied to identify an initial set of noise points. 
The latter is implemented in the \texttt{NNclean()} function in \textsf{R}'s \texttt{prabclus} package \citep{Henn:Bern:prab:2015}. 
Agglomerative hierarchical clustering based on ML criteria for normal mixtures proposed by \citet{Banf:Raft:mode:1993} is then used for finding initial normal clusters in the non-noise data. 
This is implemented in the \texttt{hc()} function of \textsf{R}'s \texttt{mclust} package.

Before presenting the obtained results, we want to underline that the average elapsed time (in seconds over the 30,000 replications) to fit a single CNM is 0.692 seconds.
This information is useful to have an idea of the computational burden required by our ECM algorithm.
Computation is performed on a Windows 8.1 PC, with Intel i7 3.50GHz CPU, 16.0 GB RAM, using \textsf{R} 32 bit, and the elapsed time is computed via the \texttt{proc.time()} function of the \texttt{base} package.


\subsection{Parameter estimation}

For comparison's sake, we report the bias (BIAS) and the standard deviation (STD) of the estimates for the mixture weight $\pi_1$, the univariate means $\mu_{11}$ and $\mu_{21}$ (elements of $\bmu_1$), and the univariate means $\mu_{12}$ and $\mu_{22}$ (elements of $\bmu_2$). 
Before to illustrate the obtained results, it is important to underline that under mixture models there are well known label switching issues (see, e.g., \citealp{Cele:Hurn:Robe:Comp:2000}, \citealp{Step:Deal:2000}, and \citealp{Yao:Mode:2012}) when evaluating properties of the estimators of the parameters using simulation studies.
There are no generally accepted labeling methods.
In our simulation study, as in \cite{Bai:Yao:Boyer:Robu:2012} and \cite{Yao:Wei:Yu:Robu:2014}, we choose the labels by minimizing the distance to the true parameter values.

\tablename~\ref{tab:TABa} reports the results under scenario~\ref{item:m1}, that is when there are no bad points. 
Here, as expected, NMs, $t$Ms, and CNMs perform comparably because, in this situation, the $t$M and the CNM tend to the NM.
NCMs work well too, apart from the estimation of the mixture weights, while NUMs provide the worst results. 
Finally, regardless of both the considered model and the parameter of interest, the BIAS and the STD values improve with the increase of $n$ and they are better under the ``far'' case, as expected. 
\setlength{\tabcolsep}{5.1pt}
\renewcommand{\arraystretch}{0.41}
\begin{table}[!ht]
\caption{
Scenario~\ref{item:m1}: Simulation results on $1,000$ replications.
\label{tab:TABa}
}
\centering
\resizebox{\textwidth}{!}{
\begin{tabular}{lcll rr l rr l rr l rr l rr}
  \toprule
 & & & & \multicolumn{2}{c}{NM} & & \multicolumn{2}{c}{$t$M} & & \multicolumn{2}{c}{CNM} & & \multicolumn{2}{c}{NUM} & & \multicolumn{2}{c}{NCM} \\
 & $n$ & & &	BIAS	&	STD &	&	BIAS	&	STD &	&	BIAS	&	STD &	&	BIAS	&	STD &	&	BIAS	&	STD \\ 
  \midrule
Far & 100 & $\pi_1=0.3$ &  & 0.001 & 0.045 &  & 0.001 & 0.045 &  & 0.001 & 0.045 &  & 0.006 & 0.085 &  & 0.127 & 0.276 \\ 
&  & $\mu_{11}=0$ &  & -0.002 & 0.184 &  & -0.002 & 0.185 &  & -0.001 & 0.184 &  & 0.087 & 0.352 &  & -0.001 & 0.218 \\ 
&  & $\mu_{21}=-3$ &  & 0.002 & 0.189 &  & 0.002 & 0.189 &  & 0.003 & 0.188 &  & -0.125 & 0.313 &  & 0.021 & 0.421 \\ 
&  & $\mu_{12}=0$ &  & 0.004 & 0.114 &  & 0.003 & 0.115 &  & 0.004 & 0.114 &  & 0.036 & 0.163 &  & 0.016 & 0.173 \\ 
&  & $\mu_{22}=3$ &  & 0.005 & 0.119 &  & 0.004 & 0.119 &  & 0.004 & 0.119 &  & 0.047 & 0.148 &  & 0.022 & 0.203 \\[2mm] 
& 200 & $\pi_1=0.3$ &  & -0.001 & 0.033 &  & -0.001 & 0.033 &  & -0.001 & 0.033 &  & 0.007 & 0.074 &  & 0.086 & 0.317 \\ 
&  & $\mu_{11}=0$ &  & -0.009 & 0.128 &  & -0.008 & 0.129 &  & -0.009 & 0.128 &  & 0.096 & 0.270 &  & -0.009 & 0.138 \\ 
&  & $\mu_{21}=-3$ &  & 0.008 & 0.130 &  & 0.007 & 0.131 &  & 0.008 & 0.130 &  & -0.136 & 0.223 &  & 0.013 & 0.229 \\ 
&  & $\mu_{12}=0$ &  & 0.001 & 0.089 &  & 0.000 & 0.089 &  & 0.000 & 0.089 &  & 0.043 & 0.133 &  & 0.000 & 0.103 \\ 
&  & $\mu_{22}=3$ &  & 0.001 & 0.089 &  & -0.000 & 0.089 &  & 0.000 & 0.089 &  & 0.052 & 0.113 &  & -0.000 & 0.094 \\[2mm]
& 500 & $\pi_1=0.3$ &  & 0.001 & 0.021 &  & 0.001 & 0.021 &  & 0.001 & 0.021 &  & 0.016 & 0.057 &  & 0.051 & 0.338 \\ 
&  & $\mu_{11}=0$ &  & -0.001 & 0.082 &  & -0.001 & 0.082 &  & -0.001 & 0.082 &  & 0.067 & 0.175 &  & -0.002 & 0.084 \\ 
&  & $\mu_{21}=-3$ &  & -0.002 & 0.085 &  & -0.001 & 0.086 &  & -0.001 & 0.085 &  & -0.101 & 0.140 &  & 0.001 & 0.087 \\ 
&  & $\mu_{12}=0$ &  & 0.000 & 0.055 &  & 0.000 & 0.055 &  & 0.000 & 0.055 &  & 0.039 & 0.096 &  & -0.001 & 0.055 \\ 
&  & $\mu_{22}=3$ &  & -0.003 & 0.053 &  & -0.003 & 0.053 &  & -0.003 & 0.053 &  & 0.037 & 0.074 &  & -0.004 & 0.053 \\[4mm]
Close & 100 & $\pi_1=0.3$ &  & 0.003 & 0.051 &  & 0.003 & 0.051 &  & 0.003 & 0.050 &  & 0.084 & 0.088 &  & 0.114 & 0.238 \\ 
&  & $\mu_{11}=0$ &  & 0.007 & 0.184 &  & 0.010 & 0.183 &  & 0.010 & 0.183 &  & 0.109 & 0.272 &  & -0.038 & 0.266 \\ 
&  & $\mu_{21}=-1$ &  & -0.003 & 0.214 &  & -0.005 & 0.211 &  & -0.006 & 0.209 &  & -0.226 & 0.301 &  & 0.297 & 0.925 \\ 
&  & $\mu_{12}=0$ &  & 0.001 & 0.122 &  & -0.002 & 0.121 &  & -0.000 & 0.121 &  & -0.002 & 0.181 &  & 0.065 & 0.345 \\ 
&  & $\mu_{22}=3$ &  & -0.003 & 0.134 &  & -0.008 & 0.132 &  & -0.004 & 0.131 &  & -0.031 & 0.165 &  & 0.047 & 0.342 \\[2mm]
& 200 & $\pi_1=0.3$ &  & 0.003 & 0.036 &  & 0.003 & 0.036 &  & 0.003 & 0.036 &  & 0.071 & 0.074 &  & 0.150 & 0.288 \\ 
&  & $\mu_{11}=0$ &  & 0.001 & 0.149 &  & 0.004 & 0.147 &  & 0.002 & 0.148 &  & 0.073 & 0.199 &  & -0.012 & 0.159 \\ 
&  & $\mu_{21}=-1$ &  & 0.007 & 0.161 &  & 0.004 & 0.157 &  & 0.005 & 0.159 &  & -0.173 & 0.234 &  & 0.041 & 0.266 \\ 
&  & $\mu_{12}=0$ &  & 0.007 & 0.087 &  & 0.004 & 0.087 &  & 0.006 & 0.087 &  & 0.015 & 0.141 &  & 0.002 & 0.114 \\ 
&  & $\mu_{22}=3$ &  & 0.005 & 0.096 &  & 0.001 & 0.095 &  & 0.004 & 0.095 &  & -0.016 & 0.120 &  & -0.006 & 0.104 \\[2mm]
& 500 & $\pi_1=0.3$ &  & 0.001 & 0.022 &  & 0.001 & 0.022 &  & 0.001 & 0.022 &  & 0.056 & 0.055 &  & 0.135 & 0.323 \\ 
&  & $\mu_{11}=0$ &  & 0.002 & 0.093 &  & 0.003 & 0.092 &  & 0.002 & 0.092 &  & 0.055 & 0.124 &  & -0.007 & 0.094 \\ 
&  & $\mu_{21}=-1$ &  & 0.004 & 0.103 &  & 0.004 & 0.102 &  & 0.004 & 0.103 &  & -0.102 & 0.154 &  & 0.019 & 0.106 \\ 
&  & $\mu_{12}=0$ &  & 0.002 & 0.056 &  & 0.001 & 0.056 &  & 0.002 & 0.056 &  & 0.028 & 0.109 &  & -0.001 & 0.056 \\ 
&  & $\mu_{22}=3$ &  & 0.002 & 0.063 &  & -0.001 & 0.063 &  & 0.002 & 0.063 &  & -0.012 & 0.083 &  & -0.003 & 0.062 \\ 
   \bottomrule
\end{tabular}
}
\end{table}

\tablename~\ref{tab:TABb} and \ref{tab:TABc} report the results under scenarios~\ref{item:m2} and \ref{item:m3}, respectively.
Here, the robust approaches ($t$M, CNM, NUM, and NCM) are better than the traditional NM.
As expected, the best performer is the $t$M under scenario~\ref{item:m2} and the CNM under scenario~\ref{item:m3}.
These models are the best two under these scenarios; their comparable behavior agrees with the simulation results of \citet{Litt:Robu:1988} about the single $t$ and the contaminated normal distributions.
NUMs and NCMs work slightly better than NMs but far worse from $t$Ms and CNMs (see, e.g., the STD values, in the case $n=100$, in \tablename~\ref{tab:TABc}).  
\begin{table}[!ht]
\caption{
Scenario~\ref{item:m2}: Simulation results on $1,000$ replications.
\label{tab:TABb}
}
\centering
\resizebox{\textwidth}{!}{
\begin{tabular}{lcll rr l rr l rr l rr l rr}
  \toprule
 & & & & \multicolumn{2}{c}{NM} & & \multicolumn{2}{c}{$t$M} & & \multicolumn{2}{c}{CNM} & & \multicolumn{2}{c}{NUM} & & \multicolumn{2}{c}{NCM} \\
 & $n$ & & &	BIAS	&	STD &	&	BIAS	&	STD &	&	BIAS	&	STD &	&	BIAS	&	STD &	&	BIAS	&	STD \\ 
  \midrule
Far & 100 & $\pi_1=0.3$ &  & -0.015 & 0.063 &  & -0.020 & 0.053 &  & -0.008 & 0.052 &  & 0.005 & 0.069 &  & 0.089 & 0.238 \\ 
&  & $\mu_{11}=0$ &  & -0.221 & 0.298 &  & -0.072 & 0.283 &  & -0.157 & 0.301 &  & 0.043 & 0.341 &  & 0.118 & 0.319 \\ 
&  & $\mu_{21}=-3$ &  & 0.084 & 0.454 &  & -0.026 & 0.278 &  & 0.057 & 0.327 &  & -0.069 & 0.381 &  & -0.075 & 0.425 \\ 
&  & $\mu_{12}=0$ &  & -0.385 & 0.641 &  & -0.140 & 0.163 &  & -0.162 & 0.162 &  & -0.064 & 0.201 &  & -0.070 & 0.427 \\ 
&  & $\mu_{22}=3$ &  & -0.096 & 0.645 &  & -0.093 & 0.153 &  & -0.093 & 0.147 &  & -0.051 & 0.171 &  & -0.081 & 0.402 \\[2mm] 
 & 200 & $\pi_1=0.3$ &  & -0.016 & 0.048 &  & -0.018 & 0.038 &  & -0.010 & 0.037 &  & -0.002 & 0.044 &  & 0.101 & 0.275 \\ 
&  & $\mu_{11}=0$ &  & -0.179 & 0.223 &  & -0.058 & 0.196 &  & -0.121 & 0.243 &  & -0.098 & 0.257 &  & 0.021 & 0.240 \\ 
&  & $\mu_{21}=-3$ &  & 0.043 & 0.330 &  & -0.031 & 0.205 &  & 0.036 & 0.285 &  & 0.028 & 0.324 &  & -0.043 & 0.258 \\ 
&  & $\mu_{12}=0$ &  & -0.357 & 0.312 &  & -0.125 & 0.112 &  & -0.148 & 0.115 &  & -0.138 & 0.144 &  & -0.101 & 0.199 \\ 
&  & $\mu_{22}=3$ &  & -0.116 & 0.116 &  & -0.084 & 0.105 &  & -0.092 & 0.104 &  & -0.077 & 0.117 &  & -0.082 & 0.160 \\[2mm]
 & 500 & $\pi_1=0.3$ &  & -0.020 & 0.023 &  & -0.019 & 0.025 &  & -0.019 & 0.024 &  & -0.004 & 0.031 &  & 0.059 & 0.308 \\ 
&  & $\mu_{11}=0$ &  & -0.152 & 0.159 &  & -0.040 & 0.118 &  & -0.079 & 0.148 &  & -0.209 & 0.218 &  & -0.060 & 0.187 \\ 
&  & $\mu_{21}=-3$ &  & 0.009 & 0.201 &  & -0.046 & 0.122 &  & -0.025 & 0.186 &  & 0.131 & 0.324 &  & 0.002 & 0.263 \\ 
&  & $\mu_{12}=0$ &  & -0.326 & 0.074 &  & -0.108 & 0.069 &  & -0.139 & 0.072 &  & -0.171 & 0.103 &  & -0.145 & 0.382 \\ 
&  & $\mu_{22}=3$ &  & -0.116 & 0.070 &  & -0.075 & 0.066 &  & -0.095 & 0.067 &  & -0.084 & 0.081 &  & -0.083 & 0.338 \\[4mm]
Close & 100 & $\pi_1=0.3$ &  & 0.084 & 0.154 &  & 0.006 & 0.093 &  & 0.036 & 0.084 &  & 0.060 & 0.092 &  & 0.106 & 0.225 \\ 
&  & $\mu_{11}=0$ &  & -0.668 & 0.508 &  & -0.129 & 0.475 &  & -0.287 & 0.424 &  & -0.141 & 0.409 &  & 0.089 & 0.424 \\ 
&  & $\mu_{21}=-1$ &  & 0.749 & 0.899 &  & 0.058 & 0.527 &  & 0.218 & 0.485 &  & 0.124 & 0.563 &  & 0.070 & 0.763 \\ 
&  & $\mu_{12}=0$ &  & -0.413 & 1.015 &  & -0.137 & 0.226 &  & -0.095 & 0.222 &  & 0.022 & 0.232 &  & -0.053 & 0.514 \\ 
&  & $\mu_{22}=3$ &  & 0.103 & 0.658 &  & -0.068 & 0.225 &  & -0.021 & 0.208 &  & 0.007 & 0.227 &  & -0.084 & 0.586 \\[2mm]
 & 200 & $\pi_1=0.3$ &  & 0.088 & 0.168 &  & 0.016 & 0.090 &  & 0.052 & 0.076 &  & 0.051 & 0.082 &  & 0.072 & 0.272 \\ 
&  & $\mu_{11}=0$ &  & -0.677 & 0.450 &  & -0.189 & 0.429 &  & -0.364 & 0.374 &  & -0.282 & 0.371 &  & -0.100 & 0.377 \\ 
&  & $\mu_{21}=-1$ &  & 0.783 & 0.888 &  & 0.151 & 0.528 &  & 0.362 & 0.469 &  & 0.291 & 0.532 &  & 0.225 & 0.697 \\ 
&  & $\mu_{12}=0$ &  & -0.535 & 1.564 &  & -0.102 & 0.190 &  & -0.032 & 0.185 &  & -0.023 & 0.250 &  & -0.118 & 0.730 \\ 
&  & $\mu_{22}=3$ &  & 0.180 & 1.272 &  & -0.049 & 0.185 &  & 0.016 & 0.168 &  & 0.009 & 0.195 &  & -0.047 & 0.509 \\[2mm]
 & 500 & $\pi_1=0.3$ &  & 0.052 & 0.165 &  & 0.062 & 0.088 &  & 0.094 & 0.056 &  & 0.067 & 0.074 &  & 0.064 & 0.292 \\ 
&  & $\mu_{11}=0$ &  & -0.515 & 0.507 &  & -0.384 & 0.374 &  & -0.539 & 0.253 &  & -0.417 & 0.342 &  & -0.279 & 0.356 \\ 
&  & $\mu_{21}=-1$ &  & 0.579 & 0.895 &  & 0.457 & 0.514 &  & 0.650 & 0.341 &  & 0.485 & 0.500 &  & 0.357 & 0.516 \\ 
&  & $\mu_{12}=0$ &  & -0.517 & 1.367 &  & -0.017 & 0.180 &  & 0.062 & 0.121 &  & 0.003 & 0.186 &  & -0.038 & 0.426 \\ 
&  & $\mu_{22}=3$ &  & 0.049 & 1.038 &  & 0.035 & 0.166 &  & 0.106 & 0.110 &  & 0.055 & 0.163 &  & 0.018 & 0.290 \\ 
   \bottomrule
\end{tabular}
}
\end{table}

\begin{table}[!ht]
\caption{
Scenario~\ref{item:m3}: Simulation results on $1,000$ replications.
\label{tab:TABc}
}
\centering
\resizebox{\textwidth}{!}{
\begin{tabular}{lcll rr l rr l rr l rr l rr}
  \toprule
 & & & & \multicolumn{2}{c}{NM} & & \multicolumn{2}{c}{$t$M} & & \multicolumn{2}{c}{CNM} & & \multicolumn{2}{c}{NUM} & & \multicolumn{2}{c}{NCM} \\
 & $n$ & & &	BIAS	&	STD &	&	BIAS	&	STD &	&	BIAS	&	STD &	&	BIAS	&	STD &	&	BIAS	&	STD \\ 
  \midrule
Far & 100 & $\pi_1=0.3$ &  & 0.097 & 0.199 &  & 0.010 & 0.055 &  & 0.007 & 0.055 &  & 0.098 & 0.073 &  & 0.023 & 0.227 \\ 
&  & $\mu_{11}=0$ &  & -0.190 & 0.559 &  & 0.012 & 0.229 &  & -0.001 & 0.224 &  & -0.076 & 0.396 &  & -0.013 & 0.426 \\ 
&  & $\mu_{21}=-3$ &  & 0.591 & 1.568 &  & -0.035 & 0.234 &  & -0.010 & 0.236 &  & 0.621 & 0.914 &  & 0.625 & 1.541 \\ 
&  & $\mu_{12}=0$ &  & 0.187 & 1.746 &  & -0.001 & 0.160 &  & 0.001 & 0.154 &  & -0.000 & 0.217 &  & 0.119 & 2.291 \\ 
&  & $\mu_{22}=3$ &  & 0.220 & 1.747 &  & 0.005 & 0.163 &  & 0.002 & 0.153 &  & -0.011 & 0.258 &  & 0.186 & 2.012 \\[2mm] 
& 200 & $\pi_1=0.3$ &  & 0.083 & 0.179 &  & 0.005 & 0.038 &  & 0.001 & 0.038 &  & 0.086 & 0.076 &  & 0.039 & 0.231 \\ 
&  & $\mu_{11}=0$ &  & -0.108 & 0.391 &  & 0.013 & 0.159 &  & -0.001 & 0.154 &  & -0.061 & 0.308 &  & -0.024 & 0.179 \\ 
&  & $\mu_{21}=-3$ &  & 0.641 & 1.468 &  & -0.021 & 0.152 &  & 0.003 & 0.149 &  & 0.571 & 0.814 &  & 0.273 & 0.959 \\ 
&  & $\mu_{12}=0$ &  & 0.085 & 1.274 &  & 0.002 & 0.107 &  & 0.005 & 0.101 &  & 0.003 & 0.154 &  & 0.008 & 1.439 \\ 
&  & $\mu_{22}=3$ &  & 0.073 & 1.001 &  & 0.005 & 0.107 &  & 0.003 & 0.102 &  & -0.021 & 0.236 &  & 0.134 & 1.189 \\[2mm]
& 500 & $\pi_1=0.3$ &  & 0.067 & 0.138 &  & 0.006 & 0.025 &  & 0.001 & 0.023 &  & 0.082 & 0.073 &  & 0.071 & 0.230 \\ 
&  & $\mu_{11}=0$ &  & -0.033 & 0.273 &  & 0.020 & 0.097 &  & 0.009 & 0.093 &  & -0.020 & 0.211 &  & -0.019 & 0.096 \\ 
&  & $\mu_{21}=-3$ &  & 0.643 & 1.343 &  & -0.026 & 0.095 &  & -0.005 & 0.093 &  & 0.499 & 0.789 &  & 0.063 & 0.398 \\ 
&  & $\mu_{12}=0$ &  & 0.026 & 0.313 &  & 0.000 & 0.070 &  & 0.000 & 0.064 &  & 0.002 & 0.101 &  & -0.008 & 0.567 \\ 
&  & $\mu_{22}=3$ &  & -0.010 & 0.361 &  & 0.009 & 0.067 &  & 0.005 & 0.061 &  & -0.018 & 0.185 &  & 0.040 & 0.463 \\[4mm]
Close & 100 & $\pi_1=0.3$ &  & 0.396 & 0.292 &  & 0.010 & 0.077 &  & 0.013 & 0.069 &  & 0.158 & 0.202 &  & 0.005 & 0.231 \\ 
&  & $\mu_{11}=0$ &  & -0.108 & 0.650 &  & 0.068 & 0.275 &  & 0.017 & 0.266 &  & -0.022 & 0.531 &  & -0.011 & 0.289 \\ 
&  & $\mu_{21}=-1$ &  & 2.323 & 0.945 &  & -0.052 & 0.338 &  & -0.001 & 0.330 &  & 0.969 & 1.063 &  & 0.660 & 1.165 \\ 
&  & $\mu_{12}=0$ &  & -0.033 & 2.966 &  & -0.007 & 0.175 &  & 0.003 & 0.161 &  & -0.066 & 0.793 &  & 0.072 & 2.420 \\ 
&  & $\mu_{22}=3$ &  & -0.349 & 2.672 &  & -0.013 & 0.190 &  & 0.006 & 0.189 &  & -0.298 & 0.904 &  & 0.159 & 2.303 \\[2mm]
& 200 & $\pi_1=0.3$ &  & 0.348 & 0.301 &  & -0.002 & 0.047 &  & 0.006 & 0.043 &  & 0.126 & 0.198 &  & 0.033 & 0.230 \\ 
&  & $\mu_{11}=0$ &  & -0.055 & 0.406 &  & 0.063 & 0.170 &  & 0.000 & 0.173 &  & -0.028 & 0.428 &  & -0.032 & 0.195 \\ 
&  & $\mu_{21}=-1$ &  & 2.445 & 0.597 &  & -0.067 & 0.187 &  & 0.005 & 0.198 &  & 0.960 & 1.000 &  & 0.270 & 0.768 \\ 
&  & $\mu_{12}=0$ &  & -0.020 & 2.093 &  & 0.000 & 0.115 &  & 0.011 & 0.114 &  & 0.014 & 0.431 &  & 0.051 & 1.473 \\ 
&  & $\mu_{22}=3$ &  & -0.663 & 1.580 &  & -0.020 & 0.125 &  & 0.004 & 0.124 &  & -0.335 & 0.617 &  & 0.090 & 1.346 \\[2mm]
& 500 & $\pi_1=0.3$ &  & 0.363 & 0.300 &  & -0.004 & 0.030 &  & 0.001 & 0.027 &  & 0.119 & 0.203 &  & 0.039 & 0.228 \\ 
&  & $\mu_{11}=0$ &  & -0.031 & 0.225 &  & 0.056 & 0.108 &  & -0.005 & 0.105 &  & 0.012 & 0.370 &  & -0.027 & 0.123 \\ 
&  & $\mu_{21}=-1$ &  & 2.813 & 0.281 &  & -0.058 & 0.117 &  & 0.008 & 0.115 &  & 1.060 & 0.992 &  & 0.094 & 0.441 \\ 
&  & $\mu_{12}=0$ &  & 0.109 & 0.816 &  & -0.007 & 0.071 &  & 0.002 & 0.067 &  & 0.003 & 0.348 &  & 0.048 & 0.609 \\ 
&  & $\mu_{22}=3$ &  & -0.811 & 0.631 &  & -0.020 & 0.077 &  & 0.003 & 0.071 &  & -0.367 & 0.587 &  & 0.062 & 0.629 \\ 
   \bottomrule
\end{tabular}
}
\end{table}

\tablename~\ref{tab:TABd} reports the results under scenario~\ref{item:m4}, that is when there is the 1\% of bad points with a specific location in the space.
By focusing on the STD values, $t$Ms and CNMs perform comparably, with a slightly better performance for CNMs when the sample size is small ($n=100$).
Surprisingly, NUM and NCM perform worse than NM (refer, e.g., to the case $n=100$).
\begin{table}[!ht]
\caption{
Scenario~\ref{item:m4}: Simulation results on $1,000$ replications.
\label{tab:TABd}
}
\centering
\resizebox{\textwidth}{!}{
\begin{tabular}{lcll rr l rr l rr l rr l rr}
  \toprule
 & & & & \multicolumn{2}{c}{NM} & & \multicolumn{2}{c}{$t$M} & & \multicolumn{2}{c}{CNM} & & \multicolumn{2}{c}{NUM} & & \multicolumn{2}{c}{NCM} \\
 & $n$ & & &	BIAS	&	STD &	&	BIAS	&	STD &	&	BIAS	&	STD &	&	BIAS	&	STD &	&	BIAS	&	STD \\ 
  \midrule
Far & 100 & $\pi_1=0.3$ &  & -0.024 & 0.068 &  & -0.011 & 0.062 &  & -0.004 & 0.059 &  & 0.051 & 0.091 &  & 0.139 & 0.298 \\ 
&  & $\mu_{11}=0$ &  & 0.062 & 0.270 &  & 0.023 & 0.256 &  & 0.006 & 0.236 &  & 0.063 & 0.348 &  & 0.010 & 0.290 \\ 
&  & $\mu_{21}=-3$ &  & -0.085 & 0.300 &  & -0.014 & 0.263 &  & 0.009 & 0.247 &  & -0.003 & 0.566 &  & 0.070 & 0.669 \\ 
&  & $\mu_{12}=0$ &  & -0.015 & 0.167 &  & -0.000 & 0.170 &  & 0.005 & 0.157 &  & 0.031 & 0.197 &  & 0.112 & 0.459 \\ 
&  & $\mu_{22}=3$ &  & 0.146 & 0.302 &  & 0.004 & 0.173 &  & 0.006 & 0.158 &  & 0.039 & 0.187 &  & 0.245 & 1.051 \\[2mm] 
& 200 & $\pi_1=0.3$ &  & -0.016 & 0.038 &  & -0.009 & 0.037 &  & -0.005 & 0.037 &  & 0.052 & 0.057 &  & 0.115 & 0.321 \\ 
&  & $\mu_{11}=0$ &  & 0.034 & 0.159 &  & 0.015 & 0.153 &  & -0.004 & 0.152 &  & 0.049 & 0.208 &  & -0.000 & 0.217 \\ 
&  & $\mu_{21}=-3$ &  & -0.066 & 0.149 &  & -0.023 & 0.150 &  & 0.005 & 0.152 &  & -0.042 & 0.346 &  & 0.169 & 0.954 \\ 
&  & $\mu_{12}=0$ &  & -0.011 & 0.101 &  & 0.000 & 0.105 &  & 0.004 & 0.100 &  & 0.035 & 0.121 &  & 0.113 & 0.465 \\ 
&  & $\mu_{22}=3$ &  & 0.124 & 0.104 &  & 0.008 & 0.101 &  & 0.010 & 0.098 &  & 0.031 & 0.113 &  & 0.151 & 0.801 \\[2mm]
& 500 & $\pi_1=0.3$ &  & -0.013 & 0.022 &  & -0.008 & 0.022 &  & -0.005 & 0.022 &  & 0.046 & 0.034 &  & 0.128 & 0.324 \\ 
&  & $\mu_{11}=0$ &  & 0.030 & 0.088 &  & 0.019 & 0.086 &  & 0.001 & 0.086 &  & 0.034 & 0.114 &  & -0.002 & 0.086 \\ 
&  & $\mu_{21}=-3$ &  & -0.057 & 0.092 &  & -0.028 & 0.091 &  & -0.002 & 0.092 &  & -0.035 & 0.225 &  & 0.027 & 0.378 \\ 
&  & $\mu_{12}=0$ &  & -0.011 & 0.058 &  & -0.001 & 0.061 &  & 0.001 & 0.058 &  & 0.024 & 0.073 &  & 0.012 & 0.167 \\ 
&  & $\mu_{22}=3$ &  & 0.102 & 0.063 &  & 0.002 & 0.061 &  & 0.003 & 0.059 &  & 0.019 & 0.075 &  & 0.010 & 0.145 \\[4mm]
Close & 100 & $\pi_1=0.3$ &  & -0.116 & 0.105 &  & -0.019 & 0.079 &  & -0.002 & 0.071 &  & 0.119 & 0.093 &  & 0.112 & 0.270 \\ 
&  & $\mu_{11}=0$ &  & 0.336 & 0.445 &  & 0.058 & 0.321 &  & 0.019 & 0.262 &  & 0.137 & 0.345 &  & 0.002 & 0.360 \\ 
&  & $\mu_{21}=-1$ &  & -0.260 & 0.721 &  & -0.025 & 0.371 &  & -0.003 & 0.285 &  & 0.155 & 0.843 &  & 0.261 & 0.957 \\ 
&  & $\mu_{12}=0$ &  & -0.056 & 0.218 &  & -0.005 & 0.195 &  & 0.002 & 0.169 &  & -0.002 & 0.236 &  & 0.120 & 0.496 \\ 
&  & $\mu_{22}=3$ &  & -0.247 & 0.524 &  & -0.026 & 0.209 &  & -0.003 & 0.182 &  & -0.063 & 0.280 &  & 0.110 & 0.920 \\[2mm]
& 200 & $\pi_1=0.3$ &  & -0.118 & 0.029 &  & -0.021 & 0.047 &  & -0.004 & 0.046 &  & 0.082 & 0.083 &  & 0.077 & 0.326 \\ 
&  & $\mu_{11}=0$ &  & 0.329 & 0.251 &  & 0.041 & 0.205 &  & -0.007 & 0.191 &  & 0.065 & 0.254 &  & -0.011 & 0.314 \\ 
&  & $\mu_{21}=-1$ &  & -0.407 & 0.284 &  & -0.046 & 0.239 &  & 0.003 & 0.217 &  & 0.225 & 0.732 &  & 0.259 & 0.901 \\ 
&  & $\mu_{12}=0$ &  & -0.076 & 0.128 &  & -0.006 & 0.124 &  & 0.006 & 0.119 &  & 0.008 & 0.167 &  & 0.146 & 0.612 \\ 
&  & $\mu_{22}=3$ &  & -0.260 & 0.289 &  & -0.024 & 0.132 &  & 0.003 & 0.126 &  & -0.048 & 0.214 &  & 0.069 & 1.116 \\[2mm]
& 500 & $\pi_1=0.3$ &  & -0.103 & 0.034 &  & -0.020 & 0.026 &  & -0.006 & 0.027 &  & 0.045 & 0.069 &  & 0.065 & 0.334 \\ 
&  & $\mu_{11}=0$ &  & 0.322 & 0.040 &  & 0.042 & 0.110 &  & -0.005 & 0.111 &  & 0.048 & 0.177 &  & -0.006 & 0.121 \\ 
&  & $\mu_{21}=-1$ &  & -0.403 & 0.254 &  & -0.046 & 0.129 &  & 0.004 & 0.128 &  & 0.162 & 0.571 &  & 0.072 & 0.426 \\ 
&  & $\mu_{12}=0$ &  & -0.090 & 0.058 &  & -0.013 & 0.068 &  & -0.002 & 0.067 &  & -0.010 & 0.110 &  & 0.009 & 0.317 \\ 
&  & $\mu_{22}=3$ &  & -0.230 & 0.108 &  & -0.023 & 0.071 &  & -0.000 & 0.071 &  & -0.047 & 0.154 &  & -0.009 & 0.373 \\ 
   \bottomrule
\end{tabular}
}
\end{table}

Finally, \tablename~\ref{tab:TABe} reports the results under scenario~\ref{item:m5}, that is when there is the 5\% of bad points on the background of the bulk of the data.
For the way the outliers are added, the NCM should be the best performer; instead, the best performance is for the CNM, regardless of both the overlap and the sample size.
\begin{table}[!ht]
\caption{
Scenario~\ref{item:m5}: Simulation results on $1,000$ replications.
\label{tab:TABe}
}
\centering
\resizebox{\textwidth}{!}{
\begin{tabular}{lcll rr l rr l rr l rr l rr}
  \toprule
 & & & & \multicolumn{2}{c}{NM} & & \multicolumn{2}{c}{$t$M} & & \multicolumn{2}{c}{CNM} & & \multicolumn{2}{c}{NUM} & & \multicolumn{2}{c}{NCM} \\
 & $n$ & & &	BIAS	&	STD &	&	BIAS	&	STD &	&	BIAS	&	STD &	&	BIAS	&	STD &	&	BIAS	&	STD \\ 
  \midrule
Far & 100 & $\pi_1=0.3$ &  & 0.030 & 0.095 &  & 0.026 & 0.089 &  & 0.014 & 0.059 &  & 0.019 & 0.075 &  & 0.103 & 0.303 \\ 
&  & $\mu_{11}=0$ &  & -0.014 & 0.419 &  & 0.004 & 0.247 &  & 0.005 & 0.246 &  & 0.009 & 0.365 &  & 0.026 & 0.455 \\ 
&  & $\mu_{21}=-3$ &  & -0.029 & 0.621 &  & 0.043 & 0.516 &  & -0.003 & 0.235 &  & 0.078 & 0.548 &  & 0.191 & 1.007 \\ 
&  & $\mu_{12}=0$ &  & -0.070 & 0.941 &  & -0.012 & 0.870 &  & 0.004 & 0.150 &  & 0.021 & 0.175 &  & 0.032 & 1.637 \\ 
&  & $\mu_{22}=3$ &  & 0.083 & 0.563 &  & -0.020 & 0.344 &  & -0.008 & 0.165 &  & 0.023 & 0.168 &  & -0.064 & 1.220 \\[2mm] 
& 200 & $\pi_1=0.3$ &  & 0.034 & 0.086 &  & 0.029 & 0.084 &  & 0.012 & 0.036 &  & 0.010 & 0.059 &  & 0.086 & 0.306 \\ 
&  & $\mu_{11}=0$ &  & -0.058 & 0.260 &  & -0.004 & 0.151 &  & -0.005 & 0.141 &  & -0.011 & 0.282 &  & -0.010 & 0.218 \\ 
&  & $\mu_{21}=-3$ &  & 0.108 & 0.619 &  & 0.058 & 0.555 &  & -0.007 & 0.149 &  & 0.137 & 0.497 &  & 0.144 & 0.838 \\ 
&  & $\mu_{12}=0$ &  & -0.059 & 0.856 &  & -0.073 & 0.911 &  & 0.005 & 0.097 &  & 0.029 & 0.113 &  & -0.044 & 1.197 \\ 
&  & $\mu_{22}=3$ &  & 0.076 & 0.363 &  & -0.009 & 0.347 &  & 0.001 & 0.094 &  & 0.033 & 0.110 &  & 0.003 & 0.838 \\[2mm]
& 500 & $\pi_1=0.3$ &  & 0.045 & 0.079 &  & 0.028 & 0.057 &  & 0.010 & 0.024 &  & 0.008 & 0.057 &  & 0.096 & 0.305 \\ 
&  & $\mu_{11}=0$ &  & -0.040 & 0.172 &  & 0.004 & 0.092 &  & 0.008 & 0.086 &  & -0.005 & 0.239 &  & 0.002 & 0.085 \\ 
&  & $\mu_{21}=-3$ &  & 0.353 & 0.648 &  & 0.035 & 0.365 &  & -0.004 & 0.089 &  & 0.250 & 0.444 &  & 0.025 & 0.334 \\ 
&  & $\mu_{12}=0$ &  & -0.021 & 0.882 &  & -0.021 & 0.592 &  & 0.001 & 0.055 &  & 0.034 & 0.070 &  & -0.037 & 0.525 \\ 
&  & $\mu_{22}=3$ &  & 0.045 & 0.449 &  & -0.001 & 0.188 &  & 0.001 & 0.054 &  & 0.046 & 0.069 &  & -0.012 & 0.441 \\[4mm]
Close & 100 & $\pi_1=0.3$ &  & 0.141 & 0.241 &  & 0.051 & 0.120 &  & 0.032 & 0.107 &  & 0.094 & 0.116 &  & 0.033 & 0.304 \\ 
&  & $\mu_{11}=0$ &  & -0.015 & 0.576 &  & 0.026 & 0.314 &  & 0.019 & 0.361 &  & -0.047 & 0.423 &  & 0.031 & 0.390 \\ 
&  & $\mu_{21}=-1$ &  & 0.577 & 1.212 &  & 0.100 & 0.582 &  & 0.047 & 0.561 &  & 0.660 & 0.794 &  & 0.530 & 1.199 \\ 
&  & $\mu_{12}=0$ &  & -0.055 & 1.658 &  & -0.031 & 1.079 &  & 0.008 & 0.522 &  & 0.085 & 0.334 &  & -0.002 & 2.072 \\ 
&  & $\mu_{22}=3$ &  & -0.462 & 1.919 &  & -0.076 & 0.796 &  & -0.089 & 0.703 &  & 0.013 & 0.458 &  & -0.338 & 2.204 \\[2mm]
& 200 & $\pi_1=0.3$ &  & 0.111 & 0.172 &  & 0.040 & 0.092 &  & 0.016 & 0.044 &  & 0.102 & 0.090 &  & 0.055 & 0.308 \\ 
&  & $\mu_{11}=0$ &  & -0.046 & 0.374 &  & 0.007 & 0.170 &  & 0.009 & 0.187 &  & -0.112 & 0.247 &  & -0.007 & 0.166 \\ 
&  & $\mu_{21}=-1$ &  & 0.669 & 0.771 &  & 0.074 & 0.416 &  & 0.007 & 0.193 &  & 0.736 & 0.607 &  & 0.279 & 0.839 \\ 
&  & $\mu_{12}=0$ &  & -0.112 & 1.315 &  & 0.022 & 0.908 &  & 0.005 & 0.104 &  & 0.088 & 0.216 &  & -0.100 & 1.777 \\ 
&  & $\mu_{22}=3$ &  & -0.162 & 1.111 &  & -0.044 & 0.596 &  & -0.007 & 0.133 &  & 0.078 & 0.388 &  & -0.195 & 1.667 \\[2mm]
& 500 & $\pi_1=0.3$ &  & 0.117 & 0.093 &  & 0.047 & 0.091 &  & 0.015 & 0.025 &  & 0.125 & 0.071 &  & 0.089 & 0.306 \\ 
&  & $\mu_{11}=0$ &  & -0.045 & 0.202 &  & -0.008 & 0.101 &  & -0.001 & 0.098 &  & -0.161 & 0.154 &  & -0.012 & 0.097 \\ 
&  & $\mu_{21}=-1$ &  & 0.800 & 0.400 &  & 0.083 & 0.411 &  & -0.002 & 0.107 &  & 0.884 & 0.488 &  & 0.110 & 0.528 \\ 
&  & $\mu_{12}=0$ &  & -0.056 & 0.891 &  & -0.106 & 0.980 &  & -0.000 & 0.060 &  & 0.115 & 0.100 &  & -0.029 & 1.105 \\ 
&  & $\mu_{22}=3$ &  & -0.011 & 0.589 &  & -0.014 & 0.532 &  & -0.005 & 0.067 &  & 0.146 & 0.137 &  & -0.152 & 1.367 \\ 
   \bottomrule
\end{tabular}
}
\end{table}



\subsection{Classification performance}

\tablename~\ref{tab:classification} summarizes the obtained average misclassification rates.
\begin{table}[!ht]
\caption{
Average misclassification rates.
Values refer to averages across 1,000 replications.
\label{tab:classification}
}
\centering
\resizebox{\textwidth}{!}{
\begin{tabular}{cr c cc c cc c cc c cc c cc}
\toprule
	&	&	&	\multicolumn{2}{c}{NM}	&&	\multicolumn{2}{c}{$t$M}	&&	\multicolumn{2}{c}{CNM} &&	\multicolumn{2}{c}{NUM}	&&	\multicolumn{2}{c}{NCM}	\\
		\cline{4-5}\cline{7-8}\cline{10-11}\cline{13-14}\cline{16-17}
	&	\multicolumn{1}{c}{$n$}	 &	&	Far	&	Close	&	&	Far	&	Close	&	&	Far	&	Close	&	&	Far	&	Close	&	&	Far	&	Close	\\
\midrule
Scenario~\ref{item:m1}	&	100	&	&	0.002	&	0.025	&	&	0.002	&	0.025	&	&	0.002	&	0.025	&	&	0.086	&	0.078	&	&	0.002	&	0.038	\\
	&	200	&	&	0.002	&	0.023	&	&	0.002	&	0.023	&	&	0.002	&	0.023	&	&	0.045	&	0.054	&	&	0.001	&	0.023	\\
	&	500	&	&	0.001	&	0.021	&	&	0.001	&	0.021	&	&	0.001	&	0.021	&	&	0.018	&	0.037	&	&	0.001	&	0.021	\\
	&		&	&		&		&	&		&		&	&		&		&	&		&		&	&		&		\\
Scenario~\ref{item:m2}	&	100	&	&	0.018	&	0.071	&	&	0.021	&	0.068	&	&	0.020	&	0.067	&	&	0.043	&	0.089	&	&	0.008	&	0.059	\\
	&	200	&	&	0.021	&	0.075	&	&	0.023	&	0.071	&	&	0.023	&	0.072	&	&	0.027	&	0.085	&	&	0.008	&	0.072	\\
	&	500	&	&	0.022	&	0.080	&	&	0.024	&	0.084	&	&	0.024	&	0.085	&	&	0.024	&	0.099	&	&	0.012	&	0.083	\\
	&		&	&		&		&	&		&		&	&		&		&	&		&		&	&		&		\\
Scenario~\ref{item:m3}	&	100	&	&	0.091	&	0.240	&	&	0.036	&	0.069	&	&	0.035	&	0.067	&	&	0.102	&	0.195	&	&	0.059	&	0.118	\\
	&	200	&	&	0.090	&	0.273	&	&	0.034	&	0.061	&	&	0.033	&	0.060	&	&	0.099	&	0.205	&	&	0.036	&	0.073	\\
	&	500	&	&	0.086	&	0.293	&	&	0.033	&	0.058	&	&	0.030	&	0.054	&	&	0.095	&	0.213	&	&	0.022	&	0.051	\\
	&		&	&		&		&	&		&		&	&		&		&	&		&		&	&		&		\\
Scenario~\ref{item:m4}	&	100	&	&	0.005	&	0.040	&	&	0.002	&	0.026	&	&	0.002	&	0.026	&	&	0.057	&	0.070	&	&	0.004	&	0.043	\\
	&	200	&	&	0.005	&	0.055	&	&	0.002	&	0.025	&	&	0.002	&	0.023	&	&	0.021	&	0.050	&	&	0.008	&	0.048	\\
	&	500	&	&	0.004	&	0.061	&	&	0.002	&	0.024	&	&	0.001	&	0.022	&	&	0.011	&	0.039	&	&	0.003	&	0.028	\\
	&		&	&		&		&	&		&		&	&		&		&	&		&		&	&		&		\\
Scenario~\ref{item:m5}	&	100	&	&	0.006	&	0.077	&	&	0.006	&	0.033	&	&	0.003	&	0.036	&	&	0.024	&	0.073	&	&	0.023	&	0.090	\\
	&	200	&	&	0.007	&	0.066	&	&	0.005	&	0.028	&	&	0.002	&	0.025	&	&	0.025	&	0.076	&	&	0.012	&	0.051	\\
	&	500	&	&	0.011	&	0.050	&	&	0.004	&	0.028	&	&	0.001	&	0.022	&	&	0.037	&	0.080	&	&	0.003	&	0.031	\\
\bottomrule	
\end{tabular}
}
\end{table}
Misclassification rates are computed via the \texttt{classError()} function of the \texttt{mclust} package for {\sf R} \citep{Fral:Raft:Murp:Scru:mclu:2012}. 
Under scenarios~\ref{item:m4} and \ref{item:m5}, misclassification rates are computed only with respect to the true good observations; for the NCM only, under all of the considered scenarios, the computation of the misclassification rates is further restricted to the observations which are not assigned, via the MAP operator, to the noise component of the model.
Under scenario~\ref{item:m1}, in the far case, NM, $t$M, CNM, and NCM show similar misclassification rates, while misclassification rates from NUM are greater.
In the close case, NM, $t$M, CNM, and NCM provide analogous results when the sample size is 200 or 500, while NCM gives a slightly greater misclassification rate (0.038) when the sample size is 100.
Regardless of the considered sample size, NUM has the worst performance.  
Under scenario~\ref{item:m2}, in the far case, NM, $t$M, and CNM show similar misclassification rates.
As concerns the remaining models, NCM has the best performance while NUM the worst.
However, the best performance for NCM could be related to the fact that misclassification rates are computed only over the observations classified as good by the model; this means that ``problematic'' observations in terms of classification (i.e., observations having a similar probability to belong to the two clusters) could be removed from this computation because assigned to the noise component.
In the close case, NM, $t$M, CNM, and NCM provide analogous results when the sample size is 200 or 500, while NCM gives a slightly lower misclassification rate (0.059) when the sample size is 100.
Regardless of the considered sample size, NUM has the worst performance.  
Under scenario~\ref{item:m3}, regardless of both the overlap between clusters and the sample size, the lowest misclassification rates are obtained (apart from the case $n=500$) for CNM, followed by $t$M which provides similar results. NUM provides the worst results in the far case, while NM gives the worst misclassification rates in the close case. 
Under scenarios~\ref{item:m4} and \ref{item:m5}, CNM provides almost always the best results, followed by $t$M.
It is interesting to note how CNM works better than NCM under scenarios~\ref{item:m5}, which should be the best scenario for NCM.
Also in this case, NUM does not provide good results, especially for the far case if compared to the competing models. 



\subsection{Outlier detection}
\label{subsec:Outliers detection}

We now compare the performance of $t$Ms, CNMs, NUMs, and NCMs in detecting outliers.
While the MAP operator is adopted to detect outliers for CNMs (cf.~Section~\ref{subsec:Automatic detection of noise}), NUMs, and NCMs, for $t$Ms the \textit{a~posteriori} procedure illustrated by \citet[][p.~232]{McLa:Peel:fini:2000}, and summarized at the end of Section~\ref{subsec:Automatic detection of noise}, is considered (with the 95th percentile).

For the purpose of evaluation of the performance of the competing models in detecting outliers, we report the true positive rate (TPR), measuring the proportion of bad points that are correctly identified as bad points, and the false positive rate (FPR), corresponding to the proportion of good points incorrectly classified as bad points. 
\tablename~\ref{tab:TPR and FPR} reports these measures for scenarios~\ref{item:m4} and \ref{item:m5}. 
Under scenario~\ref{item:m4}, $t$Ms and CNMs show the highest (almost optimal) TPRs, but CNM gives lower (almost optimal) FPRs.
The remaining approaches are outperformed by the CNM both in terms of TPRs and FPRs.
Under scenario~\ref{item:m5}, $t$M gives the highest TPRs.
However, this is counterbalanced by higher FPRs.  
In other words, with the selected percentile, the detection rule for $t$M tends to declare more observations as outliers, but these detected outliers are sometimes not true outliers.
If the aim is to remove from the sample the detected outliers, the practical consequence of these results is that, if we use the detection rule from $t$Ms (with the classical percentile we considered), then we are induced to also remove some good observations with a consequent loss of information.
Apart from this consideration, the detection rule from $t$Ms needs the specification of a percentile and the simulation results we report show how this choice is not so obvious.
On the contrary, the detection rule for CNM provides almost optimal results in terms of FPRs, being their values always close to zero.
The fact that the TPRs do not approach at one is not necessarily an error: the way the outliers are inserted into the data makes possible that some of them will have values related to good points and, as such, these points will be detected as good points by our model.
Apart from this consideration, the detection rule from $t$-based models needs the specification of a percentile and the simulation results we report show how this choice is not so obvious.
With respect to the remaining approaches, regardless of the considered scenario, NCM works better than NUM but worse than $t$M and CNM.
\begin{table}[!ht]
\caption{
Values of TPRs and FPRs; they refers to rates across 1,000 replications.
\label{tab:TPR and FPR}
}
\centering
\begin{tabular}{ccc c cc c cc c cc c cc}
\toprule
	&		&		&	&	\multicolumn{2}{c}{$t$M} &	&	\multicolumn{2}{c}{CNM}	 &	&	\multicolumn{2}{c}{NUM}	&	&	\multicolumn{2}{c}{NCM}			\\
	\cline{5-6}\cline{8-9}\cline{11-12}\cline{14-15}	
	&	Overlap	&	\multicolumn{1}{c}{$n$}	&	&	TPR	&	FPR	&	&	TPR	&	FPR	&	&	TPR	&	FPR	&	&	TPR	&	FPR	\\
\midrule	
Scenario~\ref{item:m4}	&	Far	&	100	&	&	1.000	&	0.070	&	&	1.000	&	0.001	&	&	0.976	&	0.124	&	&	0.966	&	0.096	\\	
	&		&	200	&	&	1.000	&	0.077	&	&	1.000	&	0.001	&	&	0.977	&	0.042	&	&	0.995	&	0.049	\\	
	&		&	500	&	&	1.000	&	0.079	&	&	1.000	&	0.000	&	&	0.971	&	0.018	&	&	1.000	&	0.006	\\[1.5mm]
	&	Close	&	100	&	&	0.995	&	0.059	&	&	0.995	&	0.002	&	&	0.970	&	0.098	&	&	0.967	&	0.111	\\	
	&		&	200	&	&	1.000	&	0.068	&	&	1.000	&	0.002	&	&	0.971	&	0.030	&	&	0.991	&	0.038	\\	
	&		&	500	&	&	1.000	&	0.070	&	&	1.000	&	0.001	&	&	0.977	&	0.007	&	&	1.000	&	0.005	\\	
	&		&		&	&		&		&	&		&		&	&		&		&	&		&		\\	
Scenario~\ref{item:m5}	&	Far	&	100	&	&	0.912	&	0.077	&	&	0.829	&	0.010	&	&	0.695	&	0.045	&	&	0.826	&	0.048	\\	
	&		&	200	&	&	0.920	&	0.074	&	&	0.833	&	0.006	&	&	0.626	&	0.042	&	&	0.824	&	0.011	\\	
	&		&	500	&	&	0.923	&	0.069	&	&	0.839	&	0.002	&	&	0.589	&	0.054	&	&	0.834	&	0.002	\\[1.5mm]
	&	Close	&	100	&	&	0.908	&	0.068	&	&	0.804	&	0.012	&	&	0.694	&	0.030	&	&	0.805	&	0.027	\\	
	&		&	200	&	&	0.916	&	0.062	&	&	0.854	&	0.006	&	&	0.697	&	0.014	&	&	0.828	&	0.002	\\	
	&		&	500	&	&	0.920	&	0.055	&	&	0.859	&	0.002	&	&	0.694	&	0.007	&	&	0.847	&	0.001	\\	
\bottomrule	
\end{tabular}
\end{table}

\section{Data analyses}
\label{sec:Data analysis}

In this section, we will evaluate the performance of the 14 parsimonious CNM models on artificial and real data sets.
Particular attention will be devoted to the problem of detecting bad points.
A comparison with parsimonious families of NMs, $t$Ms, NUMs, and NCMs, will be also provided.
These families are implemented by functions and packages already discussed in Section~\ref{sec:Numerical studies}.
All the EM-based algorithms used to fit these models are initialized as explained in Section~\ref{sec:Numerical studies}.
As concerns the family of parsimonious NUMs and NCMs, based on the \textsf{R} functions used, only a subset of 10 of the 14 parsimonious structures in \tablename~\ref{tab:models} can be implemented; they are: EII, VII, EEI, VEI, EVI, VVI, EEE, EEV, VEV, and VVV.

\subsection{Artificial data with uniform noise}
\label{subsec:Overall uniform noise}

In this first analysis, a sample of $n=180$ simulated bivariate points is generated from an EEE-NM model with $G=2$ clusters of equal size ($n_1=n_2=90$).
Twenty noise points are also added from a uniform distribution over the range $-10$ to $10$ on each variate; hence, the generated data can be meant as arising from an EEE-NCM with $G=2$ clusters. 
Note that when a point from this uniform distribution effectively falls inside a cluster, which seems to happen five times (see \figurename~\ref{fig:simdataunif1}), we would expect it to be classified as belonging to the associated cluster.
\begin{figure}[!ht]
  \centering
\includegraphics[width=0.49\textwidth]{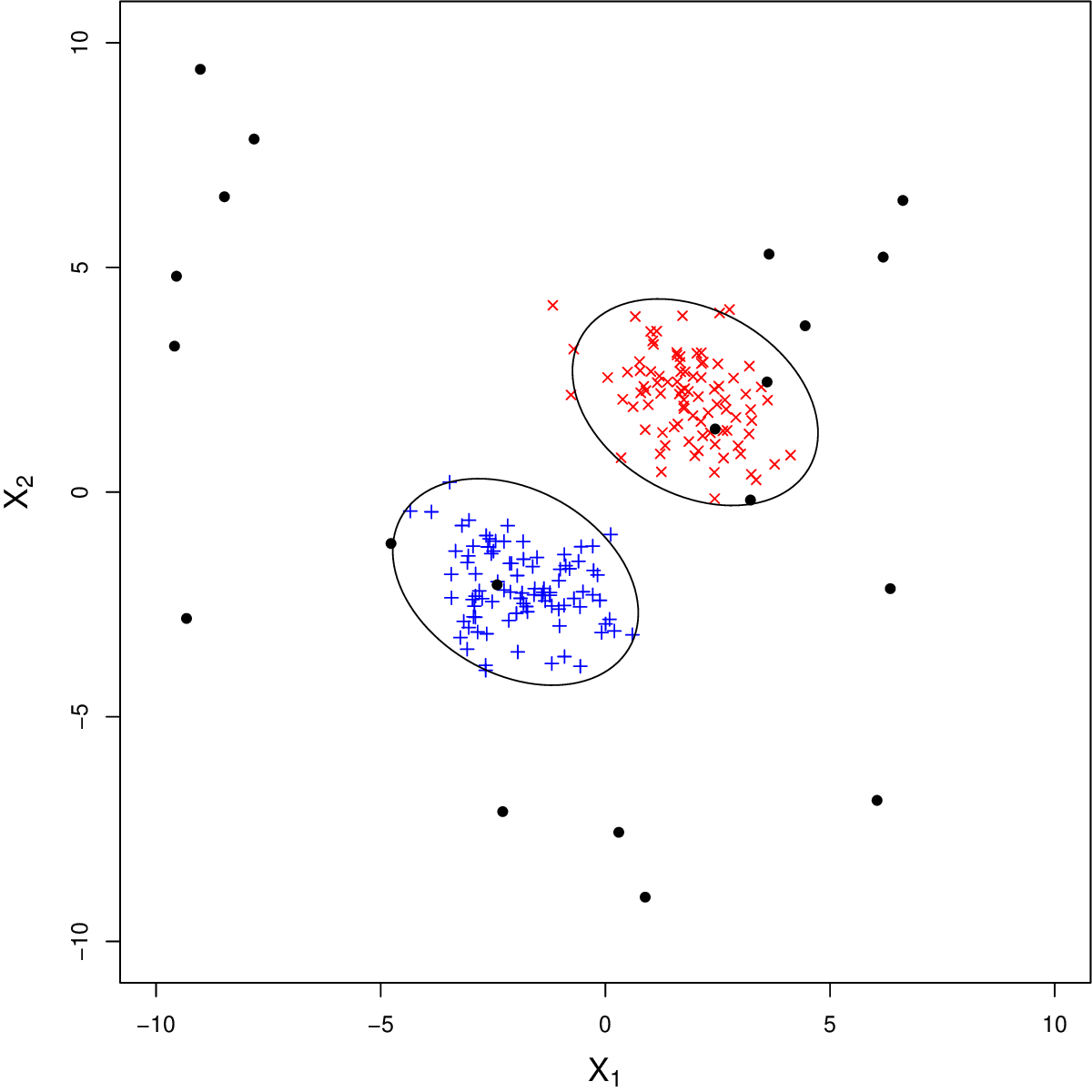}
\caption{Simulated data from Section~\ref{subsec:Overall uniform noise}: Scatterplot where uniform noise points are denoted by~$\bullet$.}
\label{fig:simdataunif1}       
\end{figure}

The competing models are run for $G\in\left\{1,2,3\right\}$.
The corresponding BIC values are reported in \figurename~\ref{fig:noise BIC}.
\begin{figure}[!ht]
\centering
\subfigure[NMs\label{fig:BIC NM}]
{\resizebox{0.322\textwidth}{!}{\includegraphics{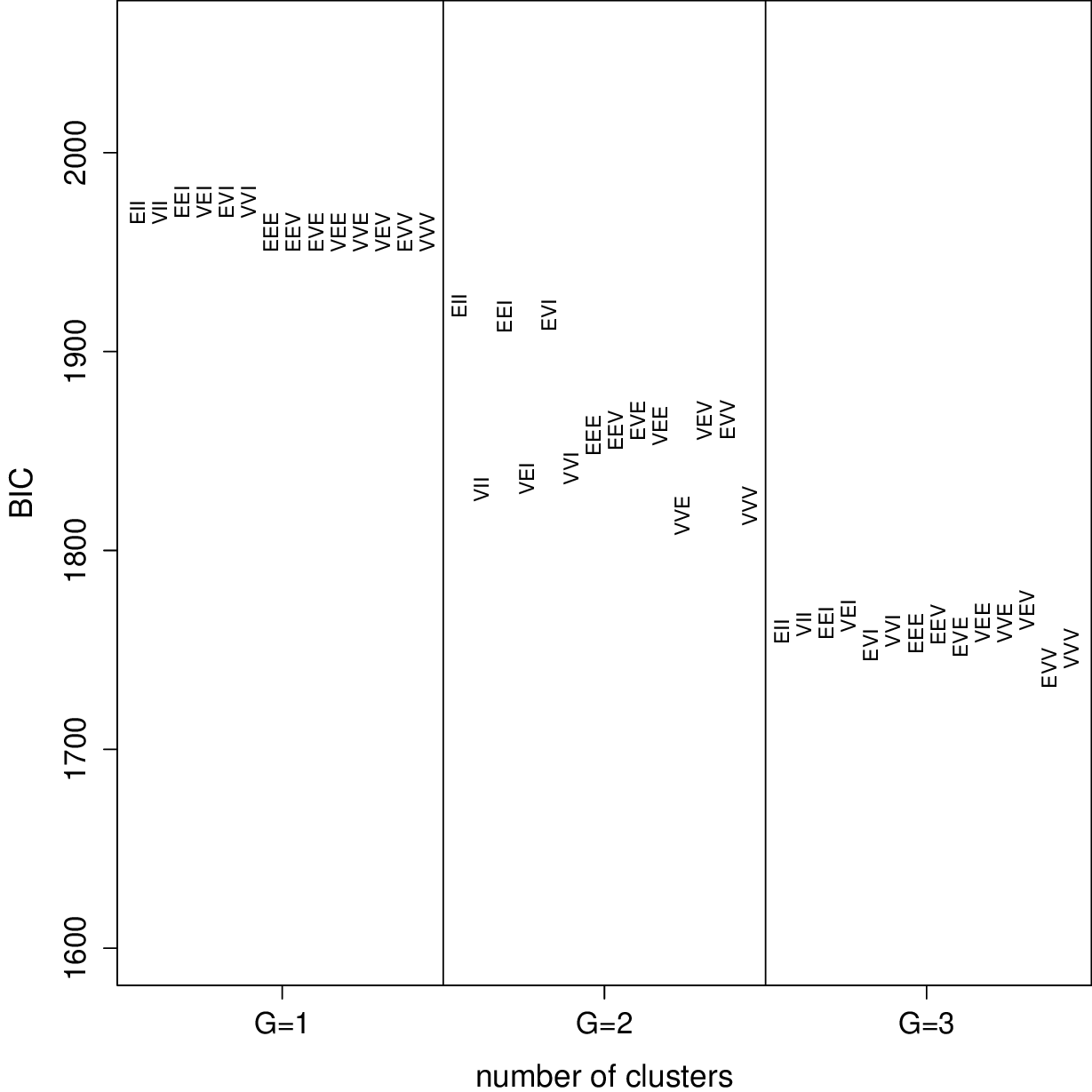}}}
\subfigure[$t$Ms\label{fig:BIC tM}]
{\resizebox{0.322\textwidth}{!}{\includegraphics{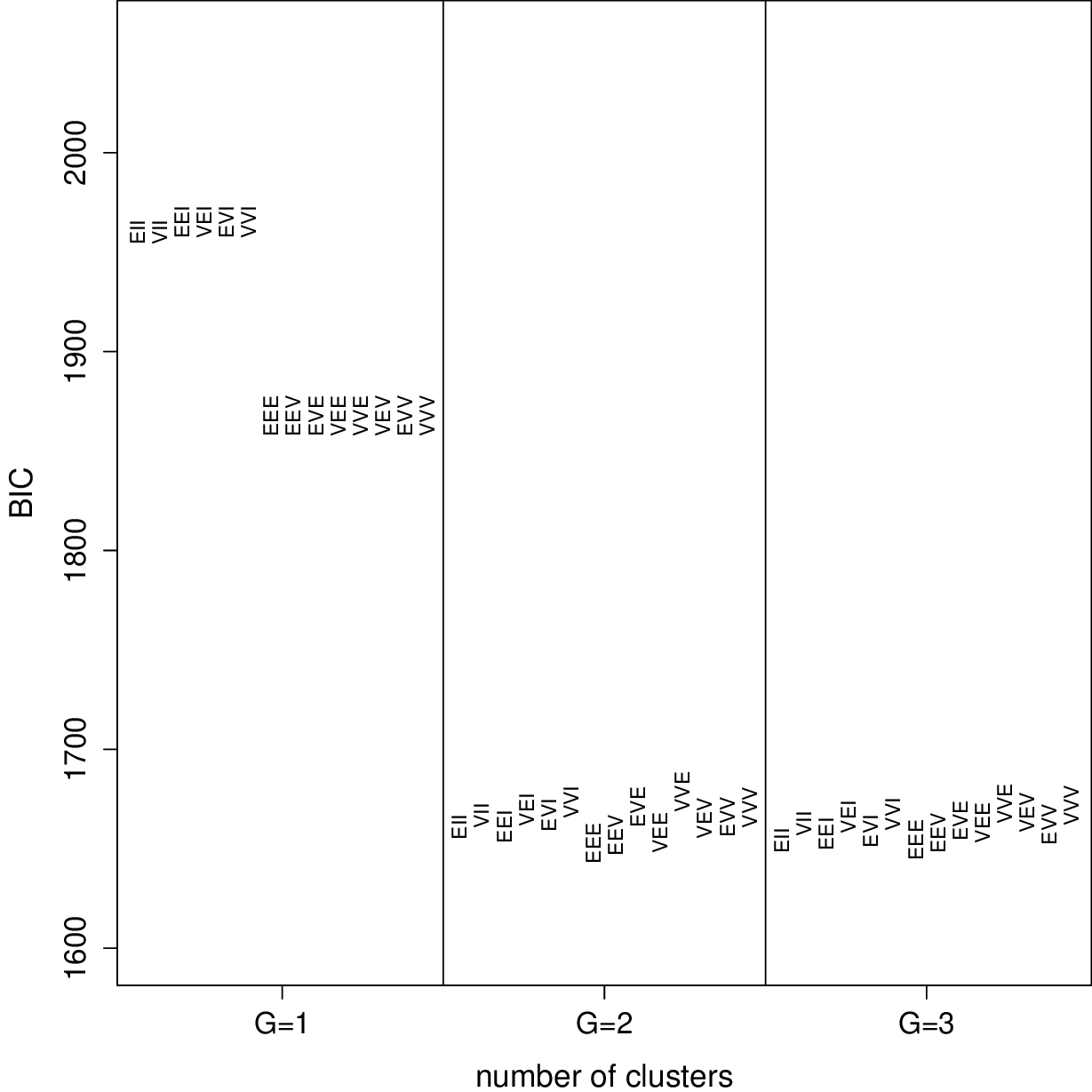}}}
\subfigure[CNMs\label{fig:BIC CNM}]
{\resizebox{0.322\textwidth}{!}{\includegraphics{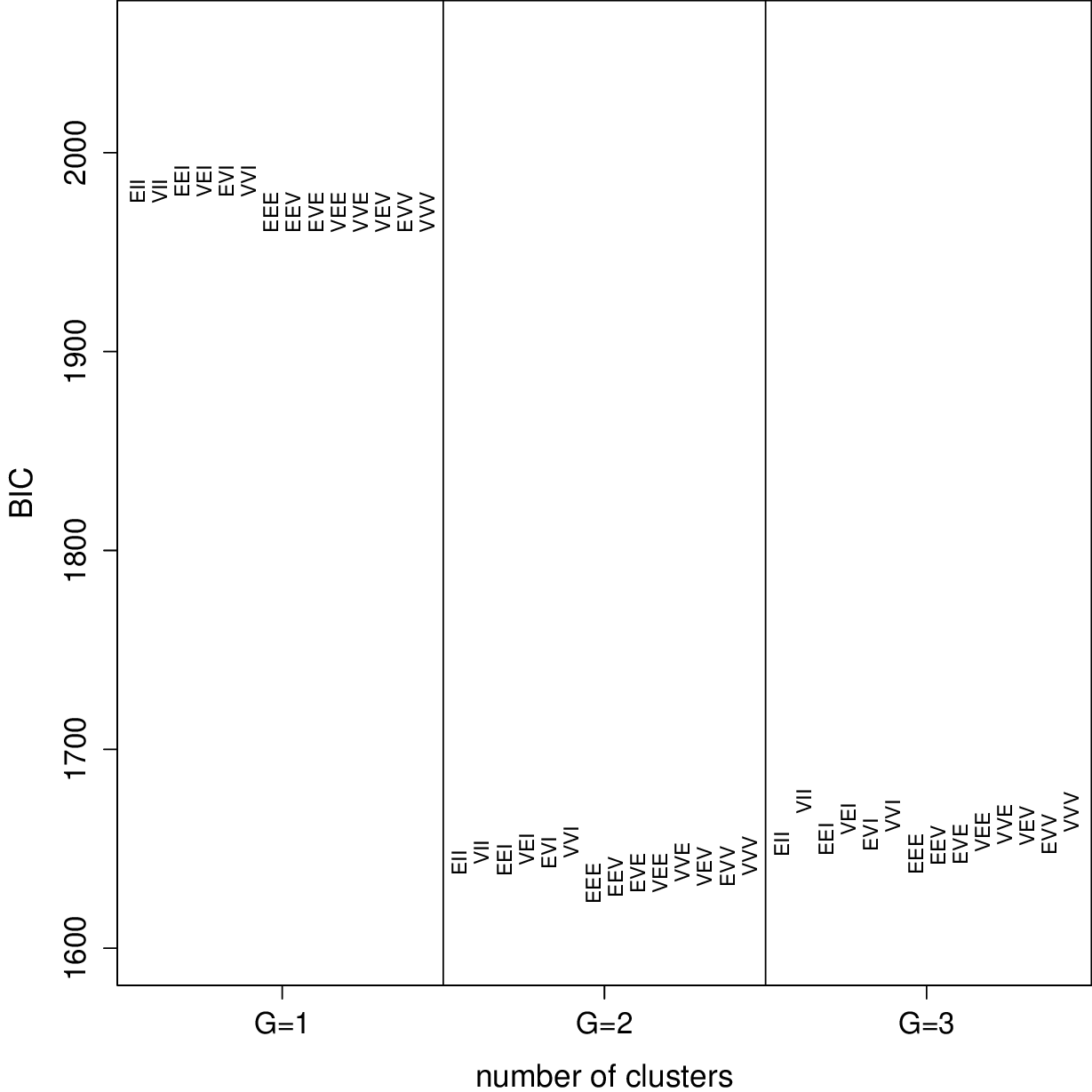}}}
\subfigure[NUMs\label{fig:BIC NUM}]
{\resizebox{0.322\textwidth}{!}{\includegraphics{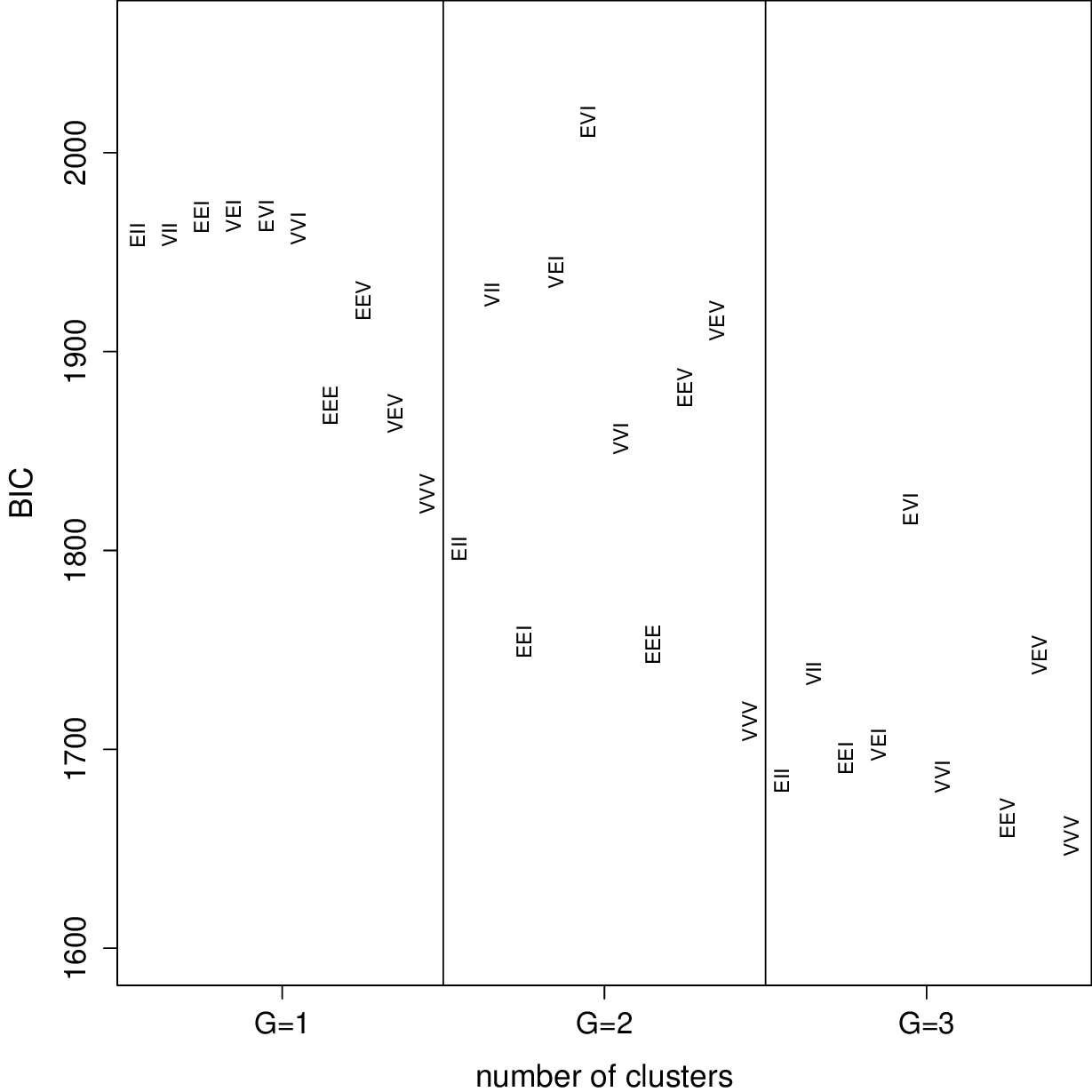}}}
\subfigure[NCMs\label{fig:BIC NCM}]
{\resizebox{0.322\textwidth}{!}{\includegraphics{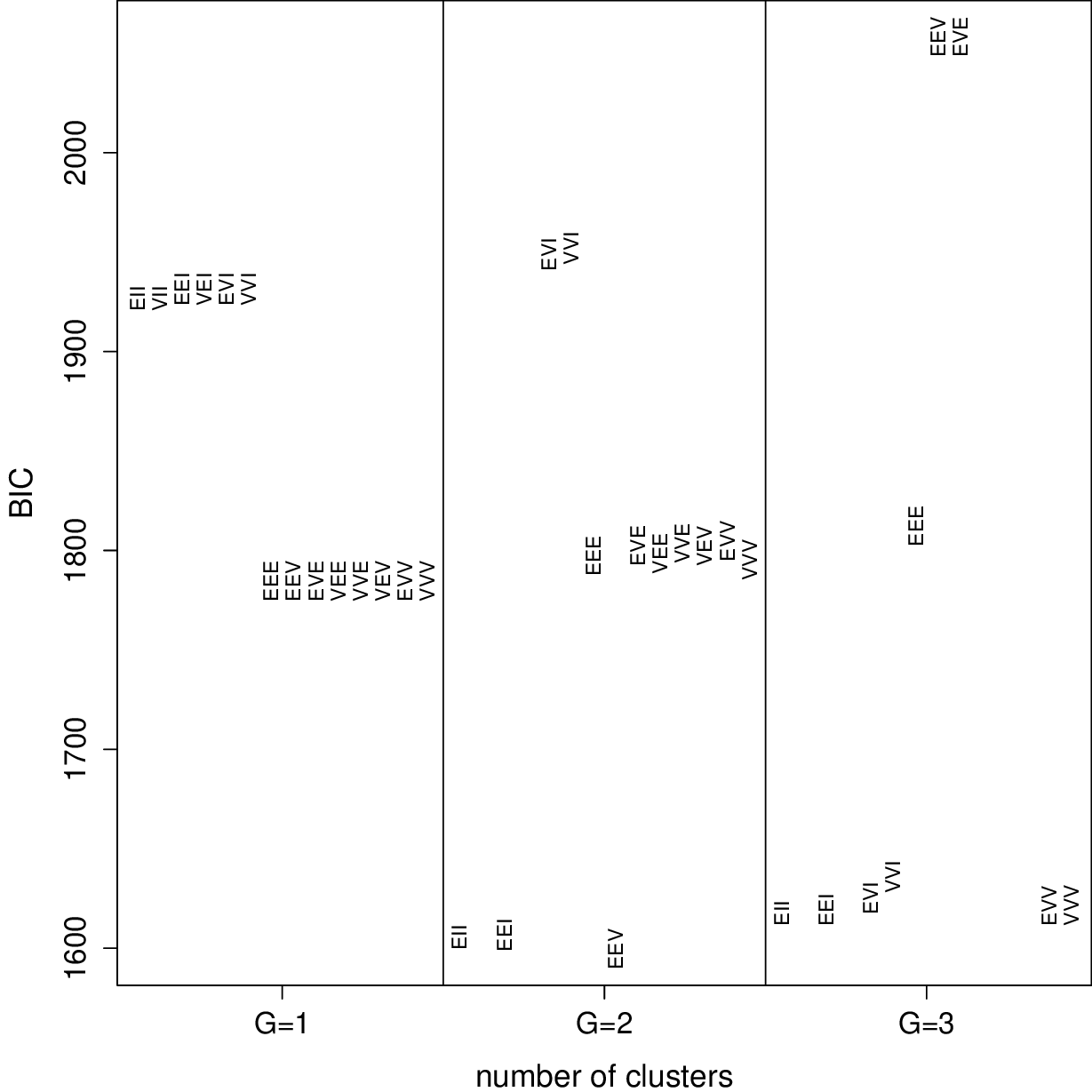}}}
\caption{
Simulated data from Section~\ref{subsec:Overall uniform noise}: BIC values for the fitted models.
\label{fig:noise BIC}
}
\end{figure}
From \figurename~\ref{fig:BIC NM}, the NMs with $G=3$ clusters have the lowest BIC values.
For $t$Ms and CNMs, $G=2$ and $G=3$ clusters provide lower BIC values than $G=1$.
For NUMs, 6 of the parsimonious models with $G=3$ (VVV, EEV, EII, VVI, EEI, and VEI) have the lowest BIC values.
Finally, for NCMs, the best 3 models in terms of BIC have $G=2$ clusters and covariance structures EEV, EEI, and EII.

For each considered family, the best models according to the BIC are graphically represented in \figurename~\ref{fig:noise plot}; for the selected $t$M, CNM, NUM, and NCM, detected outliers are denoted by black bullets. 
\begin{figure}[!ht]
\centering
\subfigure[NM: EVV with $G=3$\label{fig:noise NM}]
{\resizebox{0.322\textwidth}{!}{\includegraphics{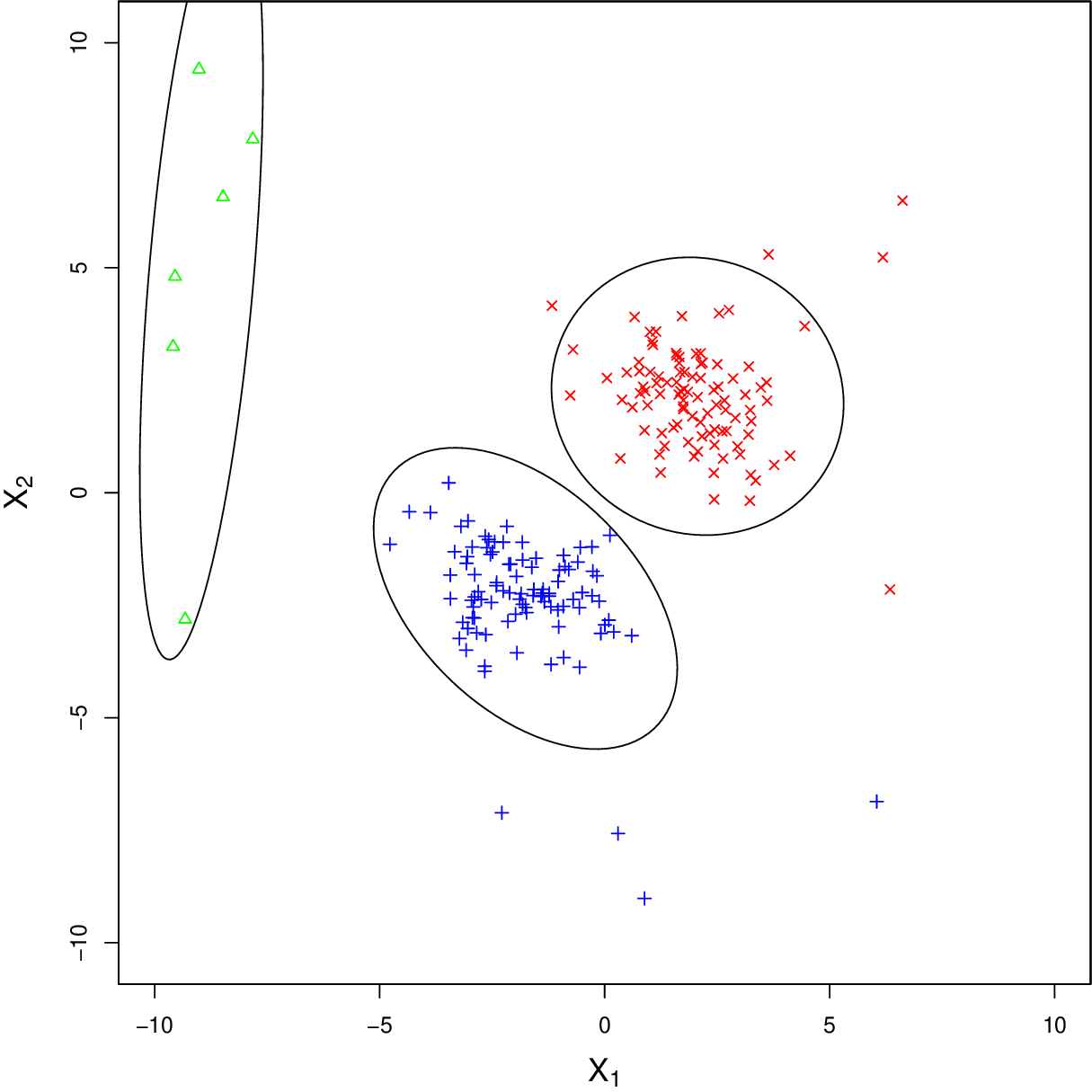}}}
\subfigure[$t$M: EEE with $G=2$\label{fig:noise tM}]
{\resizebox{0.322\textwidth}{!}{\includegraphics{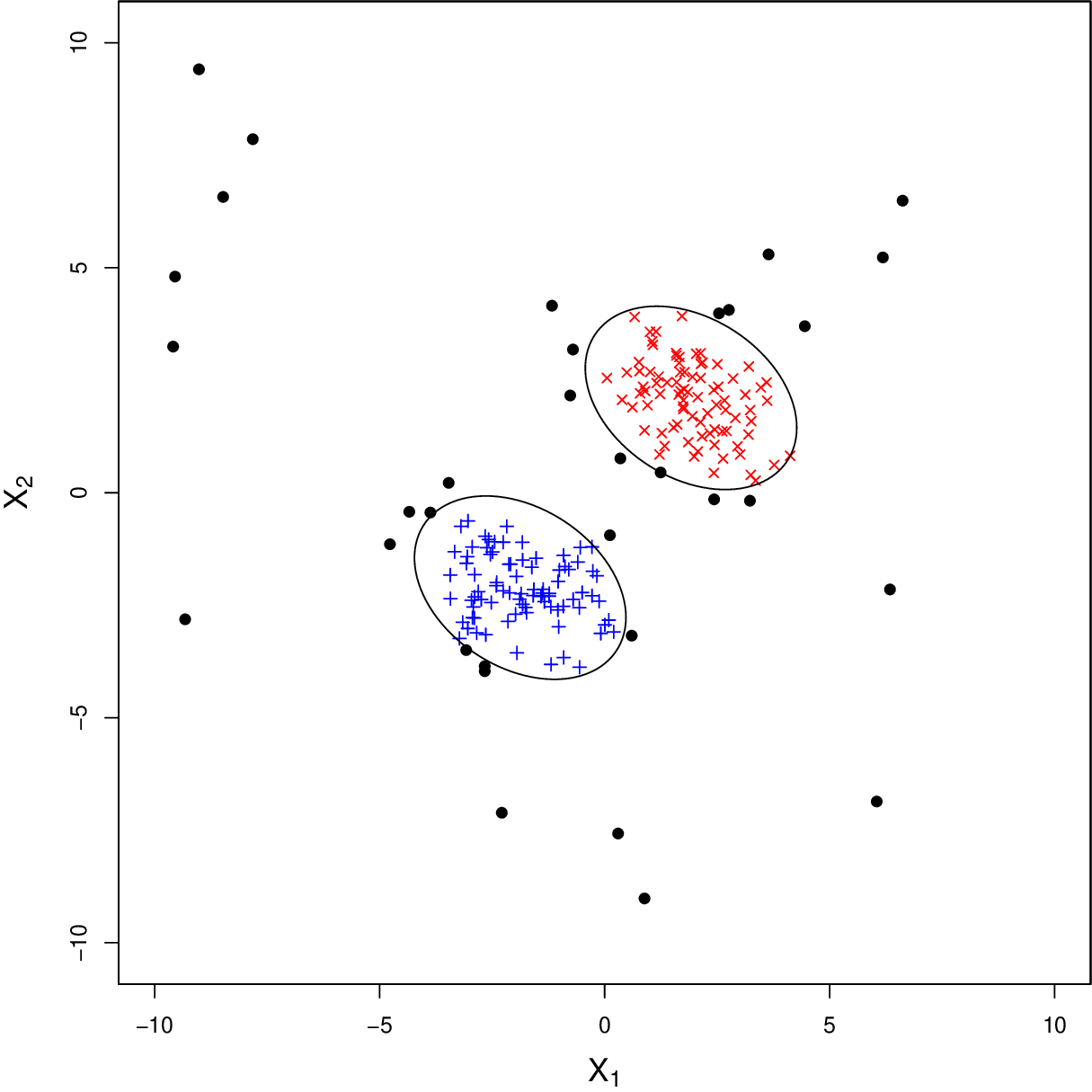}}}
\subfigure[CNM: EEE with $G=2$\label{fig:noise CNM}]
{\resizebox{0.322\textwidth}{!}{\includegraphics{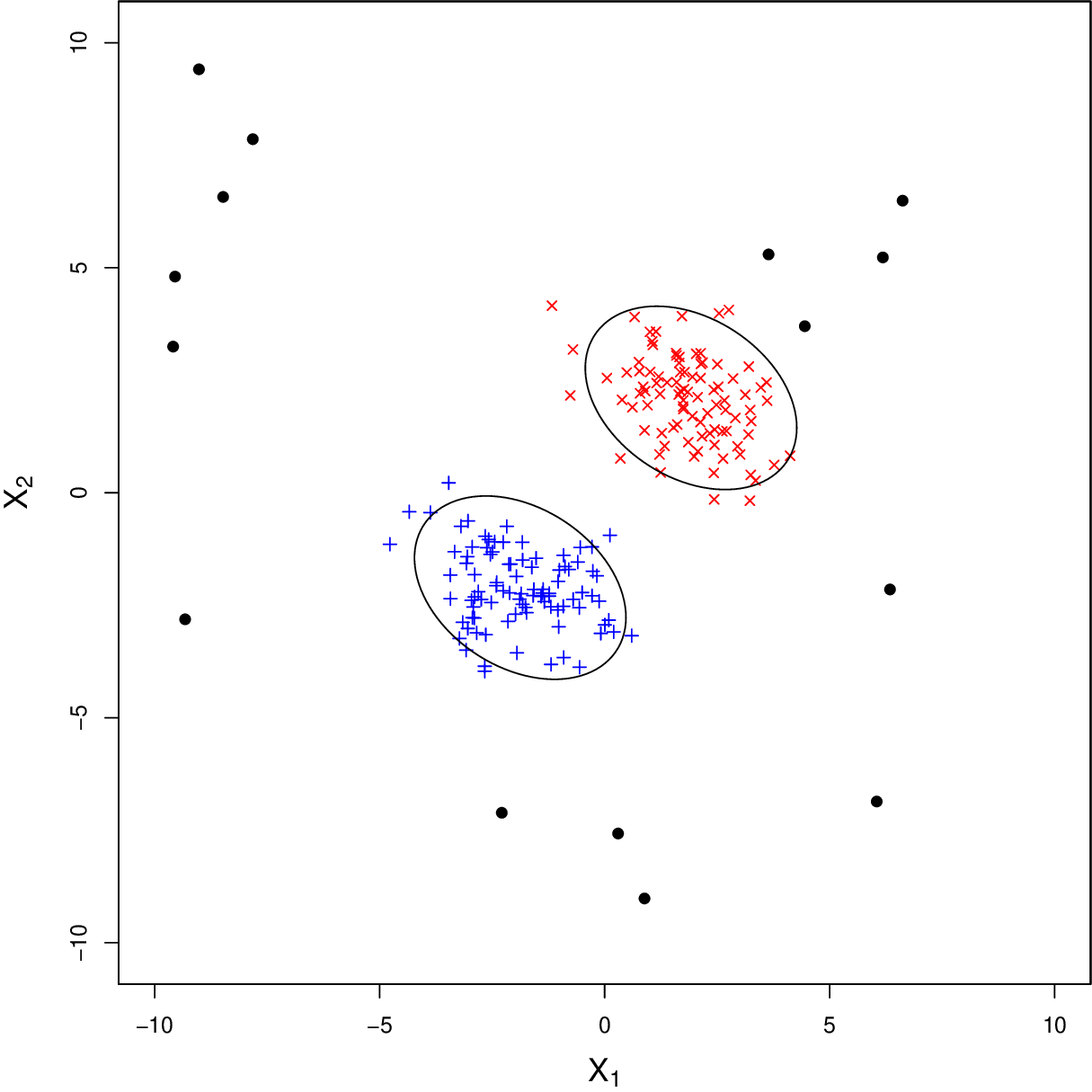}}}
\subfigure[NUM: VVV with $G=3$\label{fig:noise NUM}]
{\resizebox{0.322\textwidth}{!}{\includegraphics{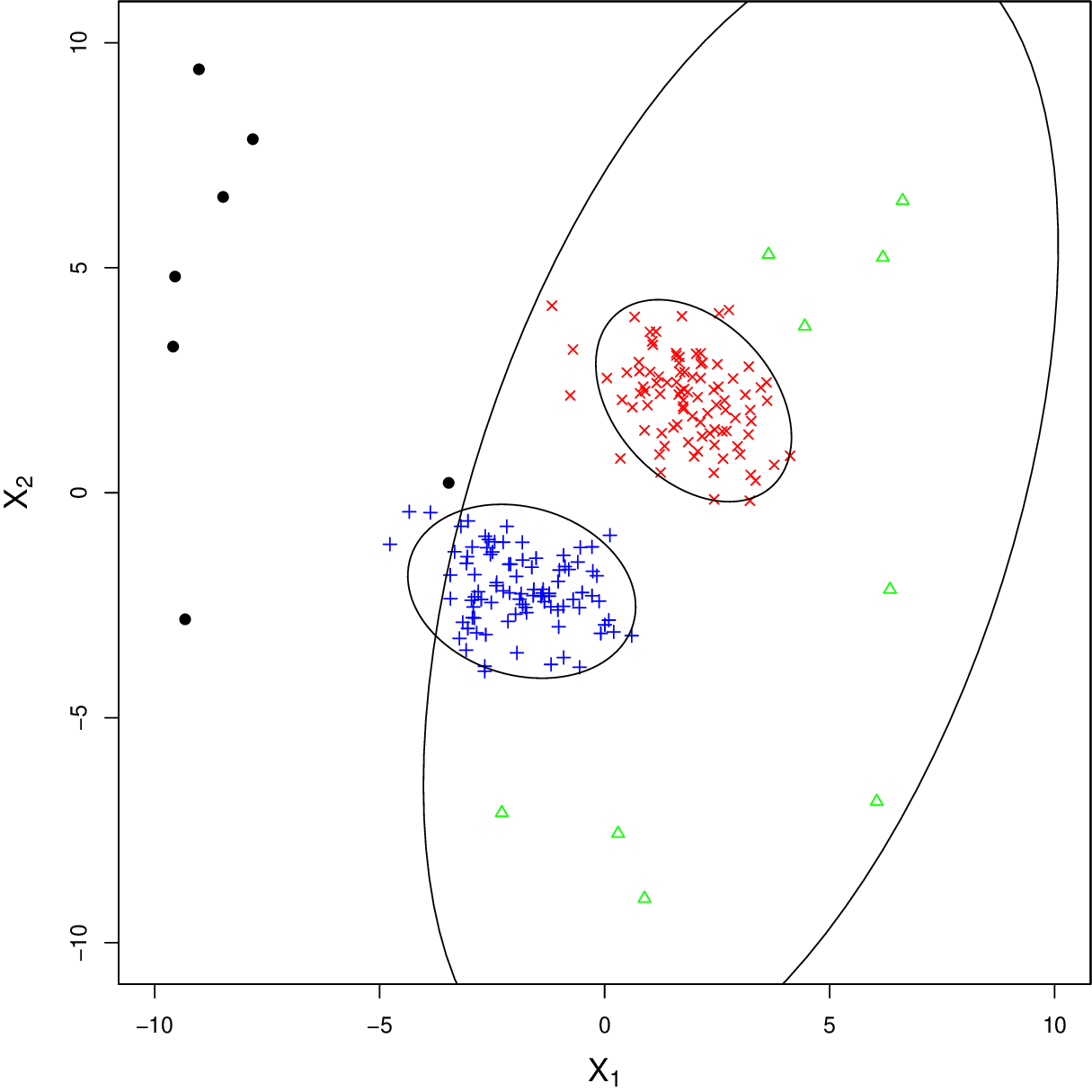}}}
\subfigure[NCM: EEV with $G=2$\label{fig:noise NCM}]
{\resizebox{0.322\textwidth}{!}{\includegraphics{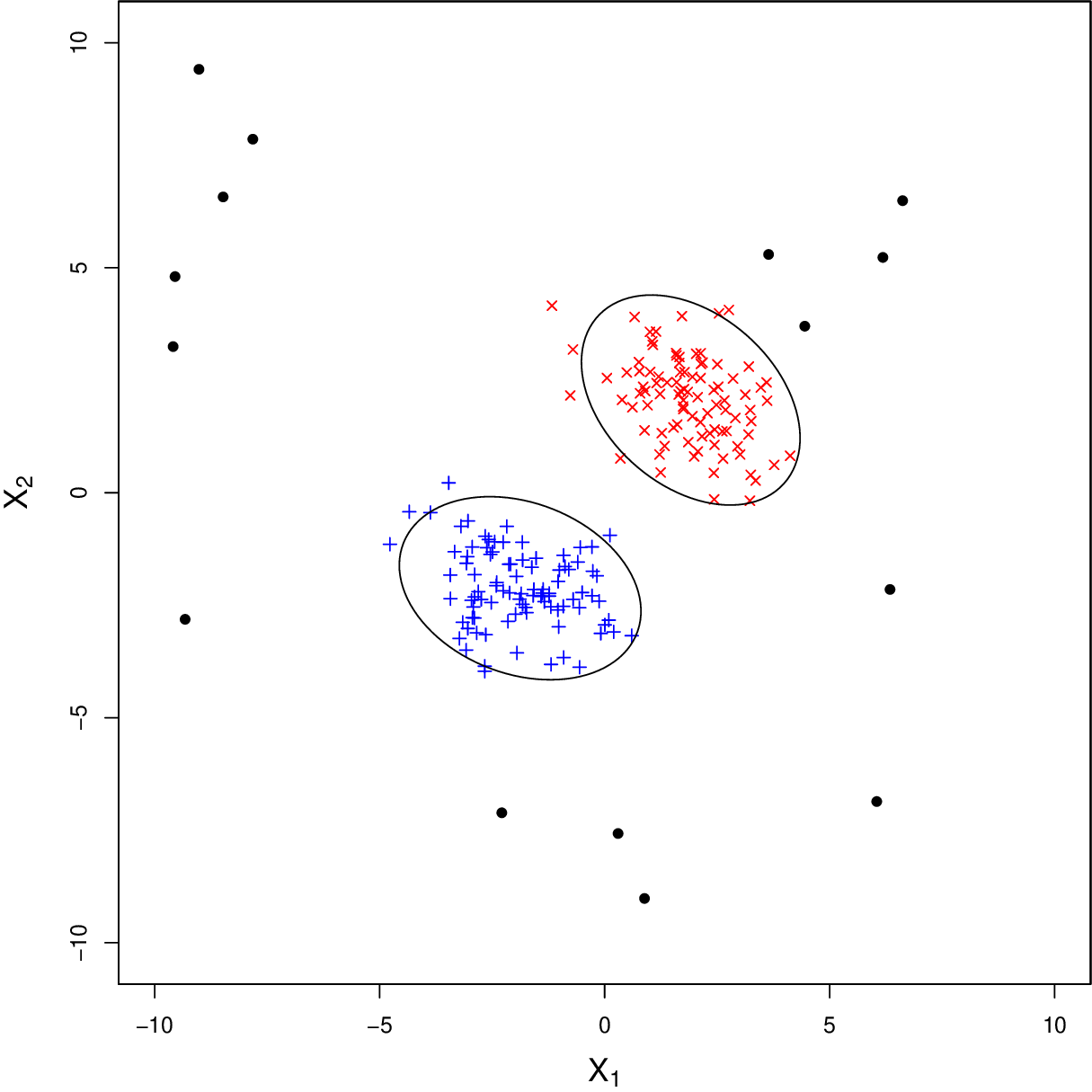}}}
\caption{
Simulated data from Section~\ref{subsec:Overall uniform noise}: scatterplots illustrating the best models, according to the BIC, for each considered family.
Bullets denote detected bad points.
\label{fig:noise plot}
}
\end{figure}
For NMs, the best model according to the BIC has $G=3$ clusters with an EVV covariance structure (\figurename~\ref{fig:noise NM}).
We can note how the additional third cluster is attempting to model part of the background noise; however, the remaining part of the noise is erroneously assigned to the other clusters and this contribute to affect the detection of the underlying EEE structure.
For $t$Ms and CNMs, the best model according to the BIC is the true one, with corresponding clustering represented in \figurename~\ref{fig:noise tM} and \figurename~\ref{fig:noise CNM}, respectively.
However, in conformity with the simulation results of Section~\ref{subsec:Outliers detection}, the detection rule for $t$Ms, based on the 95th percentile, tends to declare more observations as outliers, but these detected outliers are often not true outliers.
The CNM in \figurename~\ref{fig:noise CNM} compares very well with the true model (\figurename~\ref{fig:simdataunif1}), recognizing 15 out of 20 noise observations; as said before, each of the 5 outliers that it does not recognize falls within one of the two clusters (cf.~\figurename~\ref{fig:simdataunif1}). 
For NUMs, the best model according to the BIC has $G=3$ clusters with a VVV covariance structure.
Amongst the detected 7 outliers, there are 6 true outliers and one point, of the cluster on the left, erroneously detected as outlier (compare \figurename~\ref{fig:noise NUM} with \figurename~\ref{fig:simdataunif1}).
Moreover, the third cluster models the part on the right of the background noise.
For NCMs, the best model according to the BIC has the correct number of clusters ($G=2$) but an EEV covariance structure.
Apart from the erroneously identified covariance structure, outliers are detected as for CNMs.

Finally, \tablename~\ref{tab:noise misclassification} reports the number of misclassified observations for each of the best models according to the BIC.   
This number is computed by considering the true classification of the points in: cluster 1, cluster 2, and noise.
We can note how the best performers are the CNM and the NCM, with only 5 misclassified observations corresponding to the 5 noisy points falling into the clusters (compare \figurename~\ref{fig:noise CNM} and \figurename~\ref{fig:noise NCM} with \figurename~\ref{fig:simdataunif1}).
\begin{table}[!ht]
\caption{Simulated data from Section~\ref{subsec:Overall uniform noise}: number of misclassified observations for the best models according to the BIC.}
\label{tab:noise misclassification}
\centering
\begin{tabular}{l ccc}
\toprule
Model & $G$ & Covariance structure & \# of misclassified observations \\
\midrule
NM   & 3 & EVV & 14 \\ 
$t$M & 2 & EEE & 19 \\ 
CNM  & 2 & EEE &  5 \\ 
NUM  & 3 & VVV & 12 \\ 
NCM  & 2 & EEV &  5 \\ 
\bottomrule
\end{tabular}
\end{table}

\subsection{Sensitivity study based on the blue crabs data}
\label{subsec:Sensitivity study}

As a second analysis, a sensitivity study, based on the very popular crabs data set of \citet{Camp:Maho:amul:1974}, is here described to compare how a single bad point affects the behaviour of the competing models.
Attention is focused on the sample of $n=100$ blue crabs of the genus \textit{Leptograpsus}, of which there are 50 males and 50 females (\figurename~\ref{fig:crab}). 
\begin{figure}[!ht]
\centering
  \includegraphics[width=0.35\textwidth]{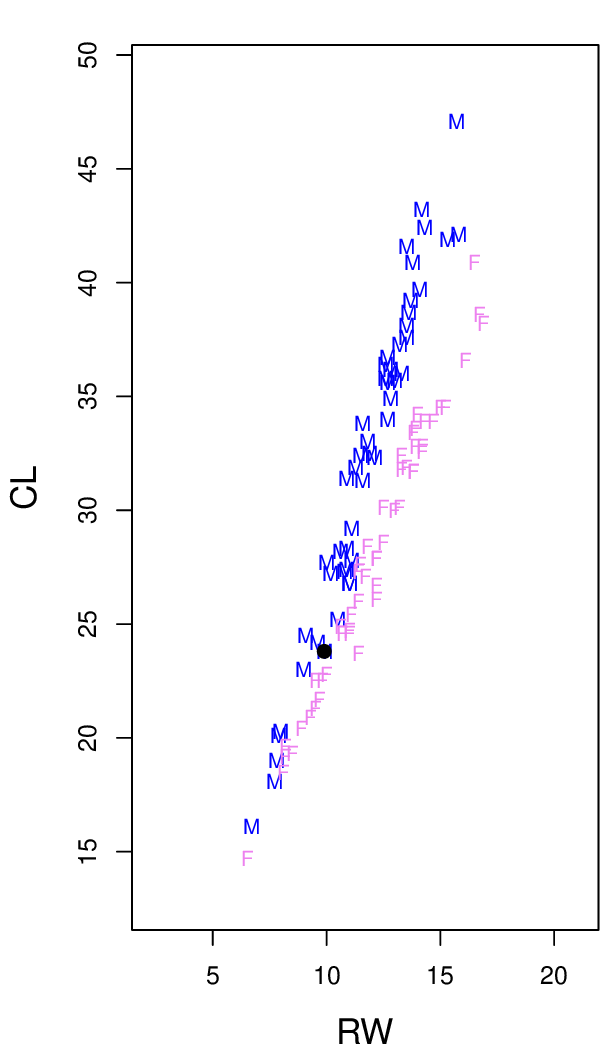} %
\caption{Blue crabs data: Scatterplot (\textsf{F} denotes female and \textsf{M} male; $\bullet$ denotes the observation perturbed for the analysis of Section~\ref{subsec:Sensitivity study}).}
\label{fig:crab}       
\end{figure} 
For each specimen, we consider two measurements (in millimeters), namely the rear width (RW) and the length along the midline of the carapace (CL).
In the fashion of \citet{Peel:McLa:Robu:2000}, thirteen ``perturbed'' data sets are generated by substituting the original value of CL for the 7th point (highlighted by a bullet in \figurename~\ref{fig:crab}) with thirteen anomalous values shown in the first column of \tablename~\ref{tab:crab clustering}.
\setlength{\tabcolsep}{3pt}
\renewcommand{\arraystretch}{1}
\begin{table}[ht]
\caption{
Blue crabs data: Summary information about the fitted VVV models (``\#M'' = number of misallocations, without considering the true outlier; ``weight'' = weight given to the true outlier in the estimation of the parameters; ``d.f'' = estimated degrees of freedom, for the $t$M, in the cluster containing the true outlier; ``\#B'' = number of detected outliers; ``$\widehat{v}_{7g}$'' = probability for the true outlier to be a good point in the cluster $g$ the outlier is assigned; ``$\widehat{\eta}_g$'' = estimated inflation parameter, for the CNM, in the cluster containing the true outlier).
In the column labeled as ``bad'', \ding{51} and \ding{55} indicate if the true outlier is detected or not, respectively, by the model.
}
\label{tab:crab clustering}
\centering
\resizebox{\textwidth}{!}{
\begin{tabular}{r c rr c rrrrcc c rrcrrcc c rrccc c rrccc}
  \toprule
    && \multicolumn{2}{c}{NM} && \multicolumn{6}{c}{$t$M} && \multicolumn{7}{c}{CNM} && \multicolumn{5}{c}{NUM} && \multicolumn{5}{c}{NCM}\\
    \cline{3-4}\cline{6-11}\cline{13-19}\cline{21-25}\cline{27-31}  
 CL 
 && 
 \multicolumn{1}{c}{BIC} & \multicolumn{1}{c}{\#M} 
 && 
 \multicolumn{1}{c}{BIC} & \multicolumn{1}{c}{\#M} & \multicolumn{1}{c}{weight} & \multicolumn{1}{c}{d.f.} & bad & \#B 
 &&
 \multicolumn{1}{c}{BIC} & \multicolumn{1}{c}{\#M} & $\widehat{v}_{7g}$ & \multicolumn{1}{c}{weight} & \multicolumn{1}{c}{$\widehat{\eta}_g$} & bad & \#B 
 &&
 \multicolumn{1}{c}{BIC} & \multicolumn{1}{c}{\#M} & $\widehat{v}_{7g}$ & bad & \#B 
 &&
 \multicolumn{1}{c}{BIC} & \multicolumn{1}{c}{\#M} & $\widehat{v}_{7g}$ & bad & \#B \\ 
  \midrule
-50 && 1071.19 & 20 && 968.53 & 13 & 0.0010 & 2.00 & \ding{51} & 9 && 969.41 & 12 & 0 & 0.0008 & 1284.41 & \ding{51} & 1 && 1005.22 &  14 & 0 & \ding{51} & 2 &  &  &  &  &  \\ 
  -45 && 1066.56 & 20 && 967.99 & 13 & 0.0011 & 2.00 & \ding{51} & 9 && 969.14 & 12 & 0 & 0.0009 & 1119.84 & \ding{51} & 1 && 1005.05 & 13 & 0 & \ding{51} & 2 &  &  &  &  &  \\ 
  -40 && 1061.47 & 21 && 967.41 & 13 & 0.0013 & 2.01 & \ding{51} & 9 && 968.84 & 12 & 0 & 0.0010 & 966.45 & \ding{51} & 1 &&      1004.81  & 13 & 0 & \ding{51} & 2   &  &  &  &  &  \\ 
  -35 && 1055.84 & 20 && 966.78 & 13 & 0.0016 & 2.04 & \ding{51} & 9 && 968.52 & 12 & 0 & 0.0012 & 824.20 & \ding{51} & 1 && 1004.40 & 13 & 0 & \ding{51} & 2 &  &  &  &  &  \\ 
  -30 && 1049.54 & 18 && 966.09 & 13 & 0.0019 & 2.07 & \ding{51} & 9 && 968.18 & 12 & 0 & 0.0014 & 693.11 & \ding{51} & 1 && 1004.32 & 14 & 0 & \ding{51} & 2 &  &  &  &  &  \\ 
  -25 && 1042.48 & 16 && 965.33 & 13 & 0.0023 & 2.10 & \ding{51} & 9 && 967.80 & 12 & 0 & 0.0017 & 573.17 & \ding{51} & 1 && 1003.38 & 34 & 0 & \ding{51} & 1 &  &  &  &  &  \\ 
  -20 && 1034.56 & 16 && 964.49 & 13 & 0.0029 & 2.15 & \ding{51} & 9 && 967.38 & 12 & 0 & 0.0022 & 464.40 & \ding{51} & 1 && 1000.75 & 18 & 0 & \ding{51} & 1 &  &  &  &  &  \\ 
  -15 && 1025.67 & 17 && 963.53 & 13 & 0.0037 & 2.20 & \ding{51} & 9 && 966.90 & 12 & 0 & 0.0027 & 366.78 & \ding{51} & 1 && 999.22 & 14 & 0 & \ding{51} & 1 &  &  &  &  &  \\ 
  -10 && 1015.71 & 15 && 962.44 & 13 & 0.0049 & 2.26 & \ding{51} & 9 && 966.37 & 12 & 0 & 0.0036 & 280.31 & \ding{51} & 1 && 997.33 & 20 & 0 & \ding{51} & 1 &  &  &  &  &  \\ 
  -5 && 1004.75 & 16 && 961.17 & 13 & 0.0068 & 2.34 & \ding{51} & 9 && 965.74 & 12 & 0 & 0.0049 & 204.99 & \ding{51} & 1 && 996.66 & 28 & 0 & \ding{51} & 1 && 948.96 & 13 & 0 & \ding{51} & 2 \\ 
  0 && 992.91 & 16 && 959.66 & 13 & 0.0100 & 2.45 & \ding{51} & 9 && 964.99 & 12 & 0 & 0.0071 & 140.83 & \ding{51} & 1 && 996.25 & 19 & 0 & \ding{51} & 1 && 947.63 & 13 & 0 & \ding{51} & 2 \\ 
  5 && 980.43 & 16 && 957.77 & 13 & 0.0161 & 2.60 & \ding{51} & 8 && 964.04 & 12 & 0 & 0.0114 & 87.77 & \ding{51} & 1 && 996.15 & 38 & 0 & \ding{51} & 1 && 945.77 & 13 & 0 & \ding{51} & 2 \\ 
  10 && 967.90 & 14 && 955.29 & 13 & 0.0297 & 2.85 & \ding{51} & 8 && 962.74 & 12 & 0 & 0.0219 & 45.59 & \ding{51} & 1 && 978.88 & 16 & 0 & \ding{51} & 1 && 942.44 & 14 & 0 & \ding{51} & 4 \\ 
   \bottomrule
\end{tabular}
}
\end{table}

\textit{Ceteris paribus} with \citet{Peel:McLa:Robu:2000}, we directly fit the competing models with $G=2$ clusters and in their VVV version only.
\tablename~\ref{tab:crab clustering} reports some of the obtained results.
Note that the majority of the NCMs have not been fitted (refer to the missing values in \tablename~\ref{tab:crab clustering}) due to computational issues with the adopted \textsf{R} function \texttt{Mclust()}.

We firstly note that, for each approach, the BIC values deteriorate (increase) in line with the departure of the perturbed value from the bulk of the data; this is due to the log-likelihood part of the BIC.
For the perturbed values of CL equal to -5, 0, 5, and 10, the lowest BIC values are obtained for NCMs; for the remaining perturbed values, the lowest BIC values are obtained for the $t$M.
However, these are not the best approaches under other aspects.   
In particular, the CNM is systematically the most robust to the perturbations, with the number of misallocated observations remaining fixed at 12 regardless of the particular value perturbed (refer to the columns labeled as ``\#M'').
This is especially in contrast to the NM, the NUM, and the NCM, where the number of misclassifications changes (and does not necessarily decreases) as the extent of the perturbation increases. 

As concerns the fitted $t$Ms and CNMs, the column labeled as ``weight'' denotes, in correspondence of the bad point and in its MAP cluster of membership, the weight assigned for parameter estimation; this weight is computed according to formula~\eqref{eq:downweight for X} for the CNM, and according to formula~(7.22) in \citet{McLa:Peel:fini:2000} for the $t$M.
As expected, by recalling that the original value of CL for the 7th point was 23.8, these weights decrease as the CL value of the perturbed point further departs from its true value.
A similar reasoning holds for the estimated degrees of freedom, in the cluster containing the bad point, for $t$Ms (refer to the column labeled as ``d.f.'') and for the estimated value of $\eta_g$, in the cluster containing the bad point, for CNMs (refer to the column labeled as ``$\widehat{\eta}_g$'').
In the former case, this means that we need a $t$ distribution with heavier tails as the bad point departs from the bulk of its cluster of membership; in the latter case, $\eta_g$ can be also meant as a sort of ``degree of badness'', i.e., as a measure of how different bad points are from the bulk of their cluster of membership.

In terms of outlier detection for the robust methods, we note that the probability to be a typical point for the bad point (refer to the columns labeled as ``$\widehat{v}_{7g}$'') is practically null for all of the approaches (such that this probability can be computed) regardless of the particular value perturbed.
We can also note how all of the approaches are able to detect the bad point (refer to the columns labeled as ``bad''); however, the CNM is the only model with a null FPR (refer to the columns labeled as ``\#B'' reporting the number of detected outliers); this is especially in contrast to the the detection rule for $t$Ms, based on the 95th percentile, which yields a number of detected bad points of either 8 or, in the majority of the cases, 9.

\subsection{Wine data}
\label{subsec:Wine Data}

The third analysis is based on the wine data set of \citet{Fori:Lear:Arma:Lant:PARV:1998} available in the \texttt{gclus} package \citep{hurley04} for {\sf R}. 
These data comprise $p=13$ physical and chemical properties of $n=178$ wines grown in the same region in Italy but derived from three different cultivars (Barbera, Barolo, Grignolino). 
We treat this as a clustering analysis by ignoring the labels.

The competing families of parsimonious models are fitted for $G\in\left\{1,2,3,4\right\}$.
For each family, the best model, in terms of BIC, is reported in \tablename~\ref{tab:wine best BIC}; the complete list of BIC values is given in \figurename~\ref{fig:wine BIC}.
Some of the BIC values are missing due to computational issues in estimating the corresponding model (see, in particular, \setlength{\tabcolsep}{5pt}
\figurename~\ref{fig:BIC NUM wine}).
\begin{table}[!ht]
\caption{Wine data: for each family of models, best number of clusters ($G$) and parsimonious covariance structure according to the BIC.}
\label{tab:wine best BIC}
\centering
\begin{tabular}{ccc}
	\toprule
Family of models  & $G$ & Parsimonious structure \\ 
	\midrule
NMs   & 3 & VVE \\ 
$t$Ms & 4 & VVI \\ 
CNMs  & 3 & EEE \\ 
NUMs  & 4 & VVI \\ 
NCMs  & 3 & VVE \\ 
	\bottomrule
\end{tabular}
\end{table}\begin{figure}[!ht]
\centering
\subfigure[NMs\label{fig:BIC NM wine}]
{\resizebox{0.322\textwidth}{!}{\includegraphics{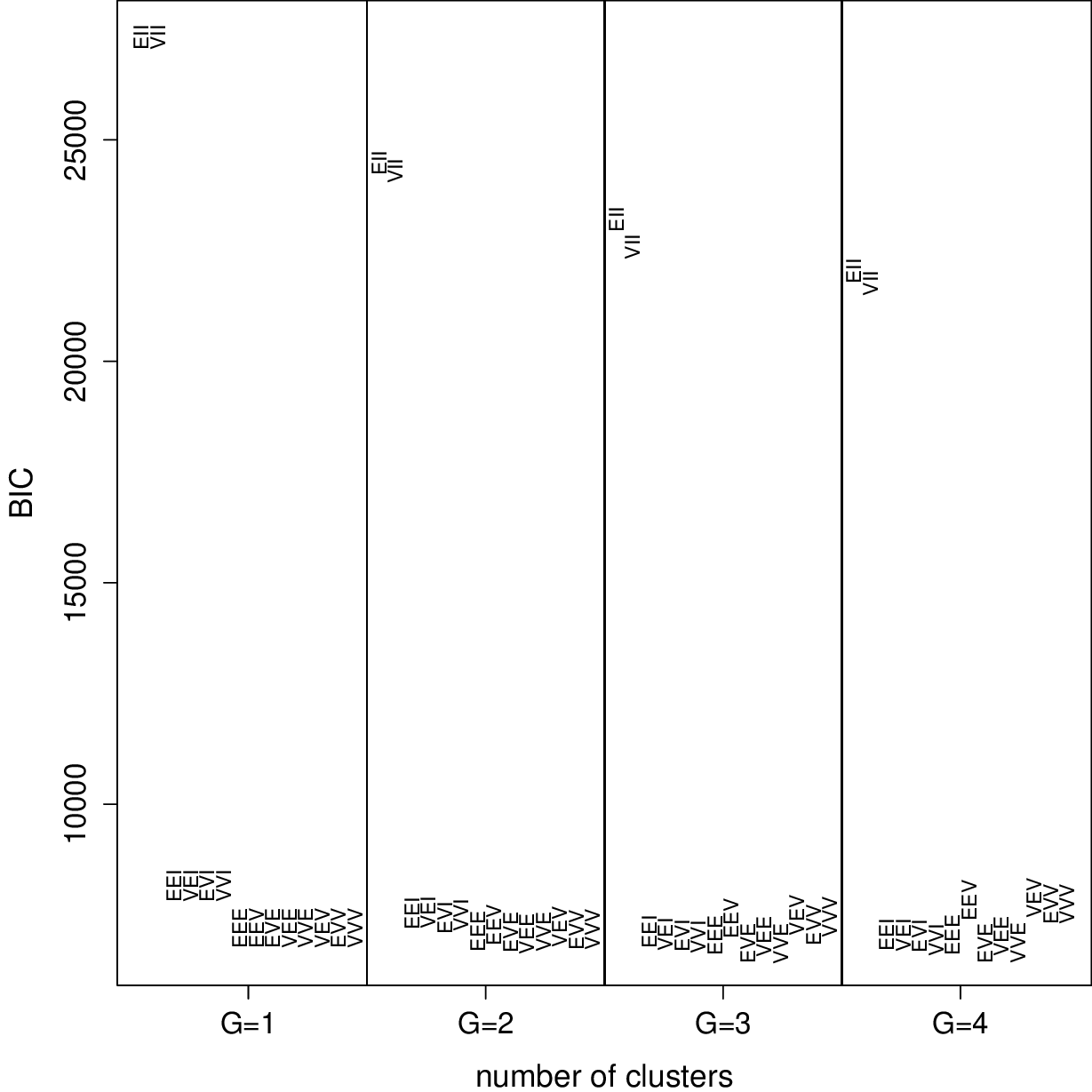}}}
\subfigure[$t$Ms\label{fig:BIC tM wine}]
{\resizebox{0.322\textwidth}{!}{\includegraphics{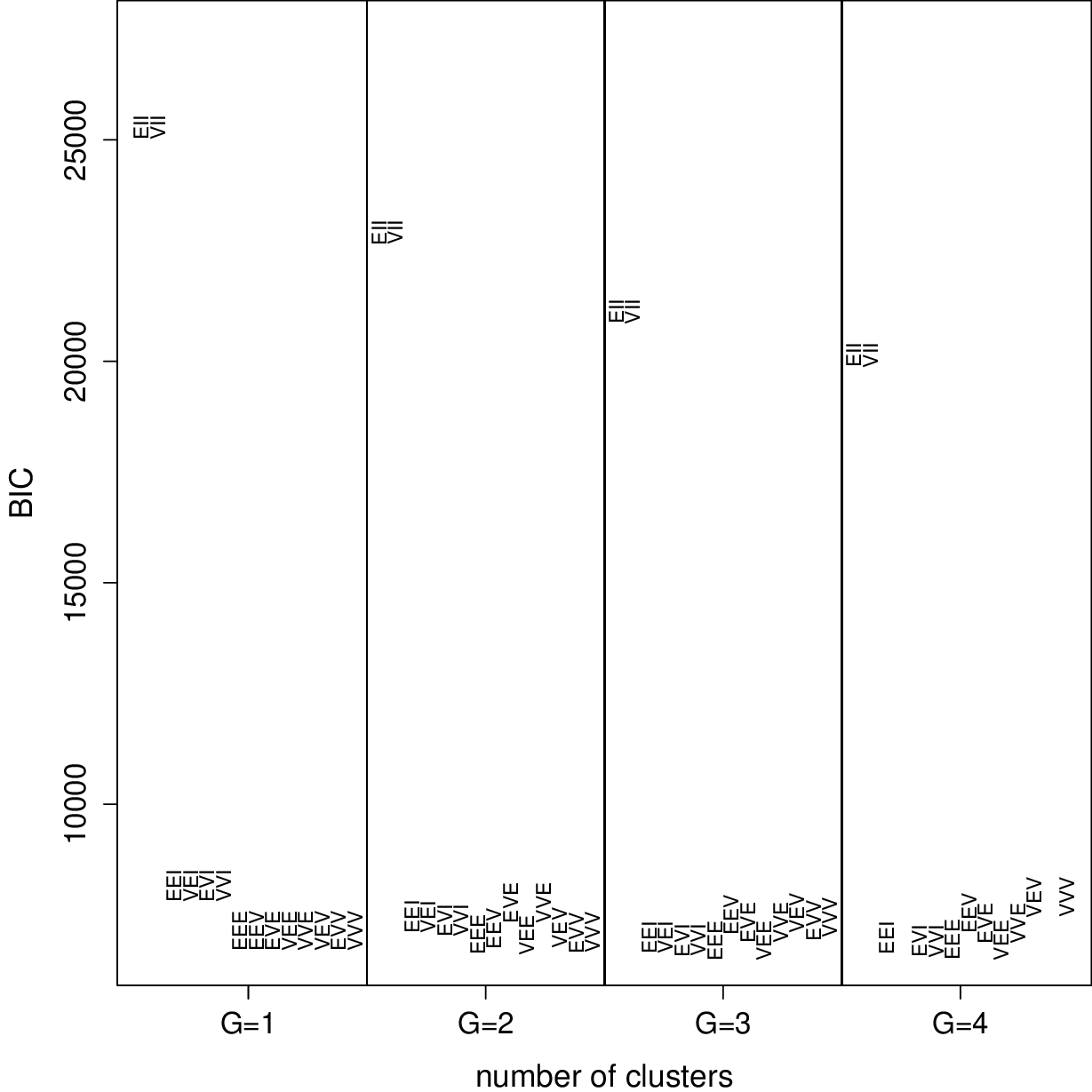}}}
\subfigure[CNMs\label{fig:BIC CNM wine}]
{\resizebox{0.322\textwidth}{!}{\includegraphics{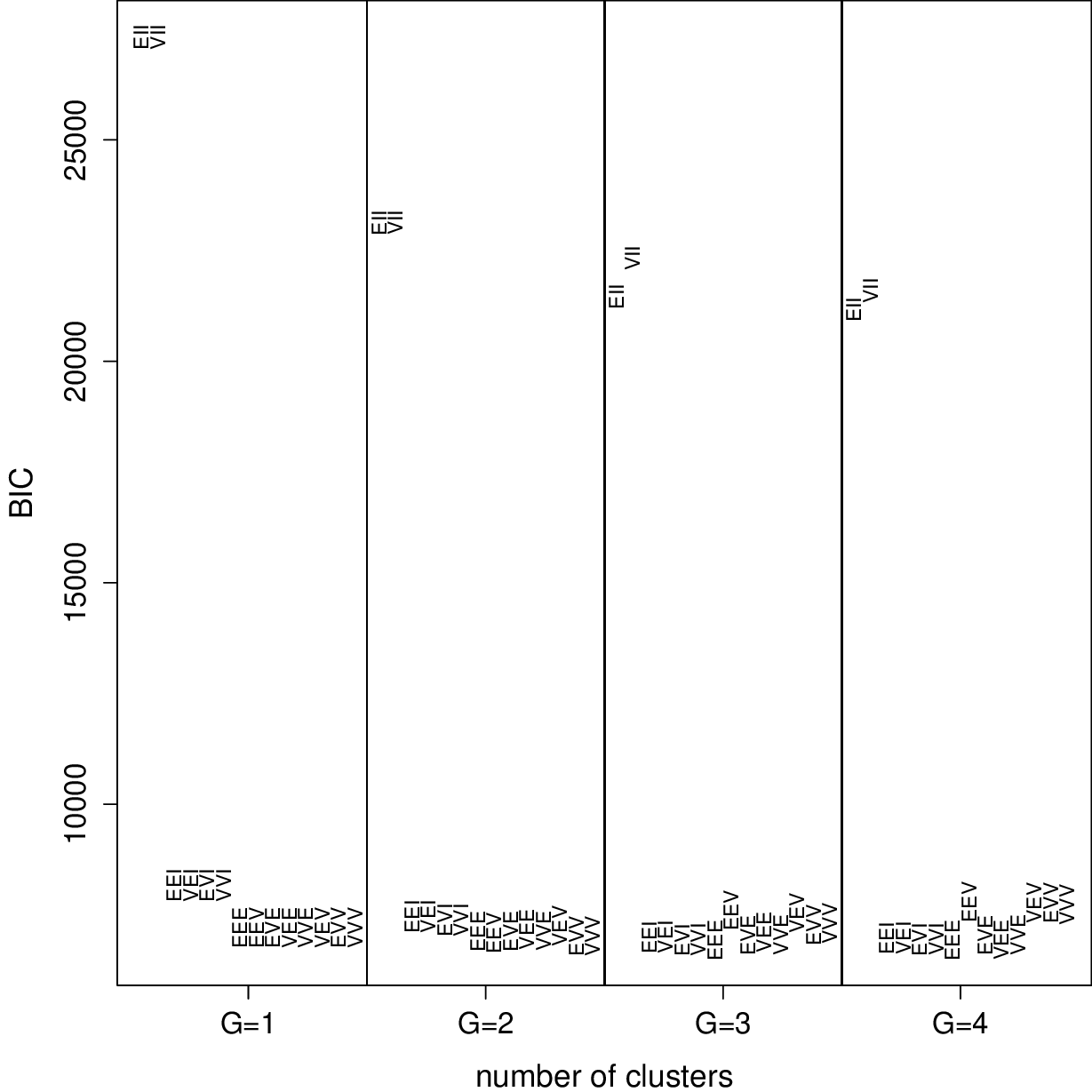}}}
\subfigure[NUMs\label{fig:BIC NUM wine}]
{\resizebox{0.322\textwidth}{!}{\includegraphics{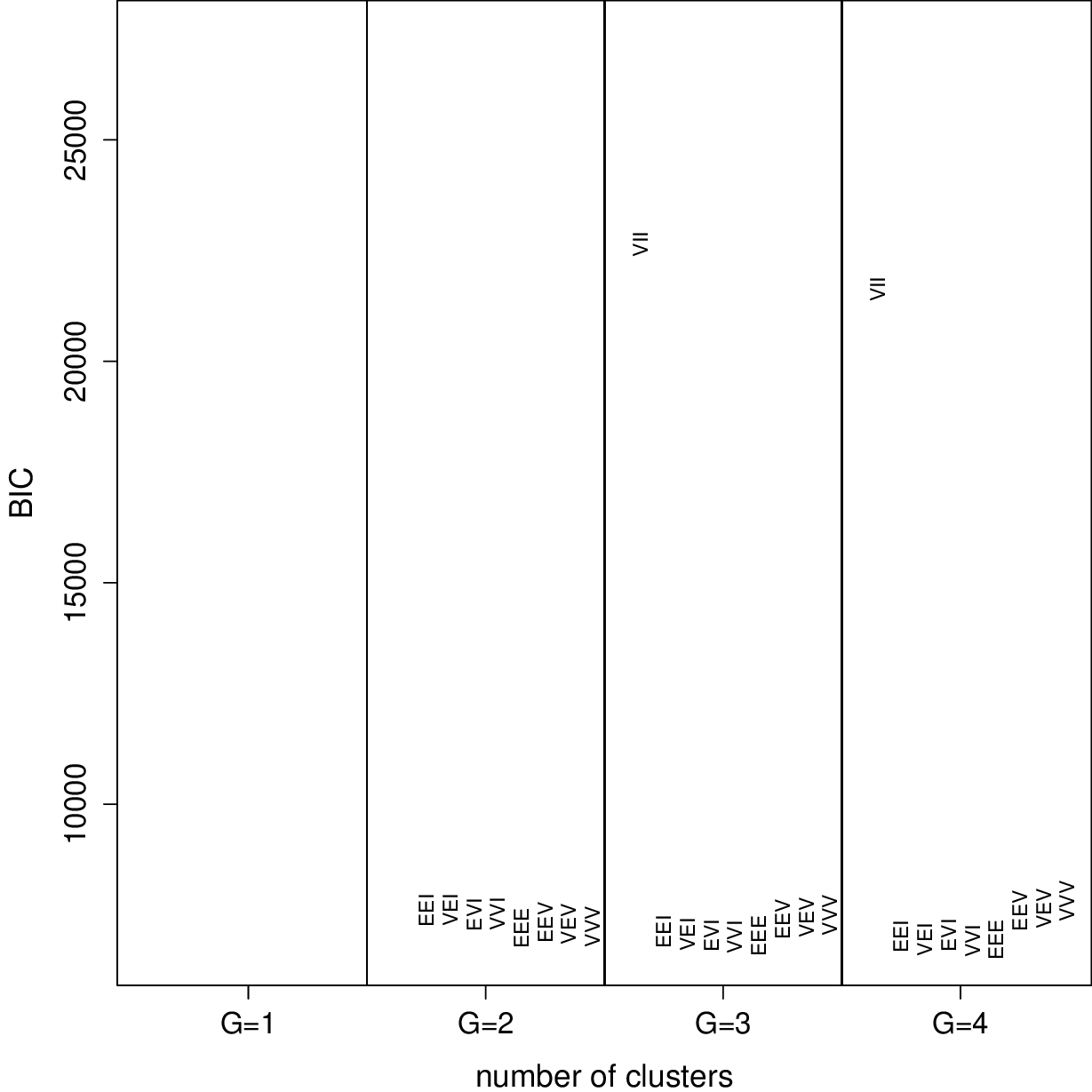}}}
\subfigure[NCMs\label{fig:BIC NCM wine}]
{\resizebox{0.322\textwidth}{!}{\includegraphics{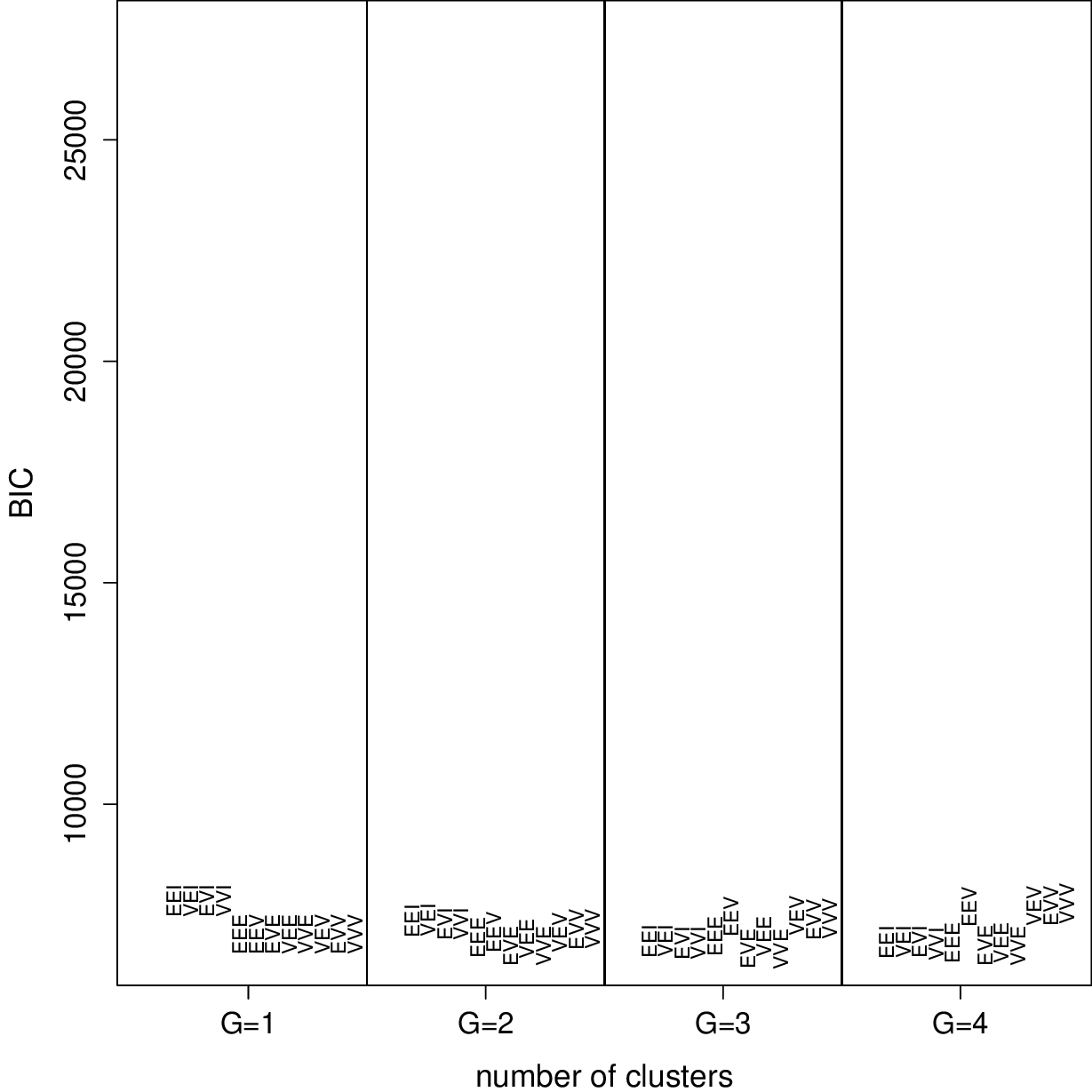}}}
\caption{
Wine data: BIC values for the fitted models.
\label{fig:wine BIC}
}
\end{figure}

The clustering results (\tablename~\ref{tab:wine clustering}) show that the selected NM, CNM, and NCM recognize the presence of three clusters, while the remaining approaches find an additional (fourth) cluster.
\setlength{\tabcolsep}{4.5pt}
\begin{table}[!ht]
\caption{Wine data: Clustering results for each of the competing methods, where good and bad samples are considered together for $t$Ms, CNMs, and NUMs.}
\label{tab:wine clustering}
\centering
\begin{tabular}{l c rrr c rrrr c rrr c rrrr c rrrr}
\toprule
    && \multicolumn{3}{c}{NM} && \multicolumn{4}{c}{$t$M} && \multicolumn{3}{c}{CNM} && \multicolumn{4}{c}{NUM} && \multicolumn{4}{c}{NCM}\\
 \cline{3-5}\cline{7-10}\cline{12-14}\cline{16-19}\cline{21-24}
 Cultivar && 1 & 2 & 3 && 1 & 2 & 3 & 4 && 1 & 2 & 3 && 1 & 2 & 3 & 4 && 1 & 2 & 3 & noise\\  
\midrule
Barbera     && 48 &    &    && 48 &    &    &    && 48 &    &    && 48 &    &    &    && 48 &    &    &    \\ 
Barolo      &&    & 59 &    &&    & 34 &    & 25 &&    & 59 &    &&    & 58 &    &  1 &&    & 58 &    & 1  \\ 
Grignolino  &&  2 &  4 & 65 &&  5 &    & 55 & 11 &&    &    & 71 &&  2 &    & 48 & 21 &&  1 &    & 60 & 10 \\ 
\bottomrule
\end{tabular}
\end{table}
In terms of classification, the Barbera cultivar is classified correctly by all of the models.
Instead, the Barolo and the Grignolino cultivars are classified correctly by the NCM only.
Summarizing, only our approach leads to a perfect clustering when we consider the good points together with the bad points.

In terms of detection of bad points, the last part of \tablename~\ref{tab:wine clustering} already reports the observations assigned to the noise component by the NCM.
No bad points are detected by the NUM, while the observations declared as bad by the $t$M and the CNM are summarized in \tablename~\ref{tab:wine clustering plus noise}.
We see that there are 35 bad points for the $t$M and 26 bad points for the CNM.
\setlength{\tabcolsep}{4.5pt}
\begin{table}[!ht]
\caption{
Clustering results for the PMCGD model on the wine data, where good and bad samples are considered separately.
}
\label{tab:wine clustering plus noise}
\centering
\begin{tabular}{l c rrrrr c rrrr}
\toprule
    && \multicolumn{5}{c}{$t$M} && \multicolumn{4}{c}{CNM}\\ 
    \cline{2-7} \cline{9-12}
     Cultivar  && 1 & 2 & 3 & 4 & Bad && 1 & 2 & 3 & Bad \\  
\midrule
Barbera      && 42 &    &    &    &   6 && 44 &    &    &  4 \\ 
Barolo       &&    & 30 &    & 23 &   6 &&    & 59 &    &    \\ 
Grignolino   &&  1 &    & 45 &  2 &  23 &&    &    & 49 & 22 \\ 
\bottomrule
\end{tabular}
\end{table} 
Considering that Grignolino was the cultivar most difficult to be struggled by the competing methods (cf.~\tablename~\ref{tab:wine clustering}), it is not surprising that the vast majority of bad points detected by the $t$M, the CNM, and the NCM, are in that cultivar.
The graphical representation of the obtained classification for the CNM is shown in \figurename~\ref{fig:wine PMCGD}.   
\begin{figure}[!ht]
\centering
  \includegraphics[width=0.8\textwidth]{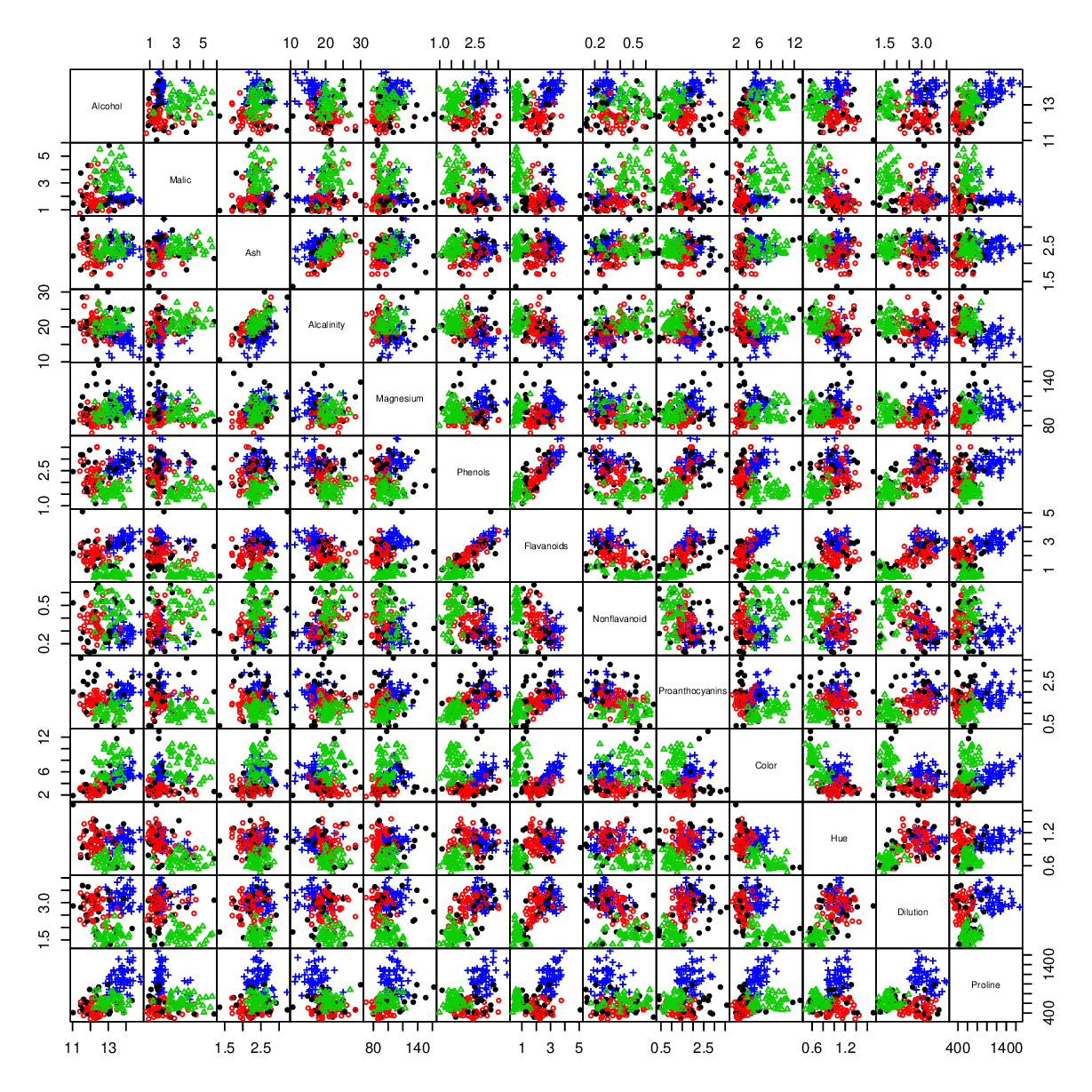} %
\caption{
Wine data: Scatterplot matrix and clustering from the EEE-PMCGD model.
Bad points are denoted by $\bullet$.
}
\label{fig:wine PMCGD}       
\end{figure}

\section{Discussion}
\label{sec:Discussion and future work}

A family of fourteen parsimonious mixtures of contaminated normal distributions has been introduced for clustering.
These models can be viewed as an extension of the famous family of parsimonious mixtures of normal distributions introduced by \citet{Cele:Gova:Gaus:1995}. 
Firstly, as discussed in Section~\ref{subsec:Some notes on robustness} and shown in the simulation study of Section~\ref{sec:Numerical studies}, they facilitate robust estimation of model parameters in the presence of outliers, which we also refer to as bad points: as an example, the estimator of the cluster-specific mean vector in \eqref{eq:update mu} is a weighted mean where the weights allow to reduce the impact of bad points in the estimation. 
Secondly, all of the members of our family of models allow for automatic detection of bad points in the same natural way as observations are typically assigned to the groups in the finite mixture models context, i.e., based on the posterior probabilities of being good or bad points.

Another distinct advantage of our contaminated approach is that we can easily extend the approach to model-based classification \citep[e.g.,][]{McNi:Mode:2010} and model-based discriminant analysis \citep{Hast:Tibs:Disc:1996}. 
In fact, there are a number of options for the type of supervision that could be used in partial classification applications for our models, i.e., one could specify some of the $\{\boldsymbol{z}_i\}_{i=1}^n$ and/or some of the $\{\boldsymbol{v}_i\}_{i=1}^n$ \textit{a~priori}. 
This provides yet more flexibility than exhibited by any competing approach, as does the ability of our approach to work in higher dimensions where bad points cannot easily be visualized.

In all the considered data analyses of Section~\ref{sec:Data analysis}, and also in the simulations of Section~\ref{sec:Numerical studies}, we demonstrated the good behaviour of our contaminated approach when compared to families of parsimonious: mixtures of normal distributions, mixtures of $t$ distributions, mixtures of mixtures of a normal and a uniform distribution, and mixtures of normal distributions plus a uniform component.  

As an open point for further research, it could be interesting to modify our approach with the aim of accommodating asymmetric contamination and/or ``groups'' of concentrated outliers.
In such a case, contamination in the mean (and not in the covariance matrix, like we do) could be considered; see, e.g., the contaminated (location-shift) normal distribution considered by \citet{Verd:Wass:Baye:1991}.

\section*{Acknowledgments}

A. Punzo acknowledges the financial support from the grant ``Finite mixture and latent variable models for causal inference and analysis of socio-economic data'' (FIRB 2012-Futuro in ricerca) funded by the Italian Government (RBFR12SHVV). 
P.D.~McNicholas acknowledges the support of the Canada Research Chairs program.
The authors finally declare the absence of any conflict of interest.

\section*{References}

\bibliographystyle{elsarticle-harv} 
\bibliography{References2}

\begin{thebibliography}{94}
\expandafter\ifx\csname natexlab\endcsname\relax\def\natexlab#1{#1}\fi
\expandafter\ifx\csname url\endcsname\relax
  \def\url#1{\texttt{#1}}\fi
\expandafter\ifx\csname urlprefix\endcsname\relax\def\urlprefix{URL }\fi

\bibitem[{Aggarwal(2013)}]{Agga:Outl:2013}
Aggarwal, C.~C., 2013. Outlier Analysis. Springer New York.

\bibitem[{Aitken(1926)}]{Aitk:OnBe:1926}
Aitken, A., 1926. {On Bernoulli's numerical solution of algebraic equations}.
  In: Proceedings of the Royal Society of Edinburgh. Vol.~46. pp. 289--305.

\bibitem[{Aitkin and Wilson(1980)}]{Aitk:Wils:Mixt:1980}
Aitkin, M., Wilson, G.~T., 1980. Mixture models, outliers, and the {EM}
  algorithm. Technometrics 22~(3), 325--331.

\bibitem[{Andrews and McNicholas(2012)}]{Andr:McNi:Mode:2012}
Andrews, J.~L., McNicholas, P.~D., 2012. Model-based clustering,
  classification, and discriminant analysis with the multivariate
  $t$-distribution: The {$t$EIGEN} family. Statistics and Computing 22~(5),
  1021--1029.

\bibitem[{Andrews et~al.(2015)Andrews, Wickins, Boers, and
  McNicholas}]{Andr:McNi:teig:2015}
Andrews, J.~L., Wickins, J.~R., Boers, N.~M., McNicholas, P.~D., 2015.
  \textsf{teigen}: Model-Based Clustering and Classification with the
  Multivariate $t$ Distribution. Version 2.1.0 (2015-11-20).
\newline\urlprefix\url{http://CRAN.R-project.org/package=teigen}

\bibitem[{Bagnato and Punzo(2013)}]{Bagn:Punz:Fine:2013}
Bagnato, L., Punzo, A., 2013. Finite mixtures of unimodal beta and gamma
  densities and the $k$-bumps algorithm. Computational Statistics 28~(4),
  1571--1597.

\bibitem[{Bai et~al.(2012)Bai, Yao, and Boyer}]{Bai:Yao:Boyer:Robu:2012}
Bai, X., Yao, W., Boyer, J.~E., 2012. Robust fitting of mixture regression
  models. Computational Statistics \& Data Analysis 56~(7), 2347--2359.

\bibitem[{Banfield and Raftery(1993)}]{Banf:Raft:mode:1993}
Banfield, J.~D., Raftery, A.~E., 1993. Model-based {G}aussian and
  non-{G}aussian clustering. Biometrics 49~(3), 803--821.

\bibitem[{Barnett and Lewis(1994)}]{Barn:Lewi:Outl:1994}
Barnett, V., Lewis, T., 1994. Outliers in Statistical Data. Wiley Series in
  Probability \& Statistics. Wiley.

\bibitem[{Becker and Gather(1999)}]{Beck:Gath:Them:1999}
Becker, C., Gather, U., 1999. The masking breakdown point of multivariate
  outlier identification rules. Journal of the American Statistical Association
  94~(447), 947--955.

\bibitem[{Berkane and Bentler(1988)}]{Berk:Bent:Esti:1988}
Berkane, M., Bentler, P.~M., 1988. Estimation of contamination parameters and
  identification of outliers in multivariate data. Sociological Methods \&
  Research 17~(1), 55--64.

\bibitem[{Biernacki(2004)}]{Bier:Anas:2004}
Biernacki, C., 2004. An asymptotic upper bound of the likelihood to prevent
  {G}aussian mixtures from degenerating. Tech. rep., Universit{\'e} de
  Franche-Comt{\'e}, Besan{\c{c}}on.

\bibitem[{Biernacki et~al.(2003)Biernacki, Celeux, and
  Govaert}]{Bier:Cele:Gova:Choo:2003}
Biernacki, C., Celeux, G., Govaert, G., 2003. Choosing starting values for the
  {EM} algorithm for getting the highest likelihood in multivariate {G}aussian
  mixture models. Computational Statistics \& Data Analysis 41~(3-4), 561--575.

\bibitem[{Biernacki et~al.(2008)Biernacki, Celeux, Govaert, Langrognet, Noulin,
  and Vernaz}]{Bier:Cele:Gova:Lang:Noul:Vern:MIXM:2008}
Biernacki, C., Celeux, G., Govaert, G., Langrognet, F., Noulin, G., Vernaz, Y.,
  2008. $\mathsf{MIXMOD}$ - Statistical Documentation. Downloadable from
  \url{http://www.mixmod.org/IMG/pdf/statdoc\_2\_1\_1.pdf}.

\bibitem[{Biernacki and Chr{\'e}tien(2003)}]{Bier:Chre:Stat:2003}
Biernacki, C., Chr{\'e}tien, S., 2003. Degeneracy in the maximum likelihood
  estimation of univariate {G}aussian mixtures with {EM}. Statistics \&
  Probability Letters 61~(4), 373--382.

\bibitem[{Bock(2002)}]{Bock:Clus:2002}
Bock, H.~H., 2002. Clustering methods: From classical models to new approaches.
  Statistics in Transition 5~(5), 725--758.

\bibitem[{B{\"o}hning(2000)}]{Bohn:Comp:2000}
B{\"o}hning, D., 2000. Computer-assisted Analysis of Mixtures and Applications:
  Meta-analysis, Disease Mapping and Others. Vol.~81 of Monographs on
  Statistics and Applied Probability. Chapman \& Hall/CRC, London.

\bibitem[{B{\"o}hning et~al.(1994)B{\"o}hning, Dietz, Schaub, Schlattmann, and
  Lindsay}]{Bohn:Diet:Scha:Schl:Lind:TheD:1994}
B{\"o}hning, D., Dietz, E., Schaub, R., Schlattmann, P., Lindsay, B., 1994. The
  distribution of the likelihood ratio for mixtures of densities from the
  one-parameter exponential family. Annals of the Institute of Statistical
  Mathematics 46~(2), 373--388.

\bibitem[{B{\"o}hning and Ruangroj(2002)}]{Bohn:Ruan:Anot:2002}
B{\"o}hning, D., Ruangroj, R., 2002. A note on the maximum deviation of the
  scale-contaminated normal to the best normal distribution. Metrika 55~(3),
  177--182.

\bibitem[{Browne and McNicholas(2014)}]{Brow:McNi:Adva:2014}
Browne, R.~P., McNicholas, P.~D., 2014. Estimating common principal components
  in high dimensions. Advances in Data Analysis and Classification 8~(2),
  217--226.

\bibitem[{Browne and McNicholas(2015)}]{Brow:McNi:mixt:2015}
Browne, R.~P., McNicholas, P.~D., 2015. \textsf{mixture}: Mixture Models for
  Clustering and Classification. Version 1.4 (2015-03-10).
\newline\urlprefix\url{http://CRAN.R-project.org/package=mixture}

\bibitem[{Browne et~al.(2012)Browne, McNicholas, and
  Sparling}]{Brow:McNi:Spar:Mode:2012}
Browne, R.~P., McNicholas, P.~D., Sparling, M.~D., 2012. Model-based learning
  using a mixture of mixtures of {G}aussian and uniform distributions. IEEE
  Transactions on Pattern Analysis and Machine Intelligence 34~(4), 814--817.

\bibitem[{Browne et~al.(2013)Browne, Subedi, and
  McNicholas}]{Brow:Sube:McNi:Cons:2013}
Browne, R.~P., Subedi, S., McNicholas, P.~D., 2013. Constrained optimization
  for a subset of the {G}aussian parsimonious clustering models. arXiv.org
  e-print 1306.5824, available at: \url{http://arxiv.org/abs/1306.5824}.

\bibitem[{Byers and Raftery(1998)}]{Byer:Raft:Near:1998}
Byers, S., Raftery, A.~E., 1998. Nearest-neighbor clutter removal for
  estimating features in spatial point processes. Journal of the American
  Statistical Association 93~(442), 577--584.

\bibitem[{Campbell(1984)}]{Camp:Mixt:1984}
Campbell, N.~A., 1984. Mixture models and atypical values. Mathematical Geology
  16~(5), 465--477.

\bibitem[{Campbell and Mahon(1974)}]{Camp:Maho:amul:1974}
Campbell, N.~A., Mahon, R.~J., 1974. A multivariate study of variation in two
  species of rock crab of genus {L}eptograpsus. Australian Journal of Zoology
  22~(3), 417--425.

\bibitem[{Celeux and Govaert(1995)}]{Cele:Gova:Gaus:1995}
Celeux, G., Govaert, G., 1995. Gaussian parsimonious clustering models. Pattern
  Recognition 28~(5), 781--793.

\bibitem[{Celeux et~al.(2000)Celeux, Hurn, and
  Robert}]{Cele:Hurn:Robe:Comp:2000}
Celeux, G., Hurn, M., Robert, C.~P., 2000. Computational and inferential
  difficulties with mixture posterior distributions. Journal of the American
  Statistical Association 95~(451), 957--970.

\bibitem[{Coretto and Hennig(2011)}]{Core:Henn:2011}
Coretto, P., Hennig, C., 2011. Maximum likelihood estimation of heterogeneous
  mixtures of {G}aussian and uniform distributions. Journal of Statistical
  Planning and Inference 141~(1), 462--473.

\bibitem[{Coretto and Hennig(2015)}]{Core:Henn:Robu:2015}
Coretto, P., Hennig, C., 2015. Robust improper maximum likelihood: tuning,
  computation, and a comparison with other methods for robust {G}aussian
  clustering. arXiv.org e-print 1406.0808, available at:
  \url{http://arxiv.org/abs/1406.0808}.

\bibitem[{Crawford(1994)}]{Craw:Anap:1994}
Crawford, S.~L., 1994. An application of the laplace method to finite mixture
  distributions. Journal of the American Statistical Association 89~(425),
  259--267.

\bibitem[{Cuesta-Albertos et~al.(1997)Cuesta-Albertos, Gordaliza, and
  Matr{\'a}n}]{Cues:Gord:Matr:Trim:1997}
Cuesta-Albertos, J.~A., Gordaliza, A., Matr{\'a}n, C., 1997. Trimmed $k$-means:
  An attempt to robustify quantizers. The Annals of Statistics 25~(2),
  553--576.

\bibitem[{Davies and Gather(1993)}]{Davi:Gath:Thei:1993}
Davies, L., Gather, U., 1993. The identification of multiple outliers. Journal
  of the American Statistical Association 88~(423), 782--792.

\bibitem[{De~Veaux and Krieger(1990)}]{deVa:Krie:1990}
De~Veaux, R.~D., Krieger, A.~M., 1990. Robust estimation of a normal mixture.
  Statistics \& Probability Letters 10~(1), 1--7.

\bibitem[{Dempster et~al.(1977)Dempster, Laird, and
  Rubin}]{Demp:Lair:Rubi:Maxi:1977}
Dempster, A.~P., Laird, N.~M., Rubin, D.~B., 1977. {Maximum likelihood from
  incomplete data via the EM algorithm}. Journal of the Royal Statistical
  Society: Series B (Statistical Methodology) 39~(1), 1--38.

\bibitem[{Di~Zio et~al.(2007)Di~Zio, Guarnera, and
  Rocci}]{DiZi:Guar:Rocc:Amix:2007}
Di~Zio, M., Guarnera, U., Rocci, R., 2007. A mixture of mixture models for a
  classification problem: The unity measure error. Computational Statistics \&
  Data analysis 51~(5), 2573--2585.

\bibitem[{Flury and Gautschi(1986)}]{Flur:Gaut:anal:1986}
Flury, B.~N., Gautschi, W., 1986. An algorithm for simultaneous orthogonal
  transformation of several positive definite matrices to nearly diagonal form.
  SIAM Journal on Scientific and Statistical Computing 7~(1), 169--184.

\bibitem[{Forina et~al.(1998)Forina, Leardi, Armanino, and
  Lanteri}]{Fori:Lear:Arma:Lant:PARV:1998}
Forina, M., Leardi, R., Armanino, C., Lanteri, S., 1998. \textsf{PARVUS}: An
  extendible package for data exploration, classification and correlation.
  Tech. rep., Institute of Pharmaceutical and Food Analysis and Technologies,
  Genoa, Italy.

\bibitem[{Fraley and Raftery(1998)}]{Fral:Raft:Howm:1998}
Fraley, C., Raftery, A.~E., 1998. How many clusters? {W}hich clustering method?
  {A}nswers via model-based cluster analysis. Computer Journal 41~(8),
  578--588.

\bibitem[{Fraley et~al.(2012)Fraley, Raftery, Murphy, and
  Scrucca}]{Fral:Raft:Murp:Scru:mclu:2012}
Fraley, C., Raftery, A.~E., Murphy, T.~B., Scrucca, L., 2012. \textsf{mclust}
  version 4 for \textsf{R}: Normal mixture modeling for model-based clustering,
  classification, and density estimation. Technical report 597, Department of
  Statistics, University of Washington, Seattle, Washington, USA.

\bibitem[{Fraley et~al.(2015)Fraley, Raftery, Scrucca, Murphy, and
  Fop}]{Fral:Raft:Scru:Murp:Fop:mclu:2015}
Fraley, C., Raftery, A.~E., Scrucca, L., Murphy, T.~B., Fop, M., 2015.
  \textsf{mclust}: Normal Mixture Modelling for Model-Based Clustering,
  Classification, and Density Estimation. Version 5.1 (2015-10-27).
\newline\urlprefix\url{http://CRAN.R-project.org/package=mclust}

\bibitem[{Gallegos and Ritter(2005)}]{Gall:Ritt:Arob:2005}
Gallegos, M.~T., Ritter, G., 2005. A robust method for cluster analysis. The
  Annals of Statistics 33~(1), 347--380.

\bibitem[{Gallegos and Ritter(2009)}]{Gall:Ritt:Trim:2009}
Gallegos, M.~T., Ritter, G., 2009. Trimmed {ML} estimation of contaminated
  mixtures. Sankhy{\=a}: The Indian Journal of Statistics, Series A 71~(2),
  164--220.

\bibitem[{Garc{\'\i}a-Escudero and Gordaliza(1999)}]{Garc:Gord:Robu:1999}
Garc{\'\i}a-Escudero, L.~A., Gordaliza, A., 1999. Robustness properties of $k$
  means and trimmed $k$ means. Journal of the American Statistical Association
  94~(447), 956--969.

\bibitem[{Garc{\'\i}a-Escudero et~al.(2003)Garc{\'\i}a-Escudero, Gordaliza, and
  Matr{\'a}n}]{Garc:Gord:Matr:Trim:2003}
Garc{\'\i}a-Escudero, L.~A., Gordaliza, A., Matr{\'a}n, C., 2003. Trimming
  tools in exploratory data analysis. Journal of Computational and Graphical
  Statistics 12~(2), 434--449.

\bibitem[{Garc{\'\i}a-Escudero et~al.(2008)Garc{\'\i}a-Escudero, Gordaliza,
  Matr{\'a}n, and Mayo-Iscar}]{Garc:Gord:Matr:Mayo:Agen:2008}
Garc{\'\i}a-Escudero, L.~A., Gordaliza, A., Matr{\'a}n, C., Mayo-Iscar, A.,
  2008. A general trimming approach to robust cluster analysis. The Annals of
  Statistics 36~(3), 1324--1345.

\bibitem[{Garc{\'\i}a-Escudero et~al.(2010)Garc{\'\i}a-Escudero, Gordaliza,
  Matr{\'a}n, and Mayo-Iscar}]{Garc:Gord:Matr:Mayo:Arew:2010}
Garc{\'\i}a-Escudero, L.~A., Gordaliza, A., Matr{\'a}n, C., Mayo-Iscar, A.,
  2010. A review of robust clustering methods. Advances in Data Analysis and
  Classification 4~(2), 89--109.

\bibitem[{Gerogiannis et~al.(2009)Gerogiannis, Nikou, and
  Likas}]{Gero:Niko:Lika:Them:2009}
Gerogiannis, D., Nikou, C., Likas, A., 2009. The mixtures of {S}tudent's
  t-distributions as a robust framework for rigid registration. Image and
  Vision Computing 27~(9), 1285--1294.

\bibitem[{Hartigan(1985)}]{Hart:Stat:1985}
Hartigan, J.~A., 1985. Statistical theory in clustering. Journal of
  classification 2~(1), 63--76.

\bibitem[{Hastie and Tibshirani(1996)}]{Hast:Tibs:Disc:1996}
Hastie, T., Tibshirani, R., 1996. Discriminant analysis by {G}aussian mixtures.
  Journal of the Royal Statistical Society: Series~B 58~(1), 155--176.

\bibitem[{Hathaway(1986)}]{Hath:Acons:1986}
Hathaway, R.~J., 1986. A constrained {EM} algorithm for univariate normal
  mixtures. Journal of Statistical Computation and Simulation 23~(3), 211--230.

\bibitem[{Hawkins(2013)}]{Hawk:Iden:2013}
Hawkins, D., 2013. Identification of Outliers. Monographs on Statistics and
  Applied Probability. Springer, The Netherlands.

\bibitem[{Hennig(2002)}]{Henn:Fixe:2002}
Hennig, C., 2002. Fixed point clusters for linear regression: computation and
  comparison. Journal of Classification 19~(2), 249--276.

\bibitem[{Hennig(2004)}]{Henn:Brea:2004}
Hennig, C., 2004. Breakdown points for maximum likelihood estimators of
  location-scale mixtures. The Annals of Statistics 32~(4), 1313--1340.

\bibitem[{Hennig and Hausdorf(2015)}]{Henn:Bern:prab:2015}
Hennig, C., Hausdorf, B., 2015. \textsf{prabclus}: Functions for Clustering of
  Presence-Absence, Abundance and Multilocus Genetic Data. Version 2.2-6
  (2015-01-14).
\newline\urlprefix\url{http://CRAN.R-project.org/package=prabclus}

\bibitem[{Holzmann et~al.(2006)Holzmann, Munk, and
  Gneiting}]{Holz:Munk:Gnei:Iden:2006}
Holzmann, H., Munk, A., Gneiting, T., 2006. Identifiability of finite mixtures
  of elliptical distributions. Scandinavian Journal of Statistics 33~(4),
  753--763.

\bibitem[{Hunter and Lange(2000)}]{Lang:Hunt:Yang:Opti:2000}
Hunter, D.~R., Lange, K., 2000. Rejoinder to discussion of ``optimization
  transfer using surrogate objective functions''. Journal of Computational and
  Graphical Statistics 9~(1), 52--59.

\bibitem[{Hurley(2004)}]{hurley04}
Hurley, C., 2004. Clustering visualizations of multivariate data. Journal of
  Computational and Graphical Statistics 13~(4), 788--806.

\bibitem[{Ingrassia(2004)}]{Ingr:Alik:2004}
Ingrassia, S., 2004. A likelihood-based constrained algorithm for multivariate
  normal mixture models. Statistical Methods and Applications 13~(2), 151--166.

\bibitem[{Ingrassia and Rocci(2007)}]{Ingr:Rocc:Cons:2007}
Ingrassia, S., Rocci, R., 2007. Constrained monotone em algorithms for finite
  mixture of multivariate {G}aussians. Computational Statistics \& Data
  Analysis 51~(11), 5339--5351.

\bibitem[{Ingrassia and Rocci(2011)}]{Ingr:Rocc:Dege:2011}
Ingrassia, S., Rocci, R., 2011. Degeneracy of the {EM} algorithm for the mle of
  multivariate {G}aussian mixtures and dynamic constraints. Computational
  Statistics \& Data Analysis 55~(4), 1715--1725.

\bibitem[{Karlis and Xekalaki(2003)}]{Karl:Xeka:Choo:2003}
Karlis, D., Xekalaki, E., 2003. {Choosing initial values for the EM algorithm
  for finite mixtures}. Computational Statistics \& Data Analysis 41~(3--4),
  577--590.

\bibitem[{Lebret et~al.(2012)Lebret, Iovleff, Langrognet, Biernacki, Celeux,
  and Govaert}]{Lebr:Iovl:Lang:Bier:Cele:Gova:Rmix:2012}
Lebret, R., Iovleff, S., Langrognet, F., Biernacki, C., Celeux, G., Govaert,
  G., 2012. {Rmixmod}: The {\sf R} Package of the Model-Based Unsupervised,
  Supervised and Semi-Supervised Classification Mixmod Library.

\bibitem[{Li(2005)}]{Li:Clus:2005}
Li, J., 2005. Clustering based on a multi-layer mixture model. Journal of
  Computational and Graphical Statistics 14~(3), 547--568.

\bibitem[{Little(1988)}]{Litt:Robu:1988}
Little, R. J.~A., 1988. Robust estimation of the mean and covariance matrix
  from data with missing values. Applied Statistics 37~(1), 23--38.

\bibitem[{Lo(2005)}]{Lo:Like:2005}
Lo, Y., 2005. Likelihood ratio tests of the number of components in a normal
  mixture with unequal variances. Statistics \& Probability Letters 71~(3),
  225--235.

\bibitem[{Lo(2008)}]{Lo:Alik:2008}
Lo, Y., 2008. A likelihood ratio test of a homoscedastic normal mixture against
  a heteroscedastic normal mixture. Statistics and Computing 18~(3), 233--240.

\bibitem[{Lo et~al.(2001)Lo, Mendell, and Rubin}]{Lo:Mend:Rubi:Test:2001}
Lo, Y., Mendell, N.~R., Rubin, D.~B., 2001. Testing the number of components in
  a normal mixture. Biometrika 88~(3), 767--778.

\bibitem[{Markatou(2000)}]{Mark:Mixt:2000}
Markatou, M., 2000. Mixture models, robustness, and the weighted likelihood
  methodology. Biometrics 56~(2), 483--486.

\bibitem[{McLachlan and Krishnan(2007)}]{McLa:Kris:TheE:2007}
McLachlan, G., Krishnan, T., 2007. The {EM} algorithm and extensions, 2nd
  Edition. Vol. 382 of Wiley Series in Probability and Statistics. John Wiley
  \& Sons, New York.

\bibitem[{McLachlan and Basford(1988)}]{McLa:Basf:mixt:1988}
McLachlan, G.~J., Basford, K.~E., 1988. Mixture Models: Inference and
  Applications to Clustering. Marcel Dekker, New York.

\bibitem[{McLachlan and Peel(1998)}]{McLa:Peel:Robu:1998}
McLachlan, G.~J., Peel, D., 1998. Robust cluster analysis via mixtures of
  multivariate $t$-distributions. In: Amin, A., Dori, D., Pudil, P., Freeman,
  H. (Eds.), Advances in Pattern Recognition. Vol. 1451 of Lecture Notes in
  Computer Science. Springer, Berlin-Heidelberg, pp. 658--666.

\bibitem[{McLachlan and Peel(2000)}]{McLa:Peel:fini:2000}
McLachlan, G.~J., Peel, D., 2000. Finite Mixture Models. John Wiley \& Sons,
  New York.

\bibitem[{McNicholas(2010)}]{McNi:Mode:2010}
McNicholas, P.~D., 2010. Model-based classification using latent {G}aussian
  mixture models. Journal of Statistical Planning and Inference 140~(5),
  1175--1181.

\bibitem[{McNicholas(2016)}]{mcnicholas16}
McNicholas, P.~D., 2016. Mixture Model-Based Classification. Chapman \&
  Hall/CRC Press, Boca Raton.

\bibitem[{McNicholas et~al.(2010)McNicholas, Murphy, McDaid, and
  Frost}]{McNi:Murp:McDa:Fros:Seri:2010}
McNicholas, P.~D., Murphy, T.~B., McDaid, A.~F., Frost, D., 2010. Serial and
  parallel implementations of model-based clustering via parsimonious
  {G}aussian mixture models. Computational Statistics \& Data Analysis 54~(3),
  711--723.

\bibitem[{Meng and Rubin(1993)}]{Meng:Rubin:Maxi:1993}
Meng, X.-L., Rubin, D.~B., 1993. Maximum likelihood estimation via the {ECM}
  algorithm: A general framework. Biometrika 80~(2), 267--278.

\bibitem[{Peel and McLachlan(2000)}]{Peel:McLa:Robu:2000}
Peel, D., McLachlan, G.~J., 2000. Robust mixture modelling using the $t$
  distribution. Statistics and Computing 10~(4), 339--348.

\bibitem[{Punzo et~al.(2016)Punzo, Browne, and
  McNicholas}]{Punz:Brow:McNi:Hypo:2016}
Punzo, A., Browne, R.~P., McNicholas, P.~D., 2016. Hypothesis testing for
  mixture model selection. Journal of Statistical Computation and Simulation.
  To appear, {\sc doi}:~10.1080/00949655.2015.1131282.

\bibitem[{Punzo et~al.(2015)Punzo, Mazza, and
  McNicholas}]{Punz:Mazz:McNi:Cont:2015}
Punzo, A., Mazza, A., McNicholas, P.~D., 2015. \textsf{ContaminatedMixt}:
  Model-Based Clustering and Classification with the Multivariate Contaminated
  Normal Distribution. Version 1.0 (2015-12-20).
\newline\urlprefix\url{http://CRAN.R-project.org/package=ContaminatedMixt}

\bibitem[{{R Core Team}(2015)}]{R}
{R Core Team}, 2015. {\sf R}: A Language and Environment for Statistical
  Computing. {\sf R} Foundation for Statistical Computing, Vienna, Austria.
\newline\urlprefix\url{http://www.R-project.org/}

\bibitem[{Raftery(1995)}]{Raft:Baye:1995}
Raftery, A.~E., 1995. Bayesian model selection in social research. Sociological
  Methodology 25, 111--164.

\bibitem[{Ritter(2015)}]{Ritt:Robu:2015}
Ritter, G., 2015. Robust Cluster Analysis and Variable Selection. Vol. 137 of
  Chapman \& Hall/CRC Monographs on Statistics \& Applied Probability. CRC
  Press.

\bibitem[{Ruwet et~al.(2012)Ruwet, Garc{\'\i}a-Escudero, Gordaliza, and
  Mayo-Iscar}]{Ruwe:Garc:Gord:Mayo:TheI:2012}
Ruwet, C., Garc{\'\i}a-Escudero, L.~A., Gordaliza, A., Mayo-Iscar, A., 2012.
  The influence function of the tclust robust clustering procedure. Advances in
  Data Analysis and Classification 6~(2), 107--130.

\bibitem[{Ruwet et~al.(2013)Ruwet, Garc{\'\i}a-Escudero, Gordaliza, and
  Mayo-Iscar}]{Ruwe:Garc:Gord:Mayo:Onth:2013}
Ruwet, C., Garc{\'\i}a-Escudero, L.~A., Gordaliza, A., Mayo-Iscar, A., 2013. On
  the breakdown behavior of the tclust clustering procedure. Test 22~(3),
  466--487.

\bibitem[{Schwarz(1978)}]{Schw:Esti:1978}
Schwarz, G., 1978. Estimating the dimension of a model. The Annals of
  Statistics 6~(2), 461--464.

\bibitem[{Stephens(2000)}]{Step:Deal:2000}
Stephens, M., 2000. Dealing with label switching in mixture models. Journal of
  the Royal Statistical Society. Series B: Statistical Methodology 62~(4),
  795--809.

\bibitem[{Teicher(1963)}]{Teic:Iden:1963}
Teicher, H., 1963. Identifiability of finite mixtures. Annals of Mathematical
  Statistics 34~(4), 1265--1269.

\bibitem[{Tukey(1960)}]{Tuke:Asur:1960}
Tukey, J.~W., 1960. A survey of sampling from contaminated distributions. In:
  Olkin, I. (Ed.), Contributions to Probability and Statistics: Essays in Honor
  of Harold Hotelling. Stanford Studies in Mathematics and Statistics. Stanford
  University Press, California, Ch.~39, pp. 448--485.

\bibitem[{Verdinelli and Wasserman(1991)}]{Verd:Wass:Baye:1991}
Verdinelli, I., Wasserman, L., 1991. Bayesian analysis of outlier problems
  using the {G}ibbs sampler. Statistics and Computing 1~(2), 105--117.

\bibitem[{Wolfe(1965)}]{wolfe65}
Wolfe, J.~H., 1965. A computer program for the maximum likelihood analysis of
  types. Technical Bulletin 65-15, U.S.\ Naval Personnel Research Activity.

\bibitem[{Yakowitz and Spragins(1968)}]{Yako:Spra:Onth:1968}
Yakowitz, S.~J., Spragins, J.~D., 1968. On the identifiability of finite
  mixtures. The Annals of Mathematical Statistics 39~(1), 209--214.

\bibitem[{Yao(2012)}]{Yao:Mode:2012}
Yao, W., 2012. Model based labeling for mixture models. Statistics and
  Computing 22~(2), 337--347.

\bibitem[{Yao et~al.(2014)Yao, Wei, and Yu}]{Yao:Wei:Yu:Robu:2014}
Yao, W., Wei, Y., Yu, C., 2014. Robust mixture regression using the
  $t$-distribution. Computational Statistics \& Data Analysis 71, 116--127.

\end{thebibliography}

\end{document}